\documentclass[ a4paper, 11pt]{article} % for arXiv
\pdfoutput=1
\usepackage{amsmath,amsfonts,amscd,amssymb, amsthm,mathrsfs, ascmac}
\usepackage{latexsym}
\usepackage{longtable,geometry}
\usepackage[english]{babel}
\usepackage[utf8]{inputenc}
\usepackage{graphicx}
\usepackage{bbm}
\usepackage{enumitem}
\usepackage{stmaryrd}
\usepackage{nicefrac}
\usepackage{bm}
\usepackage[svgnames,psnames]{xcolor}
\usepackage[colorlinks,citecolor=DarkGreen,linkcolor=FireBrick,linktocpage,unicode]{hyperref}
\usepackage{braket}
\usepackage{xparse}
\usepackage{tikz-cd}
\usetikzlibrary{positioning}

\hypersetup{  urlcolor=blue}
\usepackage{comment}
\usepackage{authblk}
\newcommand{\be}{\begin{equation}}
\newcommand{\ee}{\end{equation}}

\newcommand{\ba}{\begin{align}}
\newcommand{\ea}{\end{align}}

\newcommand{\ben}{\begin{equation*}}
\newcommand{\een}{\end{equation*}}
\makeatletter
  
  \@addtoreset{equation}{section}
\makeatother

\title{\bfseries Ground state of one-dimensional fermion-phonon systems }
\date{\empty}

\author[1]{
Tadahiro Miyao 
}
\author[1]{
Seigo Okida
}
\author[1]{
Hayato Tominaga}

\affil[1]{Department of Mathematics, Hokkaido University

Sapporo 060-0810, Japan}

\newcommand{\one}{1}
\newcommand{\h}{\mathfrak{H}}

\newcommand{\D}{\mathrm{dom}}

\newcommand{\im}{\mathrm{i}}
\newcommand{\Fock}{\mathfrak{F}}

\newcommand{\AFock}{\mathfrak{F}^{\mathrm{F}}}
\newcommand{\BFock}{\mathfrak{F}^{\mathrm{B}}}
\newcommand{\la}{\langle}
\newcommand{\ra}{\rangle}
\newcommand{\Tr}{\mathrm{Tr}}

\newcommand{\BbbR}{\mathbb{R}}
\newcommand{\BbbN}{\mathbb{N}}
\newcommand{\BbbZ}{\mathbb{Z}}

\newcommand{\BbbC}{\mathbb{C}}
\newcommand{\vepsilon}{\varepsilon}
\newcommand{\vphi}{\varphi}

\newcommand{\THL}{\tilde{\h}_{\Lambda}}
\newcommand{\THLd}{\tilde{\h}_{\Lambda\rq{}}}

\newcommand{\Cone}{\mathscr{P}}
\newcommand{\bCone}{\mathscr{P}_{\mathrm{b}, \Lambda}}

\newcommand{\rCone}{\mathscr{P}_{\mathrm{r}, \Lambda}}
\newcommand{\rConed}{\mathscr{P}_{\mathrm{r}, \Lambda\rq{}}}
\newcommand{\tCone}{\tilde{\mathscr{P}}_{\mathrm{r}, \Lambda}}
\newcommand{\tConed}{\tilde{\mathscr{P}}_{\mathrm{r}, \Lambda'}}

\newcommand{\no}{\nonumber \\}

\newcommand{\Nel}{\hat{N}^{(L)}_{ e}}
\newcommand{\Nol}{\hat{N}^{(L)}_{ o}}

\newcommand{\Ner}{\hat{N}^{(R)}_{e}}
\newcommand{\Nor}{\hat{N}^{(R)}_{o}}

\newcommand{\hn}{\hat{n}}
\newcommand{\dhn}{\delta \hat{n}}

\newcommand{\B}{{\boldsymbol b}}

\newcommand{\Ue}{U}
\newcommand{\ilim}[1][]{\mathop{\varinjlim}\limits_{#1}}

\newcommand{\bs}{\boldsymbol}
\newcommand{\bol}{\boldsymbol}

\newcommand{\IL}{\mathfrak{F}_{\LL}}
\newcommand{\ILd}{\mathfrak{F}_{\LL\rq{}}}
\newcommand{\F}{\mathrm{F}}

\newcommand{\IR}{\mathfrak{F}_{\LR}}

\newcommand{\LR}{\Lambda_R}
\newcommand{\LL}{\Lambda_L}
\newcommand{\AFL}{\Fock^{\rm F}_{\LL}}

\def\wrt{\ \text{w.r.t.}\ }
\def\tr{\mathrm{Tr}}

\begin{document}

\newtheorem{define}{Definition}[section]
\newtheorem{Thm}[define]{Theorem}
\newtheorem{Prop}[define]{Proposition}
\newtheorem{lemm}[define]{Lemma}
\newtheorem{Subl}[define]{Sublemma}
\newtheorem{rem}[define]{Remark}
\newtheorem{assum}{Condition}
\newtheorem{Ex}{Example}
\newtheorem{coro}[define]{Corollary}
\newtheorem*{Proof}{Proof}

\maketitle
\begin{abstract}

A new framework for the reflection positivity, based on order-preserving operator inequalities, is proposed. This framework is utilized to investigate one-dimensional fermion-phonon systems, with a particular focus on the detailed examination of ground state properties. Our analysis reveals that the reflection positivity provides a consistent description of the charge-density-wave (CDW) order in such systems. Additionally, we establish the existence of CDW long-range order in the ground state when the Coulomb interaction is  long range.

\end{abstract}

\setcounter{tocdepth}{2}
\tableofcontents
\section{Introduction and summary of the results derived from the theory constructed in this paper}

\subsection{Overview}\label{DerPhSys}
Interacting spinless fermions have been the subject of extensive study for several decades, serving as a fundamental model for strongly correlated systems. In fact, the interacting spinless fermion model can be considered as an effective description of the extended Hubbard model in the regime of strong coupling \cite{Dias2000, Kivelson2004}. As a result, this model is frequently employed to investigate the charge dynamics of many-electron systems. Moreover, it has recently garnered significant attention due to its relevance in the exploration of Majorana fermions in topological superconductors \cite{NadjPerge2014, Sato_2017}.
\smallskip

Of particular interest is the one-dimensional case, as it captures
 phenomena such as charge-density-wave (CDW) order \cite{Henley2001} and the nematic phase \cite{Kivelson2004}. Moreover, the one-dimensional model with nearest-neighbor interactions is mathematically equivalent to the anisotropic Heisenberg chain and possesses exact solvability.
\smallskip

In this manuscript, we investigate the interaction between fermions and dispersionless phonons in a one-dimensional setting. 
The system is described by the Hamiltonian:
\begin{align}
H_{\Lambda}=&\sum_{j\in \Lambda}(-t)
(c_{j  }^* c_{j+1}+c_{j+1}^*c_{j})
+
\sum_{i, j\in
 \Lambda}U(i-j)\delta \hn_i \delta \hn_j+\no
 &+
g\sum_{j\in \Lambda}\delta \hn_j(a_j+a_j^*)+\omega\sum_{j\in \Lambda}a_j^*a_j.
\end{align} 
For precise definitions of the symbols used, please refer to Section \ref{SetUp}.
This model is of practical importance since a wide range of quasi-one-dimensional materials undergo the Peierls instability as a result of electron-phonon interaction.
A significant body of literature exists on related models, with some notable rigorous results. For instance, in the work of \cite{Lieb1995}, the Peierls instability was analyzed using the Born--Oppenheimer approximation. In \cite{Miyao2012}, the uniqueness of the ground state was proven, and in \cite{Langmann2015}, the bosonization of a fermion-phonon model was investigated, among other contributions. However, it is worth noting that rigorous studies of this model are relatively scarce in comparison to the broader range of research on the topic.
\smallskip

The primary objective of this paper is to investigate the characteristics of the ground state of the Hamiltonian $H_{\Lambda}$ at half-filling. Specifically, we establish a proof demonstrating that the unique ground state of $H_{\Lambda}$ displays the CDW order. Furthermore, we demonstrate that this CDW order exhibits long-range behavior when the interaction $U(i-j)$ is  long-range.
Our approach is distinguished by the incorporation of the method of reflection positivity, along with the utilization of order-preserving operator inequalities. This combination of techniques serves as a distinctive feature of our study, as elucidated in the following explanation.
\smallskip

The concept of reflection positivity originated in the field of quantum field theory \cite{Osterwalder1973, Osterwalder1975} and has since found numerous applications in both classical and quantum statistical physics \cite{Biskup2006, Bjrnberg2013, DLS, Frhlich1978, Frhlich1980, Frhlich1976, Kennedy1988}. For a comprehensive understanding of the mathematical aspects of reflection positivity, we recommend referring to \cite{Neeb2018}.
In 1989, Lieb utilized the notion of reflection positivity in the spin space, known as spin reflection positivity, to establish the presence of ferrimagnetism in the ground state of the Hubbard model \cite{Lieb1989}. This approach, developed by Lieb, has proven to be valuable in the study of strongly correlated electron systems \cite{Miyao2019, Shen1998, Tian2004, Ueda1992, Yanagisawa1995}. A review of its applications can be found in \cite{Tasaki2020}.
More recently, Jaffe and Pedrocchi discovered a hidden form of reflection positivity in the Majorana fermion representation of the spinless fermion model. They applied this discovery to a Majorana fermion model with topological order \cite{Jaffe2015}. Further applications of this methodology can be found in \cite{Jaffe2016, Jaffe2020, Wei2015, yoshida2020rigorous}.
 \smallskip

In general, extending the concept of spin reflection positivity to interacting electron-phonon systems presents certain challenges. Firstly, the unbounded nature of bosonic operators requires careful mathematical treatment. Secondly, the electron-phonon interaction introduces technical difficulties when applying the hole-particle transformation.
Freericks and Lieb were the first to apply spin reflection positivity to interacting electron-phonon systems \cite{Freericks1995}, although the hole-particle transformation was not utilized in their method. Subsequently, one of the authors of this paper devised an operator-theoretic description of spin reflection positivity, enabling a detailed examination of interacting electron-phonon systems \cite{Miyao2012, Miyao2016, Miyao2019, miyao2020electronphonon}. This approach harmonizes with the hole-particle transformation and can be applied to various interacting electron-phonon systems, including the Holstein--Hubbard model, the SSH model, the Kondo lattice model, and others.
\smallskip

Since we are examining {\it spinless} fermions in this paper, we are unable to employ the technique of {\t spin} reflection positivity directly to explore the ground state properties of $H_{\Lambda}$. By further extending the operator theoretic approach of spin reflection positivity, we can effectively investigate the Hamiltonian $H_{\Lambda}$, even when $U(i-j)$  is  long range.
Mathematically, the framework of this extended approach can be described by the standard forms of von Neumann algebras. The theory of the standard forms was established by Haagerup \cite{Haagerup1975} and is currently recognized as a fundamental concept in the field of operator algebras.
This novel perspective of reflection positivity will be discussed in detail in the following sections.
In conjunction with the standard forms, we can naturally define order-preserving operator inequalities, which were initially explored in \cite{Miura2003}.
One significant advantage of introducing these inequalities is the ability to employ valuable techniques established in \cite{Miyao2012,Miyao2016,Miyao2019,miyao2020electronphonon}.
This facilitates the analysis of ground state properties of $H_{\Lambda}$, specifically demonstrating the presence of the CDW. It should be noted that, at this stage, it remains unknown whether the CDW  is a long range order.
 \smallskip

The application of reflection positivity to quantum systems was initially discovered by Dyson, Lieb, and Simon \cite{DLS}. Specifically, they established the existence of a phase transition in various quantum spin systems. Furthermore, their method was extended to demonstrate the presence of Neel order in the ground state of the three-dimensional spin-$1/2$ Heisenberg antiferromagnet on the cubic lattice \cite{Kennedy1988}. For a comprehensive overview, refer to the review by Tasaki \cite{Tasaki2020}. By combining the approach of \cite{Kennedy1988} with our methodology based on order-preserving operator inequalities, we substantiate the existence of long-range  CDW order in the ground state of $H_{\Lambda}$,  given the condition of long-range Coulomb interaction. To the best of our knowledge, this work represents the first mathematical proof of the existence of  long-range CDW order in a fermion-phonon system, provided that the Coulomb interaction
 is  long range.

\subsection{Interacting fermion-phonon  system}\label{SetUp}
Let $\mathbb{F}=\{\Lambda_{\ell}:  \ell \in \BbbN\}$, where $\Lambda_{\ell}=
[-\ell, \ell-1]\cap \BbbZ$.
Given a  $\Lambda\in \mathbb{F}$,  we examine  a one-dimensional  fermion-phonon interacting  system on $\Lambda$ at  half-filling. 
The Hamiltonian of this system is defined as follows:
\begin{align}
H_{\Lambda}=H_{\Lambda}^{\rm F}+
g\sum_{j\in \Lambda}\delta \hn_j(a_j+a_j^*)+\omega\sum_{j\in \Lambda}a_j^*a_j. \label{DefHamiFP}
\end{align} 
In this context, $H_{\Lambda}^{\text{F}}$ represents the Hamiltonian that describes the fermions, and it is defined as:
\begin{align}
H^{\rm F}_{\Lambda}=&\sum_{j\in \Lambda}(-t)
(c_{j  }^* c_{j+1}+c_{j+1}^*c_{j})
+
\sum_{i, j\in
 \Lambda}U(i-j)\delta \hn_i \delta \hn_j. \label{DefHF}
\end{align} 
The second term in equation \eqref{DefHamiFP} represents the interaction between fermions and phonons, while the last term represents the energy associated with the phonons.
The  operator $H_{\Lambda}$ acts in the 
 Hilbert space:
  \be
\Fock_{\Lambda}=\AFock_{\Lambda}\otimes \BFock_{\Lambda}.
\ee
Here,   $\AFock_{\Lambda}$ is the fermionic  Fock space over 
$\ell^2(\Lambda)$:
 $\AFock_{\Lambda}=\bigoplus_{n=0}^{|\Lambda|} \bigwedge^n \ell^2(\Lambda)$, where
 $\bigwedge^n$ stands for the $n$-fold antisymmetric tensor product, and $\BFock_{\Lambda}$ is the bosonic Fock space over $\ell^2(\Lambda)$:
  $\BFock_{\Lambda}=\bigoplus_{n=0}^{\infty} \otimes_{\rm s}^n \ell^2(\Lambda)$, where $\otimes^n_{\rm s}$ denotes the $n$-fold  symmetric tensor product.

The fermionic annihilation-  and creation operators at site $j$ are denoted by $c_{j  }$ and $c_{j  }^*$, respectively.
They satisfy the standard anti-commutation relations:
\begin{align}
\{c_{i  }, c_{j}^*\}=\delta_{i, j},\ \  \ \{c_{i  }, c_{j  }\}=0\ \ (i, j\in \Lambda), 
\end{align}
where $\delta_{i, j}$ is the Kronecker delta.
 The operator $\delta \hn_j$ is defined by $\delta \hn_j=\hn_j-1/2$, 
where $\hn_j$ is the number operator at site $j$: $\hn_{j  }=c_{j  }^* c_{j  }$. 

The bosonic annihilation and creation operators at site $j$ are denoted as $a_j$ and $a_j^*$, respectively. These operators satisfy the canonical commutation relations:
\be
[a_i, a_j]=0,\ \ [a_i, a_j^*]=\delta_{i, j}\ \ (i, j\in \Lambda).
\ee

$U(j)$ represents the energy associated with the Coulomb interaction. Mathematically, $U(j)$ is a bounded real-valued function defined on $\mathbb{Z}$. The hopping amplitude between neighboring sites is denoted as $-t$, where $t$ is a positive value. The phonons are assumed to be dispersionless, with an energy of $\omega > 0$.
In equation \eqref{DefHF}, which defines the Hamiltonian $H_{\rm F}$, periodic boundary conditions are imposed on the fermionic creation and annihilation operators for the sake of simplification: $c_{\ell}=c_{-\ell},\ c^*_{\ell}=c^*_{-\ell}$.

By the Kato--Rellich theorem \cite[Theorem X.13]{Reed1975}, $H_{\Lambda}$ is self-adjoint on $\D(N_{\rm p})$ and bounded from below,
where $N_{\rm p}$ is the phonon number operator: $N_{\rm p}=\sum_{j\in \Lambda} a_j^*a_j$.
The self-adjoint operator $H_{\Lambda}$ has purely discrete spectrum for all $\Lambda\in \mathbb{F}$.

Let $\hat{N}_{\Lambda}$ be the number operator for the fermions: $\hat{N}_{\Lambda}=\sum_{j\in
\Lambda} \hn_{j}$.
The closed subspace $\mathfrak{H}_{\Lambda}=\ker(\hat{N}_{\Lambda}-|\Lambda|/2)$ is called the {\it half-filled subspace} and will play
an important role. The half-filled  subspace can be expressed as 
\be
\h_{\Lambda}=\h^{\rm F}_{\Lambda} \otimes \BFock_{\Lambda},\ \ \ \h^{\rm F}_{\Lambda}=\bigwedge^{|\Lambda|/2} \ell^2(\Lambda).
\ee
In the subsequent analysis, the restriction of $H_{\Lambda}$ to $\mathfrak{H}_{\Lambda}$ is also referred to as $H_{\Lambda}$, unless there is potential ambiguity.
In the present paper, we will study a net of Hamiltonians:  $\{H_{\Lambda} : \Lambda\in \mathbb{F}\}$, where $\mathbb{F}$ is considered as a directed set based on the inclusion relation.
Particular attention will be given to the examination of ground state properties.

\subsection{Principal questions}

 Let $\mathfrak{X}$ be a complex Hilbert space and  let $\Cone$ be a convex cone in
$\mathfrak{X}$. The dual cone, $\Cone^{\dagger}$, of $\Cone$ is defined by 
$\Cone^{\dagger}=
\{x\in \mathfrak{X}\, :\, \la x| y\ra \ge 0\  
\forall y\in \Cone
\}.
$
We say that  $\Cone $ is  {\it self-dual} if 
$
\Cone=\Cone^{\dagger}.
$
Henceforth,  we  always  assume that $\Cone$ is self-dual. It is well-known that there is a unique involution $J$ in $\mathfrak{X}$ satisfying $Jx=x$ for all $x\in \Cone$.  We introduce a real subspace, $\mathfrak{X}_{J}$, of $ \mathfrak{X}$ by 
$\mathfrak{X}_{J}=\{x\in \mathfrak{X} : Jx=x\}$. It should be noted that each element $x\in \mathfrak{X}_J$ can be uniquely decomposed as $x=x_+-x_-$, where $x_-, x_+\in \Cone$ and $\la x_-|x_+\ra=0$. Furthermore, it holds true that $\mathfrak{X}=\mathfrak{X}_J+\im \mathfrak{X}_J$, where $\im =\sqrt{-1}$.  For further details,  refer to \cite{Bratteli1987,Takesaki2003}.

We  introduce the following order structures in $\mathfrak{X}$. 

\begin{define}
{\rm 
\begin{itemize}
\item A vector $x$ is said to be  {\it positive w.r.t. $\Cone$} if $x\in
 \Cone$.  This is denoted as  $x \ge 0$  w.r.t. $\Cone$.

\item
 Let $x, y\in \mathfrak{X}_{J}$. If $x-y\in \Cone$, then we write this as $x\ge y$ w.r.t. $\Cone$.
 \item A vector $x \in\mathfrak{X}$ is called {\it strictly positive w.r.t.} $\Cone$,  whenever $\la x| y\ra >0$ for all $ y \in \Cone\setminus\{0\}$.
 This is denoted as $x>0$ w.r.t.  $\Cone$.
\end{itemize} 
}
\end{define}

Let $\mathscr{L}(\mathfrak{X})$ denote the set of all bounded operators on the Hilbert space $\mathfrak{X}$.
Upon introducing positivity into  $\mathfrak{X}$  by fixing a particular self-dual cone, it becomes possible to define advantageous order structures within $\mathscr{L}(\mathfrak{X})$ in the following manner.

\begin{define}{\rm 

Let $A, B\in \mathscr{L}(\mathfrak{X})$. 

\begin{itemize}
 \item
If $A \Cone \subseteq \Cone$, we denote this as $A \unrhd 0$ with respect to $\Cone$. Note that this notation was introduced by Miura \cite{Miura2003}. In this context, we say that {\it $A$ preserves positivity with respect to $\Cone$.}
\item Suppose that $A\h_{J}\subseteq
 \h_{J}$ and $B\h_{J} \subseteq
	     \h_{J}$. If $(A-B) \Cone\subseteq
	     \Cone$, then we write this as $A \unrhd B$ w.r.t. $\Cone$. 
\item 
We write  $A\rhd 0$ w.r.t. $\Cone$, if  $Ax >0$ w.r.t. $\Cone$ for all $x\in
\Cone \backslash \{0\}$. 
 In this case, we say that {\it $A$ improves the
positivity w.r.t. $\Cone$.}
\end{itemize} 
} 
\end{define}

Assuming that $A$ preserves positivity w.r.t.  $\Cone$, if $x\geq y$ w.r.t.  $\Cone$, then $Ax\geq Ay$ w.r.t. $\Cone$ holds, indicating that $A$ preserves  the order. Consequently,  we commonly refer to the inequalities related to the symbol $\unrhd$ as the {\it order-preserving operator inequalities}.

Let $x$ and $y$ be elements of $\Cone$. Assuming that $A\in \mathscr{L}(\mathfrak{X})$ satisfies $A\unrhd 0$ w.r.t.  $\Cone$, it follows from the definition that
\be
\la x|A y\ra \ge 0. \label{BaseEx}
\ee
In particular, if $x$ and $y$ are strictly positive w.r.t.  $\Cone$, and $A$ is non-zero, then the strict inequality in \eqref{BaseEx} holds true. These fundamental properties will be frequently utilized in this paper.
\begin{define}
{\rm
Consider a positive self-adjoint operator $A$ such that $e^{-\beta A} \unrhd 0$ w.r.t.  $\Cone$ for all $\beta \ge 0$. We define the semigroup $e^{-\beta A}$ to be {\it ergodic w.r.t. } $\Cone$ if the following condition holds: For every $x, y\in \Cone\setminus \{0\}$, there exists a $\beta \ge 0$ such that $\la x|e^{-\beta A} y\ra >0$. It is worth noting that the value of $\beta$ may depend on the choice of $x$ and $y$.
}
\end{define}
We readily confirm that if  $e^{-\beta A}\rhd 0$ w.r.t. $\Cone$ for all $\beta>0$, then 
$e^{-\beta A}$ is ergodic w.r.t. $\Cone$.

The subsequent proposition highlights the significance of the aforementioned operator inequalities.
\begin{Prop}\label{PFF}{\rm (Perron--Frobenius--Faris)}
Let $A$ be a  self-adjoint positive operator on $\mathfrak{X}$. Suppose that 
 $0\unlhd e^{-\beta A}$ w.r.t. $\Cone$ for all $\beta \ge 0$,  and that  $E=\inf
 \mathrm{spec}(A)$ is an eigenvalue.
 Then,  the following
 are equivalent:
\begin{itemize}
\item[\rm  (i)] E is a simple eigenvalue. The corresponding eigenvector can be chosen
to be strictly positive with respect to  $\Cone$.
\item[\rm  (ii)] $e^{-\beta  A}$ is ergodic w.r.t. $\Cone$. 

\end{itemize}
\end{Prop} 
\begin{proof} See, e.g.,    \cite{Faris1972, Reed1978}.  
\end{proof}

In general, there exists an infinite number of self-dual cones in $\mathfrak{X}$. Consequently, we can define an infinite variety of positivities within the single Hilbert space $\mathfrak{X}$. This leads us to the natural questions:

\begin{description}
\item[{\bf Q. 1.}] What types of positivities are exhibited by the ground states of $H_{\Lambda}$?
\item[{\bf Q. 2.}] What is the physical significance associated with each positivity of the ground state?
\end{description}

In the remainder of this section, we will outline the results obtained from the theory based on order-preserving operator inequalities, which addresses the two questions. Further details about our theory and the proofs of these theorems can be found in the subsequent sections.

\subsection{Background positivity: $\bCone$ }

In this study, we occasionally utilize Dirac notation to facilitate the clarity of our presentation.
For any given subset $X\subseteq \Lambda$, we define a vector in $\AFock_{\Lambda}$ as follows:
\begin{align}
|X\ra=\Bigg[\prod_{j\in X} c_{j  }^*\Bigg]\Omega^{\rm F}_{\Lambda}.
\end{align}
Here,  $\Omega^{\rm F}_{\Lambda}$ represents the Fock vacuum in $\AFock_{\Lambda}$. The product $\prod_{j\in X} A_j$
is determined based on the natural ordering:
 $
 \prod_{j\in X} A_j=A_{x_1} A_{x_2}\cdots A_{x_m}
 $,  where $x_1<x_2<\cdots<x_m$.
It is evident that the set $\{|X\ra : X\subseteq \Lambda\}$ forms a complete orthonormal system (CONS) in $\AFock_{\Lambda}$.

Consider the set of particle configurations at half-filling given by:
\begin{align}
\mathcal{E}_{\Lambda}=\{X\subset \Lambda\, :\, |X|=|\Lambda|/2\}.
\end{align}
It can be readily observed that the set of vectors $
\{|X\ra\, :\, 
X\in \mathcal{E}_{\Lambda}\}$ forms a CONS of $\mathfrak{H}^{\rm F}_{\Lambda}$. Hence, each $\psi\in \mathfrak{H}_{\Lambda}$ can be  expressed as 
\be
\psi=\sum_{X\in \mathcal{E}_{\Lambda}} |X\ra \otimes |\psi_X\ra,\ \ \ \psi_X\in \BFock_{\Lambda}.
\ee

Recall that the bosonic Fock space $\BFock_{\Lambda}$ can be identified with $L^2(\BbbR^{\Lambda})$:
\be
\BFock_{\Lambda}=L^2(\BbbR^{\Lambda}). \label{BoseIdn1}
\ee
Using this identification, we can introduce a  self-dual cone in $\BFock_{\Lambda}$ by 
\be
\BFock_{\Lambda, +}=L^2_+(\BbbR^{\Lambda})=\{F \in L^2(\BbbR^{\Lambda}) : F({\bs \phi}) \ge 0 \ \ \mbox{a.e.}\}. \label{BoseIdn2}
\ee
Now we are ready to introduce the first self-dual cone.
\begin{define}
{\rm 
We define a self-dual cone $\bCone$ by 
\begin{align}
\bCone=\overline{\mathrm{coni}}(
\{|X\ra\otimes|\psi\ra\, :\, 
X\in \mathcal{E}_{\Lambda},\ \psi\in \BFock_{\Lambda, +}\}
). 
\end{align}
Here, $\overline{\mathrm{coni}}(S)$ represents the closure of the conical hull of a given subset $S$ of $\Fock_{\Lambda}$.
The positivity determined by $\bCone$ is referred to as the {\it background positivity}.

}
\end{define}

 It is worth noting that if we define a self-dual cone in  $\h_{\Lambda}^{\rm F}$ as 
 \be
 \bCone^{\rm F} =\mathrm{coni}(
\{|X\ra\, :\, 
X\in \mathcal{E}_{\Lambda}\}
),
\ee
 then we have $\bCone=\bCone^{\rm F} \otimes \BFock_{\Lambda, +}$.
Regarding the tensor product of self-dual cones, refer to  \cite{Miyao2021} for detailed information. Additionally, we can represent $\bCone$ in the form of a direct integral of self-dual cones as $\bCone = \int^{\oplus}_{\mathbb{R}^{|\Lambda|}} \bCone^{\mathrm{F}} d\bs{q}$.

Remark that $\psi\ge 0$ w.r.t. $\bCone$ if and only if $\psi_X\ge 0$  w.r.t. $\BFock_{\Lambda, +}$ for all $X\in \mathcal{E}_{\Lambda}$, 
and $\psi>0$ w.r.t. $\bCone$  if and only if $\psi_{X}>0$ w.r.t. $\BFock_{\Lambda, +}$ for all $X\in \mathcal{E}_{\Lambda}$. Furthermore, the involution $J$ associated with $\Cone_{\rm b, \Lambda}$ is given by $J\psi=\sum_{X\in \mathcal{E}_{\Lambda}}  \ket{X}\otimes \ket{\psi_X^*}$, where $\psi^*_X$ stands for the complex conjugate of  the function $\psi_X$.

\begin{Prop}\label{BGP}
For each $\Lambda\in \mathbb{F}$, 
the semigroup $e^{-\beta H_{\Lambda}}$ improves the positivity w.r.t. $\bCone$ for all $\beta>0$.
Consequently, according to Proposition \ref{PFF}, the ground state of $H_{\Lambda}$ is unique in $\mathfrak{H}_{\Lambda}$ and can be chosen to be   strictly positive w.r.t. $\bCone$. 

For the fermionic Hamiltonian $H_{\Lambda}^{\mathrm{F}}$, the same statement holds true if we replace $\mathfrak{H}_{\Lambda}$ with $\mathfrak{H}_{\Lambda}^{\mathrm{F}}$ and $\bCone$ with $\bCone^{\mathrm{F}}$.
\end{Prop}
\begin{proof}
By employing similar arguments to those presented in \cite{Miyao2012-2}, we can establish the proof of Proposition \ref{BGP}. It is noteworthy that even though \cite{Miyao2012-2} examines a spin-$1/2$ model, the absence of spin degrees of freedom does not affect the modification of the proof outlined in \cite{Miyao2012-2}.
\end{proof}

\begin{rem}\rm 
We highlight that there are no limitations imposed on the Coulomb interaction $U(i)$ in Proposition \ref{BGP}.
In contrast, the restriction that fermions can only hop to neighboring sites is essential in this theorem.
\end{rem}

Let $\psi_{\Lambda}$  be the unique ground state of $H_{\Lambda}$. The ground state expectation 
with respect to $\psi_{\Lambda}$   is then defined by 
\begin{align}
\la A\ra_{\Lambda} = \la \psi_{\Lambda} | A\psi_{\Lambda}\ra, \ \ A\in \mathscr{L}(\mathfrak{H}_{\Lambda}).
\end{align}
The expectation value with respect to the ground state of $H_{\Lambda}^{\mathrm{F}}$ can also be defined in a similar manner and denoted as $\langle A\rangle_{\Lambda}^{\mathrm{F}}$ for $A\in \mathscr{L}(\mathfrak{H}^{\mathrm{F}}_{\Lambda})$.

For each $X \subseteq \Lambda$, we set $
\hn_{X  }=\prod_{j\in X} \hn_{j  }
$ with $\hn_{\varnothing  }=\one$.
We then define a set of operators  by $\mathfrak{C}_{\Lambda}=\mathrm{coni}\{
\hn_{X}\, :\, X\subseteq \Lambda
\}$, where $\mathrm{coni}(S)$ represents the conical hull of a given set $S$.
\begin{coro}
For each $A\subseteq \mathfrak{C}_{\Lambda}$, we have $\la A\ra_{\Lambda}>0$ and $\la A\ra_{\Lambda}^{\rm F}>0$.
\end{coro}

It is worth noting that the aforementioned assertion provides only limited insight into the structure of the ground state $\psi_{\Lambda}$.
The primary role of background positivity is to ensure the uniqueness of the ground state.

\subsection{Reflection  positivity: $\rCone$}\label{RefPosiBasic}

In this paper, we will investigate an additional self-dual cone referred to as $\rCone$.
The associated positivity stemming from $\rCone$ is referred to as reflection positivity, which is particularly well-suited for describing the CDW order.
In order to define $\rCone$, we partition $\Lambda$ into two subsets, namely $\Lambda=\Lambda_L\cup \Lambda_R$, where
\be \Lambda_L=\{-\ell, -\ell+1, \dots, -1\},\ \ \Lambda_R=\{0, 1, \dots, \ell-1\}.
\ee
Using the identification
$\AFock_{\Lambda}=\AFock_{\LL} \otimes \AFock_{\LR}$ and $\BFock_{\Lambda}=\BFock_{\LL} \otimes \BFock_{\LR}$, we obtain the useful identification:
\be
\Fock_{\Lambda}=\Fock_{\LL} \otimes \Fock_{\LR}.
\ee
The von Neumann algebras, denoted as $\mathfrak{M}^{\text{F}}_{\Lambda_L}$ and $\mathfrak{M}_{\Lambda_L}$, are of considerable importance in our analysis. They are defined as follows:
\be
\mathfrak{M}^{\rm F}_{\LL}=\{A\in \mathscr{L}(\Fock^{\rm F}_{\LL}) : [A, \hat{N}_{\LL}]=0\}, \quad
\mathfrak{M}_{\LL}=\{A\in \mathscr{L}(\Fock_{\LL}) : [A, \hat{N}_{\LL}]=0\}.
\ee
We introduce the mapping $r: \Lambda_R \to \Lambda_L$ defined as follows:
\begin{align}
r(j) = -1 - j,\ \ \text{for } j \in \Lambda_R. \label{DefRefMap}
\end{align}
The mapping $r$ corresponds to a reflection about $j = -1/2$, as illustrated in Figure \ref{RPmapFig}.

\begin{figure}
\begin{center}
\includegraphics[scale=0.56]{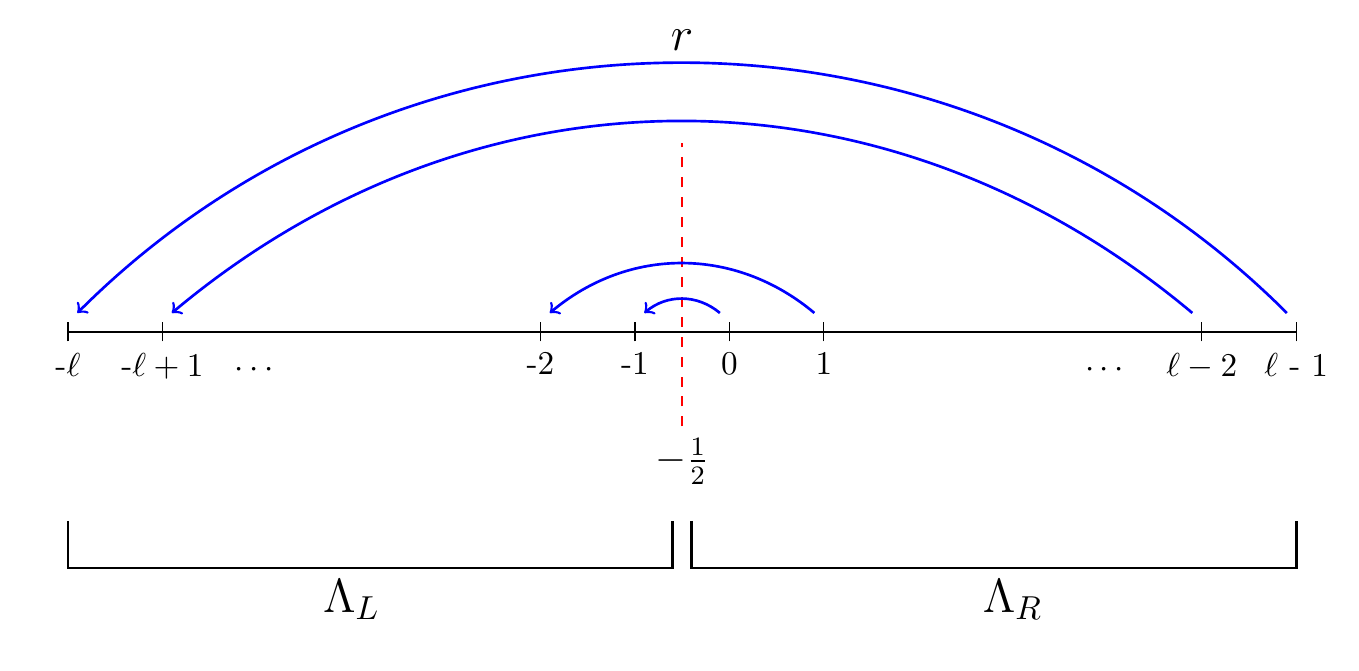}
\caption{The mapping  $r$ represents a reflection about $j=-1/2$.} 
    \label{RPmapFig}
\end{center}
\end{figure}

\begin{Thm}\label{RPConeBasic}
We have the following:
\begin{itemize}
\item[\rm 1.] There is an antiunitary operator $\tau $ from $\Fock_{\LL}$ onto $\Fock_{\LR}$ such that 
\be
\tau  \delta \hn_j  \tau^{-1}=-\delta \hn_{r(j)}
, \ \ \tau a_j\tau^{-1}=a_{r(j)},\ \  j\in \Lambda_R. \label{AbstPPI}
\ee
\item[\rm 2.] There is a self-dual cone $\rCone^{\rm F}$ in $\h^{\rm F}_{\Lambda}$  such that, for all $\vphi\in \rCone^{\rm F}$ and 
$A\in \mathfrak{M}_{\LL}^{\rm F}$, 
\be
\la \vphi|A\otimes \tau A \tau^{-1} \vphi\ra \ge 0 \label{VecPPF1}
\ee
holds. 

\item[\rm 3.] There is a self-dual cone $\rCone$ in $\h_{\Lambda}$  such that, for all $\vphi\in \rCone$ and 
$A\in \mathfrak{M}_{\LL}$, 
\be
\la \vphi|A\otimes \tau A \tau^{-1} \vphi\ra \ge 0 \label{VecPP1}
\ee
holds. 
\item[\rm 4.] Let $\rCone^{\sharp}$ represent either $\rCone^{\rm F}$ or $\rCone$. In this case, for any $\varphi\in \rCone^{\sharp}$, we have the following inequality:
\begin{align}
(-1)^m \big\la \vphi|\delta \hn_{i_m} \delta \hn_{i_{m-1}} \cdots  \delta \hn_{i_1} \delta \hn_{-i_1-1}
 \delta \hn_{-i_2-1} \cdots \delta \hn_{-i_m-1} \vphi\big\ra \ge 0 \label{VecPP2}
\end{align}
for $i_1, i_2, \dots, i_m\in \LL$. If $\vphi$ is strictly positive w.r.t. $\rCone^{\sharp}$, then  the   strict inequalities hold true in \eqref{VecPPF1},  \eqref{VecPP1} and \eqref{VecPP2}.
\end{itemize}
\end{Thm}

\begin{rem}\upshape 
The inequality \eqref{VecPP2} implies that vectors belonging to $\rCone^{\sharp}$ are well-suited for mathematically describing the CDW order.
\end{rem}

To establish a precise definition of  $\rCone^{\sharp}$, it is necessary to make some preliminary arrangements. Therefore, we will provide a comprehensive definition of  $\rCone^{\sharp}$ and present the proof of Theorem \ref{RPConeBasic} in Section \ref{Sub3.5}.

\begin{define}
\upshape 
A vector $\vphi$ belonging to $\h_{\Lambda}$ or  $\h_{\Lambda}^{\rm F}$ is referred to as {\it reflection positive} when it is positive with respect to $\rCone$ or $\rCone^{\rm F}$. In the event that $\vphi$ is strictly positive with respect to $\rCone$ or $\rCone^{\rm F}$, it is said to be {\it strictly reflection positive}.
\end{define}

\begin{Ex}\label{CDWVex} \upshape
Let 
\be
\Omega^{\rm CDW}_{\Lambda}=(-1)^{(|\Lambda|+2)/4} \Bigg[
\prod_{j\in \Lambda_o} c_j^*
\Bigg] \Omega^{\rm F}_{\Lambda}, \label{VacCDW}
\ee
where $\Lambda_o=\{j\in \Lambda : \mbox{$j$ is odd}\}$.
Example \ref{PfCDWV} in Section \ref{DefRefPosi} demonstrates that $\Omega^{\rm CDW}_{\Lambda}$ and $\Omega^{\rm CDW}_{\Lambda}\otimes \Omega_{\Lambda}^{\rm B}$ are reflection positive for all $\Lambda\in \mathbb{F}$, where $\Omega_{\Lambda}^{\rm B}$ denotes the bosonic Fock vacuum in $\BFock_{\Lambda}$.
As depicted in Figure \ref{CDWPic}, $\Omega_{\Lambda}^{\rm CDW}$ is a representative vector that characterizes the CDW.
\begin{figure}
\begin{center}
\includegraphics[scale=0.52]{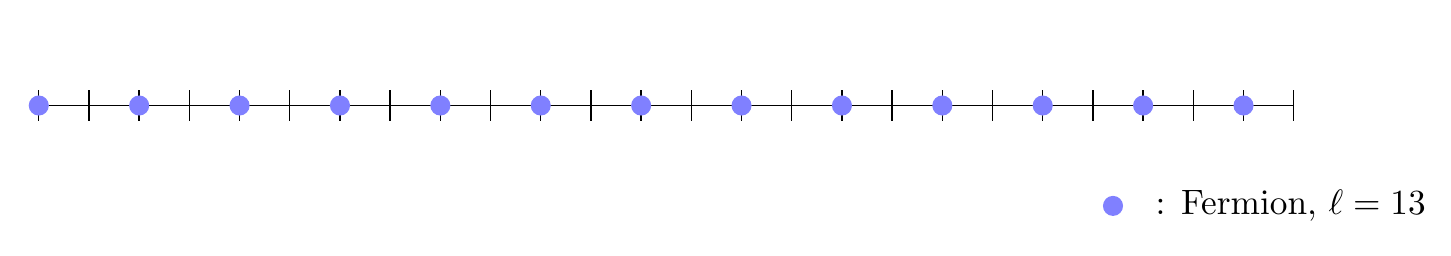}
\caption{The fermions are distributed across all sites with odd indices.} 
    \label{CDWPic}
\end{center}
\end{figure}
\end{Ex}

In order to state properties of $\rCone$, we need the following proposition.
\begin{Prop}\label{HilInd2}
\begin{itemize}
\item[\rm 1.]
There is a family of isometric linear mappings $\{\iota_{\Lambda\rq{}\Lambda} : \Lambda, \Lambda'\in \mathbb{F} \ \mbox{with $\Lambda\subseteq \Lambda'$}\}$ satisfying the following:
\begin{itemize}
\item[\rm (i)] $\iota_{\Lambda\rq{} \Lambda}: \h_{\Lambda} \to \h_{\Lambda\rq{}}$.
\item[\rm (ii)] If $\Lambda\subseteq \Lambda'\subseteq \Lambda''$, then $\iota_{\Lambda\rq{}\rq{}\Lambda'}\iota_{\Lambda'\Lambda}=\iota_{\Lambda\rq{}\rq{}\Lambda}.$
\end{itemize}
Hence, we can define the following inductive limit:
$
\h_{\BbbZ}=\ilim \h_{\Lambda}.
$
\item[\rm 2.] For each $\Lambda\in \mathbb{F}$, there exists an isometric linear mapping $\iota_{\Lambda}$ from $\h_{\Lambda}$ into 
$\mathfrak{H}_{\BbbZ}$ satisfying the following:
\begin{itemize}
\item[\rm (i)] If $\Lambda\subseteq \Lambda\rq{}$, then $\iota_{\Lambda\rq{}}\iota_{\Lambda\rq{}\Lambda}=\iota_{\Lambda}$.
\item[\rm (ii)] $\bigcup_{\Lambda \in \mathbb{F}} \iota_{\Lambda}\mathfrak{H}_{\Lambda}$ is dense in $\mathfrak{H}_{\BbbZ}$.
\end{itemize}
Furthermore, the Hilbert space $\mathfrak{H}_{\BbbZ}$ and the net of isometric linear mappings $ \{\iota_{\Lambda} : \Lambda\in \mathbb{F}\}$    are uniquely determined,   up to unitary equivalence. 
\item[\rm 3.] 
The same statements as 1. and 2. hold for the net $\{\mathfrak{H}_{\Lambda}^{\rm F} : \Lambda\in \mathbb{F}\}$. Therefore, the inductive limit $\mathfrak{H}_{\mathbb{Z}}^{\rm F}=\varinjlim \mathfrak{H}_{\Lambda}^{\rm F}$ can be defined. 
\item[\rm 4.] The mapping: $\h_{\Lambda}^{\rm F}\ni \vphi\mapsto \vphi\otimes \Omega^{\rm B}_{\Lambda}\in \h_{\Lambda}$ is an isometry. Therefore, we can consider $\h_{\Lambda}^{\rm F}$ as a closed subspace of $\h_{\Lambda}$. 
The orthogonal projection from $\h_{\Lambda}$ to $\h_{\Lambda}^{\rm F}$
 is explicitly  given by $S_{\Lambda}=1_{\Lambda}\otimes \ket{\Omega^{\rm B}_{\Lambda}}\bra{\Omega^{\rm B}_{\Lambda}}$, where
 $1_{\Lambda}$  represents the identity operator on $\h_{\Lambda}^{\rm F}$. The aforementioned statement holds true even in the case when $\Lambda=\BbbZ$.
\end{itemize}
\end{Prop}

Proposition \ref{HilInd2} is proved in Proposition \ref{PropIndLim}.

In the ensuing discussion, we shall identify $\h_{\Lambda}$ with $\iota_{\Lambda\rq{}\Lambda} \h_{\Lambda}$ and $\iota_{\Lambda} \h_{\Lambda}$, allowing us to treat $\h_{\Lambda}$ as a closed subspace of both $\h_{\Lambda\rq{}}$ and $\h_{\BbbZ}$.
The same type of identifications are also possible for $\mathfrak{H}_{\Lambda}^{\text{F}}$.
We use $P_{\Lambda\Lambda\rq{}}$ (resp. $P_{\Lambda\Lambda\rq{}}^{\rm F}$) to denote the orthogonal projection from $\h_{\Lambda\rq{}}$ onto $\h_{\Lambda}$ (resp. $\h_{\Lambda\rq{}}^{\rm F}$ onto $\h_{\Lambda}^{\rm F}$ ), and similarly, $P_{\Lambda}$ (resp. $P^{\rm F}_{\Lambda}$) denotes the orthogonal projection from $\h_{\BbbZ}$ onto $\h_{\Lambda}$ (resp. $\h_{\BbbZ}^{\rm F}$ onto $\h_{\Lambda}^{\rm F}$).
The following are some fundamental properties of $\rCone$ and $\rCone^{\rm F}$.

\begin{Thm}\label{PropConeSim}
\begin{itemize}
\item[\rm 1.]
Let $\Lambda, \Lambda\rq{}\in \mathbb{F}$. If $\Lambda\subseteq\Lambda\rq{}$, then  we have the following:
\begin{itemize}
\item[\rm (i)]
$P_{\Lambda\Lambda\rq{}}\rConed = \rCone$.
\item[\rm (ii)] $P_{\Lambda\Lambda\rq{}} \unrhd 0$ w.r.t. $\rConed$.
\end{itemize}
\item[\rm 2.] There is a self-dual cone $\Cone_{\rm r, \BbbZ}$ in $\h_{\BbbZ}$ satisfying the following: 
\begin{itemize}
\item[\rm (i)] $P_{\Lambda} \Cone_{\rm r, \BbbZ}=\Cone_{\rm r, \Lambda}$.
\item[\rm (ii)] $P_{\Lambda} \unrhd 0$ w.r.t. $\Cone_{\rm r, \BbbZ}$.
\end{itemize}
\item[\rm 3.] 
The statements 1. and 2. remain valid even if we replace $P_{\Lambda\Lambda\rq{}}, P_{\Lambda}$, $\rCone$, and $\Cone_{\rm r, \BbbZ}$ with $P^{\rm F}_{\Lambda\Lambda\rq{}}, P^{\rm F}_{\Lambda}$, $\rCone^{\rm F}$, and $\Cone^{\rm F}_{\rm r, \BbbZ}$, respectively.
\item[\rm 4.] The following statement holds true:
\begin{itemize}

\item[\rm (i)]  
$S_{\Lambda} \rCone=\Cone^{\rm F}_{\rm r, \Lambda}$. This relationship holds true even when $\Lambda=\BbbZ$.
\item[\rm (ii)] $S_{\Lambda} \unrhd 0$ w.r.t. $\rCone$.
\end{itemize}
\end{itemize}
\end{Thm}
See Theorem  \ref{StrongPropPP} for a more precise and  enhanced version of Theorem \ref{PropConeSim}.

\begin{rem}\upshape
\begin{itemize}
\item[(i)]
If $\Lambda\subseteq \Lambda'$, the reflection positivity in $\Lambda$ is consistently embedded within that in $\Lambda'$ as stipulated in Theorem \ref{PropConeSim}. This property holds even when $\Lambda'=\BbbZ$, and thus, reflection positivity can be naturally extended to the infinite chain.
\item[(ii)] 
Theorem \ref{PropConeSim}, item 4., implies the compatibility between reflection positivity in fermion-phonon systems and reflection positivity in fermionic systems.
\end{itemize}
\end{rem}

\subsection{Properties of the semigroup  $e^{-\beta H_{\Lambda}}$}
\subsubsection{The semigroup $e^{-\beta H_{\Lambda}}$ preserves  reflection positivity}

We present fundamental outcomes regarding the Hamiltonian $H_{\Lambda}$.
Hereinafter, we consistently assume the following:
\begin{description}
\item[\hypertarget{A}{ (A)}] 
\begin{itemize}
\item[(i)] $\ell$ is an odd number.
\item[(ii)] $U(i)$ is a bounded function on $\BbbZ$ such that $U(0)=0$ and $U(-i)=U(i)$ for all $i\in \BbbZ$.
\end{itemize}
\end{description}
We set $\mathbb{F}_{\rm o}=\{\Lambda_{\ell} : \mbox{$\ell$ is odd}\}$, where $\Lambda_{\ell}$ is given in Section \ref{DerPhSys}.
First, we  consider the following condition:
\begin{description}
\item[\hypertarget{B1}{ (B. 1)}] For each $\Lambda\in \mathbb{F}_{\rm o}$,  $\{ (-1)^{i+j} \Ue(i+j+1)\}_{i, j\in \LL}$ is a positive semi-definite matrix on $\BbbC^{\Lambda_L}$. Namely,
for every  $\{z_j\}_{j\in \Lambda_L}\in \BbbC^{\Lambda_L}$,
\begin{align}
\sum_{i, j\in \Lambda_L} z_i^*z_j(-1)^{i+j} \Ue(i+j+1)\ge 0.
\end{align}
\end{description}
Informally, this condition implies that the Coulomb interaction $U(i)$ is short range.

\begin{Ex}\upshape
Let us consider the nearest neighbor interaction:
\be
U(j)=U\delta_{j, -1},\ \ U\ge 0.
\ee
We can easily verify that $\{U(j): j \in \Lambda\}$ satisfies \hyperlink{B1}{\bf (B. 1)}.
It is worth noting that the Holstein model corresponds to the case where $U=0$.
In general,  $\{U(j): j \in \Lambda\}$
 satisfies \hyperlink{B1}{\bf (B. 1)} if,  and only if,  there exist a positive measure $\varrho$ on $[-1, 1]$ and a nonnegative constant  $c$ such that 
 \be
 (-1)^{j-1} U(j)= c\delta_{j, -1}+\int_{-1}^1 \lambda^{j-3} d\varrho(\lambda),\ \ j\le -1.
 \ee
 See \cite[Proof of Proposition 3.2]{Frhlich1978} for the proof.
\end{Ex}

\begin{Thm}\label{FirstThm}
Assume \hyperlink{A}{\bf (A)} and \hyperlink{B1}{\bf (B. 1)}. 
For each $\Lambda\in \mathbb{F}_{\rm o}$, we have the following.
\begin{itemize}
\item[\rm (i)] The semigroup $e^{-\beta H_{\Lambda}}$ preserves  the positivity w.r.t. $\rCone$ for all $\beta>0$.
\item[\rm (ii)]The ground state of $H_{\Lambda}$ is reflection  positive.
Moreover, we have
\be
\la A\otimes \tau A\tau^{-1} \ra_{\Lambda} \ge 0,\quad   A\in \mathfrak{M}_{\LL}. \label{GSRPP0}
\ee
In particular, we have
\begin{align}
(-1)^m \big\la \delta \hn_{i_m} \delta \hn_{i_{m-1}} \cdots  \delta \hn_{i_1} \delta \hn_{-i_1-1}
 \delta \hn_{-i_2-1} \cdots \delta \hn_{-i_m-1}  \big\ra_{\Lambda} \ge  0\label{GSRPP}
\end{align}
 for $i_1, i_2, \dots, i_m\in \LL$.
 \end{itemize}
 The aforementioned claims {\rm (i)} and {\rm (ii)} hold true even if we replace $H_{\Lambda}$, $\rCone$, $\mathfrak{M}_{\LL}$, and $\la \cdot \ra_{\Lambda}$ with $H_{\Lambda}^{\rm F}$, $\rCone^{\rm F}$, $\mathfrak{M}^{\rm F}_{\LL}$, and $\la \cdot \ra^{\rm F}_{\Lambda}$, respectively.
\end{Thm}

We will prove Theorem \ref{FirstThm} in Section \ref{Sect4}. 
\begin{rem}\upshape
\begin{itemize}
\item

The paper \cite{Lieb1995} essentially proved the reflection positivity of the ground state of a Hamiltonian for interacting fermions.
Theorem \ref{SparP} extends the result in \cite{Lieb1995} to the interacting fermion-phonon system.
\item 
Even when considering the interaction between fermions and phonons, the ground state still demonstrates CDW order, as indicated by the inequality \eqref{GSRPP}. This result aligns with existing findings in the literature, such as those presented in \cite{Fehske2005, McKenzie1996, Weie1998}.
\end{itemize}
\end{rem}

\subsubsection{Ergodicity of $e^{-\beta H_{\Lambda}^{\rm F}}$}

We consider the case of $g=0$. In this case, we can obtain more detailed information about the ground state of $H^{\rm F}_{\Lambda}$.

Let us now turn our attention to a stronger condition.

\begin{description}
\item[\hypertarget{B2}{ (B. 2)}] For each $\Lambda\in \mathbb{F}$,  $\{(-1)^{i+j} \Ue(i+j+1)\}_{i, j\in \LL}$ is a positive definite matrix on $\BbbC^{\Lambda_L}$. Namely, 
for every  $\{z_j\}_{j\in \Lambda_L}\in \BbbC^{\Lambda_L}\setminus \{{\bf 0}\}$, the following {\it  strict} inequality holds:
\begin{align}
\sum_{i, j\in \Lambda_L} z_i^*z_j(-1)^{i+j} \Ue(i+j+1)>0.
\end{align}
\end{description}
In contrast to \hyperlink{B1}{\bf (B. 1)}, this condition signifies that the interaction $\Ue(i)$ is long-range in nature.

\begin{Ex}\label{Long}\upshape
For a given  $\alpha>0$, suppose that $\Ue$ takes the following form:
\be
\Ue(j)=(-1)^{j+1}|j|^{-\alpha},\quad j\neq 0.
\ee
Here,  we set  $\Ue(0)=0$. 
By using the formula:
\be 
\Gamma(\alpha) |j|^{-\alpha}=\int_0^{\infty}e^{-j x}x^{\alpha-1}dx,
\ee
we readily confirm that $\{(-1)^{i+j}\Ue(i+j+1)\}_{i, j\in \LL}$ is positive definite for all $\alpha>0$, where $\Gamma(x)$ stands for the gamma function.
\end{Ex}

\begin{Thm}\label{SparP}
Assume \hyperlink{A}{\bf (A)} and \hyperlink{B2}{\bf (B. 2)}. 
Let us consider the Hamiltonian $H_{\Lambda}^{\rm F}$ of the  fermionic system.
For each $\Lambda\in \mathbb{F}_{\rm o}$, we have the following.
\begin{itemize}
\item[\rm (i)] The semigroup $e^{-\beta H^{\rm F}_{\Lambda}}$ is ergodic w.r.t. $\rCone^{\rm F}$.
\item[\rm (ii)]The ground state of $H_{\Lambda}^{\rm F}$  can be chosen  to be  strictly reflection  positive. 
Moreover, we have the following   strict inequality:
\be
\la A\otimes \tau A\tau^{-1} \ra_{\Lambda}^{\rm F} > 0,\quad A\in \mathfrak{M}_{\LL}^{\rm F}.
\ee
In particular, we have
\begin{align}
(-1)^m\big\la \delta \hn_{i_m} \delta \hn_{i_{m-1}} \cdots  \delta \hn_{i_1} \delta \hn_{-i_1-1}
 \delta \hn_{-i_2-1} \cdots \delta \hn_{-i_m-1}  \big\ra_{\Lambda}^{\rm F} > 0 \label{SharpCDW}
\end{align}
for $i_1, i_2, \dots, i_m\in \LL$.
\end{itemize}
\end{Thm}

We will provide a proof of Theorem \ref{SparP} in Section \ref{Sect5}.

\begin{rem}{\rm 
\begin{itemize}
\item[(i)]
Based on the strict inequality \eqref{SharpCDW}, we can infer that the CDW order in the ground state is strengthened under the condition \hyperlink{B2}{\bf (B. 2)}. However, to determine whether the order is long-range or not, further detailed analysis is required, as outlined in Theorem \ref{LRO}.

\item[(ii)] Whether Theorem \ref{SparP} holds for the case of $g\neq 0$ remains an open problem. 
\end{itemize}

}
\end{rem}

\subsection{Infinite chain}\label{SecInfChain}
\subsubsection{Properties of the ground state}
In this subsection, we assume the conditions \hyperlink{A}{\bf (A)} and \hyperlink{B1}{\bf (B. 1)}.
Let us consider a system on the infinite chain. Intuitively, the expectation value of an observable $A$ in the ground state of this system is defined as
\be
\la A\ra_{\BbbZ}=\lim_{\substack{\Lambda\in \mathbb{F}_{\rm o},  \ \Lambda\uparrow \BbbZ}} \la A\ra_{\Lambda},\ \ A\in \bigcup_{\Lambda\in \mathbb{F}}\mathscr{L}(\h_{\Lambda}). \label{DefInfGS}
\ee
This notion can be justified as follows. By employing Proposition \ref{HilInd2}, we can treat $\psi_{\Lambda}$ as a unit vector in $\h_{\BbbZ}$. Consequently, $\{\la \cdot \ra_{\Lambda} : \Lambda\in \mathbb{F}_{\rm o}\}$ forms a net of states on $\mathscr{L}(\h_{\BbbZ})$.
Since the unit ball of  the dual of $\mathscr{L}(\h_{\BbbZ})$  is compact in the weak-$*$ topology \cite[Theorem IV.21]{Reed1981}, there is a convergent subnet of    $\{\la \cdot \ra_{\Lambda} : \Lambda\in \mathbb{F}_{\rm o}\}$.
   The state $\la \cdot \ra_{\BbbZ}$  is thus defined as the weak-$*$ limit of this convergent subnet.

 Let  $\BbbZ_-=\{j\in \BbbZ : j <0\}$. 
 By Proposition \ref{HilInd2}, $\mathfrak{M}_{\LL}$ is naturally embedded in $\mathfrak{M}_{\LL\rq{}}$, provided that $\Lambda\subseteq \Lambda\rq{}$. We denote by  $\mathfrak{M}_{\BbbZ_-}$ the von Neumann algebra 
 generated by $\bigcup_{\Lambda\in \mathbb{F}_{\rm o}}\mathfrak{M}_{\LL}$.
 
 Similarly, we can define the ground state $\la \cdot \ra_{\BbbZ}^{\rm F}$ and the von Neumann algebra  $\mathfrak{M}^{\rm F}_{\BbbZ_-}$ for the fermionic system on the infinite chain.
\begin{Thm}\label{GSGProp}
Assume \hyperlink{A}{\bf (A)} and \hyperlink{B1}{\bf (B. 1)}.
We have the following:
\begin{itemize}
\item[\rm 1.] For any $A\in \mathfrak{M}_{\BbbZ_-}$, we have
\be
\big\la A\otimes \tau_{\BbbZ} A\tau_{\BbbZ}^{-1} \big\ra_{\BbbZ} \ge 0, \label{InfiniteGS}
\ee
where $\tau_{\BbbZ}$ is the unique extension of $\tau$ given in Theorem \ref{RPConeBasic}.
\item[\rm 2.] Suppose that $\la \cdot \ra_{\BbbZ}$ is a normal state on $\mathfrak{M}_{\BbbZ_-}\otimes 1:=
\{A\otimes 1 : A\in \mathfrak{M}_{\BbbZ_-}\}$.  Then there exists a unique normalized vector $\psi_{\BbbZ} \in \h_{\BbbZ}$ satisfying the following:
\begin{itemize}
\item[\rm (i)] $\psi_{\BbbZ}$ is reflection  positive, i.e., $\psi_{\BbbZ} \in \Cone_{{\rm r},  \BbbZ}$.
\item[\rm (ii)] $\la A\ra_{\BbbZ}=\la \psi_{\BbbZ}|A\psi_{\BbbZ}\ra$ for all $ A\in \mathscr{L}(\h_{\BbbZ})$.
\end{itemize}
Furthermore, if $\la \cdot \ra_{\BbbZ}$ is a faithful state on $\mathfrak{M}_{\BbbZ_-}\otimes 1$, then $\psi_{\BbbZ}$ is strictly reflection  positive, i.e., 
$\psi_{\BbbZ}>0$ w.r.t. $\Cone_{\rm r, \BbbZ}$. Hence, the strict inequality holds in \eqref{InfiniteGS}.
In particular, \eqref{SharpCDW} holds even for $\Lambda=\BbbZ$.
\end{itemize}
The aforementioned properties 1. and 2. hold true even if we replace $\mathfrak{M}_{\BbbZ_-}$, $\Cone_{{\rm r},  \BbbZ}$, and $\la\cdot \ra_{\BbbZ}$ with $\mathfrak{M}_{\BbbZ_-}^{\rm F}$, $\Cone_{{\rm r},  \BbbZ}^{\rm F}$, and $\la\cdot \ra_{\BbbZ}^{\rm F}$, respectively.
\end{Thm}

We give a proof of Theorem \ref{GSGProp} in Section \ref{PfThmInfSys}.

\begin{rem}
\upshape
We briefly mention the case where we do not assume the normality of the ground state $\la\cdot \ra_{\BbbZ}$ in  item $2$ of Theorem \ref{GSGProp}. In this case, we can simply consider the GNS representation with respect to $\la\cdot\ra_{\BbbZ}$, which leads to results similar to those stated in item $2$ of Theorem \ref{GSGProp}. However,  the drawback in this case is that the direct connection between finite systems and infinite systems, which will be discussed in detail in Section \ref{SectInfiniteCh}, becomes unclear. A thorough discussion on this matter will be provided in a separate paper.
\end{rem}

\subsubsection{Long-range order}

Thus far, we have demonstrated that the ground state of $H_{\Lambda}$ displays CDW order.
We now investigate whether this order  is long-range or not. 
To present the result, we require some preliminary information. Let us define
\be
W(j)=(-1)^{j+1} \Ue(j),\quad j\in \BbbZ.
\ee
   We assume the following condtion:
\begin{description}
\item[\hypertarget{C1}{ (C. 1)}] $\{W(j): j\in \BbbZ_+\} \in \ell^1(\BbbZ_+)$, where $\BbbZ_+=\{j\in \BbbZ :  j \ge 0\}$.
\end{description}
We now define a function on the interval $\mathbb{T}=[-\pi, \pi)$ as follows:
\be
\hat{R}(p)=4\sum_{j\ge 1} W(j)\{1-\cos (pj)\}.
\ee
It has been established in \cite{Frhlich1978} that $\hat{R}(p) \ge 0$ holds almost everywhere. For further details, refer to Section \ref{PfLRO}.
 Our second assumption is the following:
\begin{description}
\item[\hypertarget{C2}{ (C. 2)}]  $\hat{R}^{-1/2} \in L^1(\mathbb{T}, dp)$.
\end{description}
In Section \ref{PfLRO}, we will prove the following:
\begin{Thm}\label{LRO}
Assume \hyperlink{A}{\bf (A)}, \hyperlink{B1}{\bf (B. 1)},  \hyperlink{C1}{\bf (C. 1)} and \hyperlink{C2}{\bf (C. 2)}.
Define 
\be
\sigma=\frac{1}{4} (2\pi)^{1/2}  -(2\pi)^{-1/2} \int_{\mathbb{T}} dp \sqrt{\frac{F(p)}{\hat{R}(p)}},\label{DefSigma}
\ee
where 
\be
 F(p)=2t(1+\cos p).
\ee
If $t$ is sufficiently small such that $\sigma>0$, then
\begin{align}
\lim_{ |j|\to \infty} (-1)^j\big\la \delta \hn_j  \delta \hn_0\big\ra_{\BbbZ}\ge \sigma >0.
\end{align} 
Thus, CDW order is long-range ordered. 
This statement holds true even if we replace $\la\cdot \ra_{\BbbZ}$ with $\la \cdot \ra_{\BbbZ}^{\rm F}$.
\end{Thm}

The proof of Theorem \ref{LRO} will be presented in Section \ref{PfLRO}.

\begin{Ex}\upshape

We turn our attention to the long-range interaction $U(j)$ described in Example \ref{Long}.
According to \cite[Theorem 5.5]{Frhlich1978}, the function $1/\hat{R}(p)$ is integrable when $1<\alpha<2$, which consequently implies the integrability of $\hat{R}(p)^{-1/2}$. Consequently, Theorem \ref{LRO} holds true for the range $1<\alpha<2$.
\end{Ex}

\begin{rem}
\upshape 
We are uncertain about how far the parameter range in Theorem \ref{LRO} (i.e., $\sigma > 0$) is from the optimal conditions. However, given that the fermionic hopping term tends to destabilize CDW, while the Coulomb interaction term stabilizes it, the result of the theorem---that CDW exhibits long-range order when $t$ is small and $\hat{R}$ is large---appears physically reasonable.
\end{rem}

\subsection{Remark on the generalization of the interaction term}

So far, we have discussed only the case of on-site interaction when considering fermion-phonon interactions. The results of this paper can be extended to a more general fermion-phonon interaction, as elaborated below.

Suppose that the real-valued function, denoted as $g(j)$, defined on  $\mathbb{Z}$, is subject to the following conditions:
\begin{itemize}
\item[(i)] $g\in \ell^2(\BbbZ)$.
\item[(ii)] $g(-i)=g(i)$ for all $i\in \Lambda$.
\item[(iii)] $g\left(r(i)-r(j)\right)-g\left(r(i)-r(j+1)\right)=-\{g(i-j)-g(i-j-1)\}$.
\end{itemize}
In this context, we consider the extended Hamiltonian given by:
\begin{align}
H_{\Lambda}=H_{\Lambda}^{\rm F}+
\sum_{i, j\in \Lambda}g(i-j)\delta \hn_i(a_j+a_j^*)+\omega\sum_{j\in \Lambda}a_j^*a_j. \label{DefHamiFP2}
\end{align} 
It is worth noting that the previous results correspond to the special case of $g(i) = g \delta_{i, 0}$, where $\delta_{i, j}$ represents the Kronecker delta.
To express our findings, we introduce the effective Coulomb interaction as follows:
\be
U_{\rm eff, \Lambda}(i-j)=U(i-j)-\frac{2}{\omega}\sum_{k\in \Lambda} g(i-k) g(k-j).
\ee
When $\Lambda = \mathbb{Z}$, we simply denote $U_{\text{eff,} \Lambda}$ as $U_{\text{eff}}$. Under this condition, the following assertions hold:
\begin{itemize}
\item
Assuming the validity of the condition where $U$ in \hyperlink{B1}{\bf (B. 1)} is replaced by $U_{\text{eff,} \Lambda}$, Theorems \ref{FirstThm} and  \ref{GSGProp}  also hold for the extended Hamiltonian.
\item 
Assuming the conditions where $U$ is replaced by $U_{\text{eff}}$ in both \hyperlink{C1}{\bf (C. 1)} and \hyperlink{C2}{\bf (C. 2)}, Theorem \ref{LRO} holds for the extended Hamiltonian.
\end{itemize}

The proofs of these assertions are omitted here since they closely parallel the discussions in the subsequent sections regarding the case of onsite interactions.

The conditions of the theorems in this paper rarely involve those related to the interaction between fermions and phonons. The effects of this interaction become evident by examining interactions beyond the on-site interaction, as mentioned above.

\subsection{Organization}
The remaining sections of this paper are organized as follows:
In Section \ref{Sect2}, we commence by reviewing essential information regarding self-dual cones and standard forms. Additionally, we introduce a generalized form of order-preserving operator inequalities and present several fundamental lemmas. These concepts and results are crucial for expressing the notion of reflection positivity in the subsequent sections.
\smallskip

In Section \ref{DefRefPosi}, we provide a precise definition of $\rCone$ and thoroughly examine its properties. Additionally, we prove Theorems \ref{PropConeSim} and \ref{GSGProp}. These results serve as fundamental components for the subsequent sections.
\smallskip

Section \ref{Sect4} focuses on proving Theorem \ref{FirstThm}. We delve into the distinguishing properties of the semigroup generated by $H_{\Lambda}$ and establish that $e^{-\beta H_{\Lambda}}$ preserves the reflection positivity for all $\beta \ge 0$. By combining this result with Theorem \ref{PropConeSim}, we can deduce that the ground state displays CDW order.
\smallskip

Section \ref{Sect5} presents the proof of Theorem \ref{SparP}. We demonstrate that $e^{-\beta H_{\Lambda}^{\rm F}}$ is ergodic with respect to $\rCone^{\rm F}$. Due to the lengthiness of the proof, we divide it into several steps. The most intricate part of the proof is provided in Appendix \ref{PfPuzzle}. By utilizing this result in conjunction with Theorem \ref{FirstThm}, we ascertain that the CDW order is amplified when the interaction $U(i-j)$ is long-range.
\smallskip

In Section \ref{PfLRO}, we provide a proof of Theorem \ref{LRO}, which establishes the existence of long-range order in the ground state. Our proof extends the work presented in \cite{DLS}. However, we must exercise additional caution due to the presence of complex phase factors induced by the fermion-phonon interaction, which do not appear in \cite{DLS}.

\subsection*{Acknowledgements}

T. M.  was  supported by JSPS KAKENHI Grant Numbers  20KK0304 and 	23H01086.  T. M. expresses deep gratitude for the generous hospitality provided by Stefan Teufel at the Department of Mathematics, University of T\"{u}bingen, where T. M. authored a segment of this manuscript during his stay.

\subsection*{Declarations}
\begin{itemize}
\item  Conflict of interest: The Authors have no conflicts of interest to declare
that are relevant to the content of this article.
\item  Data availability: Data sharing is not applicable to this article as no
datasets were generated or analysed during the current study
\end{itemize}

\section{Order preserving operator inequalities}\label{Sect2}
\subsection{Basic properties of order preserving operator  inequalities}

In this subsection, we consider a separable complex Hilbert space $\h$. We assume that a self-dual cone $\Cone$ in $\h$ is given.

We readily confirm the following lemma:
\begin{lemm}
Let $A,B\in\mathscr L(\h)$.
 Suppose that  $A,B\unrhd0\wrt \Cone$. We have the following:
 \begin{itemize}
 \item[{\rm (i)}] If $a, b\ge 0$, then $aA+bB\unrhd0\wrt \Cone$;
 \item[{\rm  (ii)}] $AB\unrhd0\wrt \Cone$.
 \end{itemize}
\end{lemm}
\begin{proof}
For proof, see, e.g., \cite{Miura2003, Miyao2016}. 
\end{proof}

Below, we present three fundamental lemmas concerning the order-preserving operator inequalities for subsequent utilization.

\begin{lemm}
Let $A,B,C,D\in\mathscr L(\h)$. Suppose $A\unrhd B\unrhd0\wrt\Cone$ and $C\unrhd D\unrhd0\wrt\Cone$. Then we have $AC\unrhd BD\unrhd0\wrt\Cone$
\end{lemm}
\begin{proof}
For proof, see, e.g., \cite{Miura2003,Miyao2016}. 
\end{proof}

\begin{lemm}\label{WeakPP} Let $\{A_n\}_{n\in \BbbN}$ be a sequence of bounded operators on $\h$. Let $A$ be
a bounded operator on $\h$. Suppose that $A_n$ weakly converges to $A$ as $n\to \infty$.
If $A_n  \unrhd 0$  w.r.t. $\Cone$ for all $n\in \BbbN$, then $A\unrhd   0 $ w.r.t. $\Cone$.
\end{lemm}
\begin{proof}
See \cite[Proposition 2.8]{Miyao2019-2}.
\end{proof}

\begin{lemm}\label{ppiexp1}
Let $A,B$ be self-adjoint operators on $\h$. Assume that $A$ is bounded from below and  that  $B\in\mathscr L(\h)$. Furthermore, suppose that $e^{-tA}\unrhd0\wrt\Cone$ for all $t\geq0$ and $B\unrhd0\wrt\Cone$. Then we have $e^{-t(A-B)}\unrhd e^{-tA}\wrt\Cone$ for all $t\geq0$.
\end{lemm}

\begin{proof}
Because  $B\unrhd0\wrt\Cone$, we have $e^{tB}=\sum_{n=0}^{\infty}\frac{t^n}{n!}B^n\unrhd1\wrt\Cone$ for all $t\geq0$. By the Trotter product formula \cite[Theorem S. 20]{Reed1981} and Lemma \ref{WeakPP}, for all $t\geq0$, we obtain
\begin{align}
e^{-t(A-B)}=\lim_{n\to\infty}\left(e^{-\frac{t}{n}A}e^{\frac{t}{n}B}\right)^n\unrhd e^{-tA}\wrt\Cone.
\end{align}
\end{proof}

The following proposition will play an important role in the present paper.
\begin{Prop}\label{GSP}
Let $A$ be a positive self-adjoint operator. Assume that $e^{-\beta A}
 \unrhd 0$ w.r.t. $\Cone$ for some $\beta \ge 0$. Assume that $E=\inf
 \mathrm{spec}(A)$ is an eigenvalue of $A$. Then there exists a nonzero vector
 $x\in \ker(A-E)$ such that $x\ge 0$ w.r.t. $\Cone$. 
 \end{Prop}
\begin{proof} 
See  \cite[Proposition A.6]{Miyao2016(2)} for further details. It is worth noting that in \cite{Miyao2016(2)}, a stronger condition, namely $e^{-\beta A} \unrhd 0$ w.r.t.  $\Cone$ for {\it all} $\beta\ge 0$, is assumed. However, we can easily verify that the proof in \cite{Miyao2016(2)} can be extended to our current setting.
\end{proof}

\subsection{Order preserving operator inequalities on $\mathscr{L}^2(\h)$}\label{L2OrderP}

Let  $\mathscr{L}^2(\h)$ be the set of all Hilbert--Schmidt operators on $\h$:
 $\mathscr{L}^2(\h)=\{\xi\in \mathscr{L}(\h)\, :\, \tr[\xi^*\xi]<\infty\}$.
In what follows, we regard $\mathscr{L}^2(\h)$ as a Hilbert space equipped with the inner product 
$\langle \xi| \eta\rangle_2=\tr[\xi^* \eta],\ \xi, \eta\in \mathscr{L}^2(\h)$.
We will often omit the subscript $2$ in the inner product notation when there is no ambiguity.

Let $\h$ and $\overline{\h}$ be Hilbert spaces. We assume that an antiunitary operator $\vartheta$ is given, mapping $\h$ onto $\overline{\h}$. We define   the mapping $\Psi_\vartheta:\h\otimes\overline{\h}\longrightarrow\mathscr L^2(\h)$ as follows:
\begin{align}
\Psi_\vartheta(\phi\otimes\vartheta\psi)=|\phi\rangle\langle\psi|,\quad\phi,\psi\in\h. \label{DefTheta}
\end{align}
Since $\Psi_\vartheta$ is a unitary operator, we can identify $\h\otimes\overline{\h}$ with $\mathscr{L}^2(\h)$ in a natural way. This identification is expressed as:
\begin{align}
\h\otimes \overline{\h}\underset{\Psi_{\vartheta}}{=}\mathscr{L}^2(\h). \label{IdnSym}
\end{align}
Sometimes, we may omit the subscript $\Psi_{\vartheta}$ in (\ref{IdnSym}) if there is no ambiguity.

Let $A\in\mathscr L(\h)$ be given. We define the left multiplication operator $\mathcal{L}(A)$ and the right multiplication operator $\mathcal{R}(A)$ as follows:
\begin{align}
\mathcal{L}(A) \xi=A\xi,\ \ \mathcal{R}(A)\xi=\xi A,\ \ \xi\in \mathscr{L}^2(\h). \label{DefLRM}
\end{align}
It is trivial to see that $\mathcal{L}(A)$ and $\mathcal{R}(A)$ are bounded operators on $\mathscr{L}^2(\h)$.
Additionally, we can readily confirm the following properties:
\begin{align}
\mathcal{L}(A)\mathcal{L}(B)=\mathcal{L}(AB),\ \ \ \mathcal{R}(A)\mathcal{R}(B) =\mathcal{R}(BA),\\
(\mathcal{L}(A))^*=\mathcal{L}(A^*),\ \ \ (\mathcal{R}(A))^*=\mathcal{R}(A^*)
\end{align}
and 
\be
[\mathcal{L}(A), \mathcal{R}(B)]=0.
\ee
Under the identification (\ref{IdnSym}), we have the following identities:
\begin{align}
A\otimes 1=\mathcal L(A),\ \
1\otimes \vartheta A \vartheta^{-1}= \mathcal R(A^*),\ \ A\in \mathscr{L}(\h). \label{LRIden}
\end{align}
These identities establish the correspondence between the left and right multiplication operators on $\mathscr{L}^2(\h)$ and the tensor product operators on $\h\otimes \overline{\h}$ in terms of the identification (\ref{IdnSym}).

Let $\mathscr L^2_{+}(\h)$ be defined as follows:
\begin{align}
\mathscr L^2_{+}(\h)=\{\xi \in\mathscr L^2(\h)\,:\,\xi\geq0\}, \label{DefI_{2, +}}
\end{align}
where the inequality in the right-hand side of the equation denotes the standard operator inequality: Specifically,  $\xi\ge 0$ if and only if $\la x| \xi x\ra\ge 0$ for all $x\in \h$.
It is widely acknowledged that $\mathscr{L}^2_{+}(\h)$ constitutes a self-dual cone within $\mathscr{L}^2(\h)$, as presented in \cite[Proposition 2.5]{Miyao2016}. The involution associated with $\mathscr{L}^2_+(\h)$ is defined as follows:
\be
J\xi=\xi^*,\ \ \xi\in \mathscr{L}^2(\h).
\ee
Thus, we can observe that $\mathscr{L}^2(\h)_J=\{\xi\in \mathscr{L}^2(\h) : \mbox{$\xi$ is self-adjoint}\}$.
It is easily verified that
\be
J\mathcal{L}(A)J= \mathcal{R}(A^*),\ \ \ J\mathcal{R}(A)J= \mathcal{L}(A^*),\ \ A\in \mathscr{L}(\h). \label{ConjuProp}
\ee
Set 
\be
\mathfrak{N}_L=\{\mathcal{L}(A) : A\in \mathscr{L}(\h)\},\ \ \mathfrak{N}_R=\{\mathcal{R}(A) : A\in \mathscr{L}(\h)\}.
\ee
Then it follows from \eqref{ConjuProp} that 
\be
\mathfrak{N}_L\rq{}=J \mathfrak{N}_L J=\mathfrak{N}_R,
\ee
where $\mathfrak{N}_L\rq{}$ stands for the commutant of $\mathfrak{N}_L$.

Let $\xi$ denote an arbitrary vector in $\mathscr{L}^2(\h)$ that is strictly positive. Consequently, $\xi$ possesses the properties of being both cyclic and separating for $\mathfrak{N}_L$. Moreover, we can establish the equality:
\be
\mathscr{L}^2_+(\h)=\overline{\{A JAJ\xi : A\in \mathfrak{N}_L\}},
\ee
where the closure is taken with respect to the strong topology.
In conclusion, we can summarize the aforementioned points as follows.

\begin{Prop}\label{PPI2}
 The quadruple $\{\mathfrak{N}_L,  \mathscr{L}^2(\h), J, \mathscr{L}^2_+(\h)\}$ is a standard form\footnote{As for the definition of  the standard form, see \cite{Takesaki2003}} of $\mathfrak{N}_L$. In particular, 
we have $\mathcal L(A)\mathcal R(A^*)\unrhd0\wrt \mathscr{L}^2_{ +}(\h)$ for $A\in \mathscr{L}(\h)$. Hence, under  the identification (\ref{IdnSym}), we have $A\otimes \vartheta A \vartheta^{-1}
 \unrhd 0$ w.r.t. $\mathscr{L}^2_+(\h)$.
\end{Prop}

\begin{coro}\label{ExpPP}
Let $\vphi \in \mathscr{L}^2_+(\h)$. Then we have
\be
\la \vphi|\mathcal{L}(A) \mathcal{R}(A^*) \vphi\ra=\la \vphi|A\otimes \vartheta A\vartheta^{-1} \vphi\ra\ge 0,\quad A\in \mathscr{L}(\h). \label{AbstExp}
\ee 
If $\vphi$ is strictly positive w.r.t. $\mathscr{L}^2_+(\h)$, then the strict inequality holds in \eqref{AbstExp}.
\end{coro}

\begin{define} \upshape
The standard form $\{\mathfrak{N}_L,  \mathscr{L}^2(\h), J, \mathscr{L}^2_+(\h)\}$ determined by the triplet 
$\{\h, \overline{\h}, \vartheta\}$ is  called the {\it standard form associated with }$\{\h, \overline{\h}, \vartheta\}$.
\end{define}

To handle unbounded operators, we provide a brief explanation below:
Given an unbouded operator $A$ on $\h$ with a dense domain, the left multiplication operator $\mathcal{L}(A)$ can be defined as follows:
\begin{align}
\D(\mathcal{L}(A))&=\{\xi\in \mathscr{L}^2(\h) : \|A \xi\| <\infty\},\label{DefDom}\\
\mathcal{L}(A) \xi&=A\xi,\quad\xi\in \D(\mathcal{L}(A)). \label{DefAct}
\end{align}
Similarly, we can define the right multiplication operator $\mathcal{R}(A)$.
It is worth noting that if $A$ is a self-adjoint operator, then both $\mathcal{L}(A)$ and $\mathcal{R}(A)$ are also self-adjoint. Moreover, the identities \eqref{LRIden} hold in this case.

\subsection{Direct sum of self-dual cones}\label{SectDirectSum}

Consider a family of pairs of Hilbert spaces $\{\{\h_{\alpha}, \overline{\h}{\alpha}\} : \alpha\in I\}$. We assume that each pair is equipped with an antiunitary operator $\vartheta_{\alpha}$ mapping $\h_{\alpha}$ onto $\overline{\h}_{\alpha}$ for all $\alpha\in I$. For simplicity, suppose that the index set $I$ is countable.
For each $\alpha$, we can define the standard form $\mathscr{S}_{\alpha}=\{\mathfrak{N}_{L, \alpha}, \mathscr{L}^2(\h_{\alpha}), J_{\alpha}, \mathscr{L}_+^2(\h_{\alpha})\}$  associated with $\{\h_{\alpha}, \overline{\h}_{\alpha}, \vartheta_{\alpha}\}$.

By defining the following direct sums:
\be
\mathfrak{N}_L=\bigoplus_{\alpha\in I} \mathfrak{N}_{L, \alpha},\ \ \mathfrak{X}=\bigoplus_{\alpha\in I} \mathscr{L}^2(\h_{\alpha}),\ \ J=\bigoplus_{\alpha\in I}J_{\alpha},\ \ \Cone=\bigoplus_{\alpha\in I} \mathscr{L}^2_+(\h_{\alpha}),
\ee
we obtain a new standard form  $\mathscr{S}=\{\mathfrak{N}_L, \mathfrak{X}, J, \Cone\}$.
\begin{define}\label{DefDirectSum}\upshape
The standard form $\mathscr{S}$ is called the {\it direct sum of $\{\mathscr{S}_{\alpha} : \alpha \in I\}$}, and is denoted by $\mathscr{S}=\bigoplus_{\alpha\in I} \mathscr{S}_{\alpha}$.
\end{define}

For every $\mathcal{A} \in \mathfrak{N}_L$, we can associate a family of operators $A=\{A_{\alpha} \in \mathscr{L}(\h_{\alpha}) : \alpha\in I\}$ such that
$
\mathcal{A}=\bigoplus_{\alpha\in I} \mathcal{L}(A_{\alpha}).
$
To simplify notation, we can express this relationship as
\be
\mathcal{A}={\bs L}(A).
\ee
In a similar manner, for every $\mathcal{B}\in \mathfrak{N}_L' =J\mathfrak{N}_LJ$, there exists a family of operators $B=\{B_{\alpha}\in \mathscr{L}(\h_{\alpha}) : \alpha\in I\}$ such that $\mathcal{B}=\bigoplus_{\alpha\in I} \mathcal{R}(B_{\alpha})$. We denote this relationship as
\be
\mathcal{B}={\bs R}(B).
\ee

Let us define
\be
\h=\bigoplus_{\alpha\in I} \h_{\alpha},\ \ \overline{\h}=\bigoplus_{\alpha\in I} \overline{\h}_{\alpha},\ \ \vartheta=\bigoplus_{\alpha\in I} \vartheta_{\alpha}.
\ee
The standard form associated with the triplet $\{\h, \overline{\h}, \vartheta\}$ is denoted by $\mathscr{T}$.
The standard form $\mathscr{T}$ can be viewed as a natural extension of $\mathscr{S}$ in the following manner:
It is evident that $\mathfrak{X}$ is a closed subspace of $\mathscr{L}^2(\h)$. Let $P$ be the orthogonal projection from $\mathscr{L}^2(\h)$ to $\mathfrak{X}$. We can easily verify that $\Cone=P \mathscr{L}^2_+(\h)$. Furthermore, the involution $J_{\mathscr{T}}$ associated with $\mathscr{T}$ satisfies $J_{\mathscr{T}} \restriction \mathfrak{X}=J$.

For $A=\bigoplus_{\alpha\in I} A_{\alpha} \in \mathscr{L}(\h)$,  we can express
\be
{\bs L}(A)=\mathcal{L}(A) \restriction \mathfrak{X},\ \ {\bs R}(A)=\mathcal{R}(A)\restriction \mathfrak{X}. \label{ExRel}
\ee
In this context, we identify $A$ with the set of operators $\{A_{\alpha} : \alpha \in I\}$.
The matter of concern here is that ${\bs L}(A)$ and ${\bs R}(A)$ are exclusively defined when $A$ can be expressed as a direct sum. Nonetheless, during subsequent analysis, situations frequently arise where $A$ is not provided as a direct sum of operators. To address such cases as well, we extend the definitions of ${\bs L}(\cdot)$ and ${\bs R}(\cdot)$ by utilizing the relationship \eqref{ExRel} in the following manner:
Let
\be
\mathfrak{A}=\{A \in \mathscr{L}(\h) : A\xi,\ \xi A\in \mathfrak{X}\  \mbox{ for all $\xi\in \mathfrak{X}$
}\}.
\ee
It is evident that $\bigoplus_{\alpha\in I} \mathscr{L}(\h_{\alpha}) \subset \mathfrak{A}$ holds. For each $A\in \mathfrak{A}$,
we provide the following definitions:
\be
{\bs L}(A)=\mathcal{L}(A) \restriction \mathfrak{X},\ \ {\bs R}(A)=\mathcal{R}(A)\restriction \mathfrak{X}. \label{leftrightOP}
\ee
Let us emphasize that the extended definitions of ${\bs L}(A)$ and ${\bs R}(B)$ as described above indeed establish the operators on $\mathfrak{X}$ even when $A$ is not given as a direct sum.

By applying Proposition \ref{PPI2}, we obtain the following:
\begin{Prop}\label{PPI3}
\begin{itemize}
\item[\rm (i)]
We have ${\bs L}(A){\bs R}(A^*)\unrhd0\wrt \Cone$ for all  $A\in \mathfrak{A}$. Hence, by regarding $\mathfrak{X}$ as a closed subspace of $\mathscr{L}^2(\h)(=\h \otimes \overline{\h})$, we have $
A\otimes \vartheta A  \vartheta^{-1} \restriction \mathfrak{X}
 \unrhd 0$ w.r.t. $\Cone$.
 \item[\rm (ii)] Let $\vphi \in \Cone$. Then we have
\be
\la \vphi|{\bs L}(A) {\bs R}(A^*) \vphi\ra=\la \vphi|A\otimes \vartheta A\vartheta^{-1} \vphi\ra\ge 0,\quad A\in \mathfrak{A}. \label{AbstExp22}
\ee 
If $\vphi$ is strictly positive w.r.t. $\Cone$, then the strict inequality holds in \eqref{AbstExp22}.
 \end{itemize}
\end{Prop}

Lastly, we remark that the left- and right multiplication operators are defined for  unbounded operators:
Let $A$ be a densely defined unbounded operator on $\h$. As discussed in Section \ref{L2OrderP}, the left multiplication operator $\mathcal{L}(A)$ is defined by \eqref{DefDom} and \eqref{DefAct}.  
If $A$ satisfies
\be
A\xi\in \mathfrak{X},\quad  \xi\in \D(\mathcal{L}(A)),
\ee
then  we define ${\bs L}(A)$ by ${\bs L}(A)=\mathcal{L}(A) \restriction \mathfrak{X}$. 
Let us once again note that the operator $A$ is not necessarily provided as a direct sum in this context.
Similarly, for each densely defined unbouded operator $A$, we can define ${\bs R}(A)$.

\section{Reflection positivity in one-dimensional  fermion-phonon  system} \label{DefRefPosi}

\subsection{Canonical transformation of fermion operators}

Following the reference \cite{Frhlich1980}, we introduce a beneficial transformation in the following manner. For any $i\in \Lambda$, we define:
\begin{align}
u_{i  }=\Bigg\{
\prod_{j\neq i} (-\one)^{\hat{n}_{j  }}
\Bigg\}(c_{i  }^*+c_{i  }).
\end{align} 
It is noteworthy that $u_i$ represents a unitary operator.
We  readily  confirm the following:
\begin{itemize}
\item[1.] For each $i\in \Lambda$, 
\begin{align}
u_{i  }^2=\one,\ \  \ u_{i  }^*=u_{i  },
\ \ \ u_{i  } c_{i  }u_{i  }=c_{i  }^*.
\end{align}
\item[2.] For each $i, i'\in \Lambda$ with $i   \neq i'  $, 
\begin{align}
\{u_{i  },
 u_{i'  }\}=0,\ \ \ \
u_{i  }c_{i'  }u_{i  }=c_{i'  }.
\end{align} 
\end{itemize}

Let $\Lambda_e$ denote the set of even sites and $\Lambda_o$ denote the set of odd sites as defined:
\begin{align}
\Lambda_e=\{j\in \Lambda\, :\, \mbox{$j$ is even}\},\ \ \Lambda_o=\{j\in \Lambda\, :\, \mbox{$j$ is odd}\}.
\end{align}
We now introduce a unitary operator $\mathcal{U}_{\Lambda}$ on $\Fock_{\Lambda}$ given by:
\begin{align}
\mathcal{U}_{\Lambda}= \prod_{j\in \Lambda_o}
 u_{j}. \label{DefUniU}
\end{align}
Utilizing the aforementioned properties, we can establish the following lemma.
\begin{lemm}\label{Uni}
One has the following:
\begin{itemize}
\item[{\rm (i)}] $\mathcal{U}_{\Lambda} c_{j  } \mathcal{U}_{\Lambda}^* =c_{j  }^*$ for  each  $j\in
	     \Lambda_o$.
\item[{\rm (ii)}]  $\mathcal{U}_{\Lambda} c_{j  } \mathcal{U}_{\Lambda}^* =c_{j  }$ for  each  $j\in
	     \Lambda_e$.
\item[\rm (iii)] For each $\varphi \in \BFock_{\Lambda}$, 
$
\mathcal{U}_{\Lambda} \Omega^{\rm F}_{\Lambda} \otimes \varphi=\Omega^{\rm CDW}_{\Lambda} \otimes \vphi
$ holds, where $\Omega^{\rm CDW}_{\Lambda}$ is defined by \eqref{VacCDW}.
\end{itemize} 
\end{lemm}

\subsection{A representation of the half-filled subspace}

In order to describe  mathematical structures of $\mathcal{U}_{\Lambda}\mathfrak{H}_{\Lambda}$, we introduce several operators as follows:
\begin{itemize}
\item[1.] 
\begin{align}
\hat{N}_{ e}^{(L)}=\sum_{j\in \Lambda_e\cap \Lambda_L}
 \hat{n}_{j  },\ \ 
\hat{N}_{ o}^{(L)}=\sum_{j\in \Lambda_o\cap \Lambda_L}
 \hat{n}_{j  },\\
\hat{N}_{ e}^{(R)}=\sum_{j\in \Lambda_e\cap \Lambda_R}
 \hat{n}_{j  },\ \ 
\hat{N}_{ o}^{(R)}=\sum_{j\in \Lambda_o\cap \Lambda_R}
 \hat{n}_{j  }.
\end{align} 
\item[2.] 
\begin{align}
\hat{Q}_{\Lambda}^{(L)}=\Nel-\Nol,\ \
\hat{Q}_{\Lambda}^{(R)}= -\Big(\Ner-\Nor\Big).
\end{align} 

\end{itemize}

\begin{Prop}\label{UHil}
For each $\Lambda\in \mathbb{F}$, 
we have 
\begin{align}
\mathcal{U}_{\Lambda}\mathfrak{H}_{\Lambda}= \ker(\hat{Q}_{\Lambda}^{(L)}-\hat{Q}_{\Lambda}^{(R)}).
\end{align} 
\end{Prop} 
\begin{proof} We just note that 
$\mathcal{U}_{\Lambda}\hat{N}_{\Lambda}\mathcal{U}_{\Lambda}^*=\hat{Q}_{\Lambda}^{(L)}-\hat{Q}_{\Lambda}^{(R)}+|\Lambda|/2$.
 \end{proof}

\subsection{Properties  of  the  reflection mapping}

First,  we recall the following
identification:
\begin{align}
\Fock_{\Lambda}
=\IL \otimes
 \IR.\label{IdnLR}
\end{align}
Under this identification, we have
\begin{align}
c_j=
\begin{cases}
c_j\otimes 1,  & \mbox{if $j\in \Lambda_L$}\\
(-1)^{\hat{N}^{(L)}} \otimes c_j, & \mbox{if  $j\in \Lambda_R$},
\end{cases},\ \ 
a_j =
\begin{cases}
a_j \otimes \one,  & \mbox{ if $j \in \Lambda_L$}\\
\one \otimes a_j, & \mbox{ if $j\in \Lambda_R$}
\end{cases},
\end{align}
where $\hat{N}^{(L)}=\hat{N}^{(L)}_e+\hat{N}^{(L)}_o$. Recall that  $\Omega_{\Lambda}^{\rm F}$ and $\Omega^{\rm B}_{\Lambda}$ denote the Fock vacuums in $\AFock_{\Lambda}$ and $\BFock$, respectively. 
It is worth noting that, in correspondence with equation \eqref{IdnLR}, the following identifications hold:  $\Omega_{\Lambda}^{\rm \#}=\Omega_{\Lambda_L}^{\#}\otimes \Omega_{\Lambda_R}^{\rm \#}\ (\#=\rm B, F)$.
The following vector 
will be often useful:
\be
\Omega_{\Lambda}=\Omega_{\Lambda}^{\rm F} \otimes \Omega_{\Lambda}^{\rm B}. \label{DefOme}
\ee
Using the factorization properties of the Fock vacua mentioned above,  we can conclude that $\Omega_{\Lambda}=\Omega_{\Lambda_L}\otimes \Omega_{\Lambda_R}$.

For each $j\in \Lambda_L$, we define:
\begin{align}
b_{j  }=
\begin{cases}
&(-\one)^{\hat{N}^{(L)}}c_{j  },\ \ \mbox{if $j$ is even}\\
& c_{j  } (-\one)^{\hat{N}^{(L)}}, \ \ \mbox{if $j$ is odd}.
\end{cases}\label{Defb}
\end{align} 
It should be noted that $b_j$ and $b_j^*$ satisfy the canonical anti-commutation relations:
\be
\{b_i, b_j^*\}=\delta_{ij},\ \ \{b_i, b_j\}=0.
\ee

Before proceeding, recall the definition of $r$ given by equation \eqref{DefRefMap}. For convenience, we extend the mapping $r$ to the entire $\Lambda$ by defining $r(j)=-j-1$ for {\it all} $j\in \Lambda$.
Now, we introduce an involution $\vartheta$ from $\IL$ onto $\IR$ defined by the conditions $\vartheta \Omega_{\LL}=\Omega_{\LR}$ and:
\begin{align}
\vartheta b_j \vartheta^{-1}=(-1)^{r(j)}c_{r(j)}
, \ \ \vartheta a_j\vartheta^{-1}=a_{r(j)},\quad j\in \Lambda_L. \label{DefTheta}
\end{align} 
It is important to note the presence of the factor $(-1)^{r(j)}$ in equation \eqref{DefTheta}, as it will play a significant role, as seen in the proof of Lemma \ref{TransLR}.

Due to the property that $r$ maps even numbers to odd numbers, we can readily deduce the following:
\begin{align}
\hat{Q}_{\Lambda}^{(R)}=\vartheta \hat{Q}_{\Lambda}^{(L)} \vartheta^{-1}. \label{QRtoQL}
\end{align} 
We set 
\begin{align}
\mathbb{I}_{\Lambda}=
\begin{cases}
\{-|\Lambda|/2, -|\Lambda|/2+1, \dots, |\Lambda|/2-1, |\Lambda|/2\}, & \mbox{if $\ell$ is even}\\
\{-(|\Lambda|+1)/2, -(|\Lambda|+1)/2+1, \dots, (|\Lambda|-1)/2\}, & \mbox{if $\ell$ is odd}.
\end{cases}
\end{align}
It is readily confirmed that $\mathbb{I}_{\Lambda}=\mathrm{spec}(\hat{Q}^{(L)}_{\Lambda})=\mathrm{spec}(\hat{Q}_{\Lambda}^{(R)})$, where $\mathrm{spec}(A)$ represents the spectrum of the operator $A$.
For each $q\in \mathbb{I}_{\Lambda}$, we define:
\begin{align}
\IL(q)=\{\vphi\in \IL\, :\, \hat{Q}_{\Lambda}^{(L)} \vphi=q\vphi\},\ \
\IR(q)=\{\vphi\in \IR\, :\, \hat{Q}_{\Lambda}^{(R)} \vphi=q\vphi\}. 
\end{align} 
Furthermore,  we set $\AFL(q)=\{\vphi\in \AFock_{\LL} : \hat{Q}^{(L)}_{\Lambda}\vphi=q\vphi\}$. Then, we have the following expression:
\be
\IL(q)=\AFock_{\Lambda_L}(q) \otimes \BFock_{\Lambda_L}. \label{ITensor}
\ee

Applying Proposition \ref{UHil}, we have the following.
\begin{lemm}
For each $\Lambda \in \mathbb{F}$, one has 
\begin{align}
\mathcal{U}_{\Lambda}\mathfrak{H}_{\Lambda}=\bigoplus_{q\in \mathbb{I}_{\Lambda}} \IL(q)\otimes \vartheta\IL(q)
.
\label{0subspace}
\end{align}
\end{lemm} 
\begin{proof}
By utilizing the identification \eqref{IdnLR}, we can observe the following:
\be
\mathcal{U}_{\Lambda} \h_{\Lambda}=\mathcal{U}_{\Lambda} \ker(\hat{N}_{\Lambda}-|\Lambda|/2)
=\ker(\hat{Q}^{(L)}_{\Lambda}-\hat{Q}_{\Lambda}^{(R)}) =\bigoplus_{q\in \mathbb{I}_{\Lambda}}
\IL(q) \otimes \IR(q).
\ee
Furthermore, based on the relation \eqref{QRtoQL}, it holds that $\vartheta \IL(q)=\IR(q)$.
\end{proof}

\subsection{Definition of $\rCone$ }\label{DefConeRP}

 Let  
 $\tilde{\mathfrak{H}}_{\Lambda}=\mathcal{U} _{\Lambda}\mathfrak{H}_{\Lambda}$.
 With the help of \eqref{0subspace}, we can express $\tilde{\mathfrak{H}}_{\Lambda}$ as follows:
 \be
 \tilde{\mathfrak{H}}_{\Lambda}=\bigoplus_{q\in \mathbb{I}_{\Lambda}} 
 \THL(q),\ \ \ \THL(q)=\mathscr{L}^2(\IL(q)). \label{THLDirect}
 \ee
Using this expression,  we define a self-dual cone in $\tilde{\mathfrak{H}}_{\Lambda}$ as
\begin{align}
\tCone=\bigoplus_{q\in \mathbb{I}_{\Lambda}} 
\tCone(q),\ \ \tCone(q)=
\mathscr{L}^2_+(\IL(q)). \label{Fiber1}
\end{align} 
The involution associated with $\tCone$ can be expressed as:
\be
J=\bigoplus_{q\in \mathbb{I}_{\Lambda}} J_q. \label{tConeJ}
\ee
Here, $J_q$ represents the involution associated with $\tCone(q)$, given by $J_q\xi=\xi^*\ (\xi\in \mathscr{L}^2(\IL(q)))$.

Now, we can provide the precise definition of $\rCone$ as mentioned in Section \ref{RefPosiBasic}.
\begin{define}\label{DefRPCone2} {\rm
For each $\Lambda\in \mathbb{F}$, 
 we define a self-dual cone in $\mathfrak{H}_{\Lambda}$ as follows:
\begin{align}
\rCone=\mathcal{U}_{\Lambda}^{-1} \tCone. 
\end{align}
Here,  $\mathcal{U}_{\Lambda}$ is defined by \eqref{DefUniU}.
}
\end{define}

Set $\mathfrak{N}_{\Lambda, L}(q)=\{\mathcal{L}(A) : A\in \mathscr{L}(\Fock_{\LL}(q))\}\ (q\in \mathbb{I}_{\Lambda})$.
According to Proposition \ref{PPI2}, the quadruple $\{ \mathfrak{N}_{\Lambda, L}(q), \THL(q), J_q, \tCone(q)\}$
 is the standard form associated with the triplet $\{\Fock_{\LL}(q), \Fock_{\LR}(q), \vartheta\}$ for each $q\in \mathbb{I}_{\Lambda}$. Hence, by recalling  Definition \ref{DefDirectSum}, we can establish the equality between the standard forms:
 \be
 \{\mathfrak{N}_{\Lambda, L}, \THL, J, \tCone\} = \bigoplus_{q\in \mathbb{I}_{\Lambda}} \{
 \mathfrak{N}_{\Lambda, L}(q), \THL(q), J_q, \tCone(q)
 \}, 
 \ee 
where $
\mathfrak{N}_{\Lambda, L}=\bigoplus_{q\in \mathbb{I}_{\Lambda}} \mathfrak{N}_{\Lambda, L}(q)
$. Define
\be
\mathfrak{A}_{\Lambda}=\{A\in \mathscr{L}(\oplus_{q\in \mathbb{I}_{\Lambda}} \Fock_{\LL}(q)) : A\xi\in \THL, \xi A\in \THL\ \mbox{ for all $\xi\in \THL$}\}.
\ee
For each $A\in \mathfrak{A}_{\Lambda}$, we can define the (extended) left- and right multiplication operators, ${\bs L}(A)$ and ${\bs R}(A)$, by \eqref{leftrightOP}.

\subsection{Basic properties of $\tCone$ and definition of $\rCone^{\rm F}$}\label{Sub3.5}

It is a general property that for Hilbert spaces $\mathfrak{X}$ and $\mathfrak{Y}$, we have $\mathscr{L}^2(\mathfrak{X} \otimes \mathfrak{Y})=\mathscr{L}^2(\mathfrak{X}) \otimes \mathscr{L}^2(\mathfrak{Y})$.
Applying this property and using the expression \eqref{ITensor}, we obtain
\be
\THL(q)=\THL^{\rm F}(q) \otimes \THL^{\rm B},\quad  q\in \mathbb{I}_{\Lambda}, 
\ee
where
\be
\THL^{\rm F}(q)=\mathscr{L}^2(\AFock_{\Lambda_L}(q)),\ \ \ \THL^{\rm B}= \mathscr{L}^2(\BFock_{\Lambda_L}).
\ee

Given the Hilbert spaces denoted as $\mathfrak{X}_1$, $\mathfrak{X}_2$, and $\mathfrak{Y}$, it is widely recognized that the following relationship holds true:
$(\mathfrak{X}_1\oplus \mathfrak{X}_2)\otimes \mathfrak{Y}=(\mathfrak{X}_1\otimes \mathfrak{Y})\oplus (\mathfrak{X}_2\otimes \mathfrak{Y})$.
By utilizing this fact, we observe that 
\be
\THL=\THL^{\rm F} \otimes \THL^{\rm B}, \ \ \THL^{\rm F}=\bigoplus_{q\in \mathbb{I}_{\Lambda}} \THL^{\rm F}(q).
\ee

Similarly, we have
\be
\tCone(q)=
\tCone^{\rm F}(q) \otimes \tCone^{\rm B}, \quad  q\in \mathbb{I}_{\Lambda}, \label{PTensor}
\ee
where
\be
\tCone^{\rm F}(q)=\mathscr{L}_+^2(\AFock_{\Lambda_L}(q)),\ \ \ \tCone^{\rm B}= \mathscr{L}_+^2(\BFock_{\Lambda_L}). \label{DefrPq}
\ee
Here, the right-hand side of \eqref{PTensor} represents the tensor product of self-dual cones (see \cite[Appendix B]{Miyao2021}).
For clarity, we explain the definitions explicitly below:
Let $\mathfrak{N}_{\Lambda, L}^{\rm B}=\{\mathcal{L}(A) : A\in \mathscr{L}(\BFock_{\LL}) \}$.
In the setting of  this section, $\tCone^{\rm F}(q) \otimes \tCone^{\rm B}$ is defined by
\be
\tCone^{\rm F}(q) \otimes \tCone^{\rm B}
=\overline{\{AJ_qAJ_q \xi^{\rm F} \otimes \xi^{\rm B} : A\in \mathfrak{N}_{\Lambda, L}(q) \otimes \mathfrak{N}_{\Lambda, L}^{\rm B}\}},
\ee
where $\xi^{\rm F}$ and $\xi^{\rm B}$ are any strictly positive vectors in $\THL^{\rm F}(q)$ and 
$\THL^{\rm B}$, respectively. 

Furthermore, it holds that 
\be
\tCone=\tCone^{\rm F} \otimes \tCone^{\rm B},\ \ \ \tCone^{\rm F}= \bigoplus_{q\in \mathbb{I}_{\Lambda}} \tCone^{\rm F}(q), \label{DirectSumP}
\ee
where  $\tCone^{\rm F} \otimes \tCone^{\rm B}$ is defined by 
\be
\tCone^{\rm F} \otimes \tCone^{\rm B}
=\overline{\{AJAJ \eta^{\rm F} \otimes \eta^{\rm B} : A\in \mathfrak{N}_{\Lambda, L} \otimes \mathfrak{N}_{\Lambda, L}^{\rm B}\}},
\ee
where $\eta^{\rm F}$ and $\eta^{\rm B}$ are  any  strictly positive vectors in $\THL^{\rm F}$ and 
$\THL^{\rm B}$, respectively

We are prepared to establish the definition of reflection positivity in the Hilbert space $\h_{\Lambda}^{\F}$.
\begin{define}\label{DefRPCone3} 
\upshape
For each $\Lambda\in \mathbb{F}$, 
 we define a self-dual cone in $\mathfrak{H}^{\F}_{\Lambda}$ as
$\rCone^{\F}=\mathcal{U}_{\Lambda}^{-1} \tCone^{\F}$.
\end{define}

Before we proceed, let us demonstrate the assertion stated in Example \ref{CDWVex} within Section \ref{RefPosiBasic}.

\begin{Ex} \label{PfCDWV}\upshape
Let us examine the vector $\Omega_{\Lambda}$ defined by \eqref{DefOme}. Using the identification \eqref{IdnLR}, we find that:
\be
\Omega_{\Lambda}=\Omega_{\LL}\otimes \Omega_{\LR}=\Omega_{\LL} \otimes \vartheta \Omega_{\LL}.
\ee
This implies that $\Omega_{\Lambda} \ge 0$ w.r.t. $\tCone$. Hence, $\mathcal{U}^{-1}_{\Lambda} \Omega_{\Lambda}$ is reflection positive. By recalling \eqref{VacCDW}, we have
$\mathcal{U}_{\Lambda}^{-1} \Omega_{\Lambda}= \Omega_{\Lambda}^{\rm CDW} \otimes \Omega_{\Lambda} ^{\rm B}$, which implies that $ \Omega_{\Lambda}^{\rm CDW} \otimes \Omega_{\Lambda} ^{\rm B}$ is reflection positive. Similary, we can show that $\Omega_{\Lambda}^{\rm CDW}$ is reflection positive as well. 
Therefore, the proof of Example \ref{CDWVex} in Section \ref{RefPosiBasic} is complete.
\end{Ex}

Set $P_{\Omega_{\LL}^{\rm B}}= |\Omega_{\LL}^{\rm B}\ra\la \Omega_{\LL}^{\rm B}|$.
We introduce  an isometric  linear operator $\iota_{\Lambda}: \THL^{\rm F}(q) \to \THL(q)$ given by 
\be
\iota_{\Lambda}\vphi=\vphi\otimes P_{\Omega_{\LL}^{\rm B}} , \ \  \  \vphi\in \THL^{\rm F}(q).
\ee
Through the identification of $\THL^{\rm F}(q)$ with $\iota_{\Lambda} \THL^{\rm F}(q)$, we can consider $\THL^{\rm F}(q)$ as a closed subspace of $\THL(q)$.
Let $\tilde{s}_{\Lambda}$ be the orthogonal projection from $\THL(q)$ to $\THL^{\rm F}(q)$. 
Note that $\tilde{s}_{\Lambda}$ can be expressed as 
\be
\tilde{s}_{\Lambda}=\mathcal{L}\Big(1_q\otimes P_{\Omega_{\LL}^{\rm B}} \Big)\mathcal{R}\Big(1_q\otimes P_{\Omega_{\LL}^{\rm B}}\Big), \label{sExLR}
\ee
where $1_q$ is the identity operator on $\AFock_{\LL}(q)$.

\begin{Prop}\label{PropsL}
For any $q\in \mathbb{I}_{\Lambda}$, we have the following:
\begin{itemize}
\item[\rm (i)] $\tilde{s}_{\Lambda} \tCone(q)=\tCone^{\rm F}(q)$. More precisely, $
\tilde{s}_{\Lambda} \tCone(q)=\tCone^{\rm F}(q)\otimes P_{\Omega_{\LL}^{\rm B}}
$.
\item[\rm (ii)] $\tilde{s}_{\Lambda} \unrhd 0$ w.r.t. $\tCone(q)$.
\end{itemize}
\end{Prop}
\begin{proof}
By using \eqref{sExLR}, we readily confirm the assertions.
\end{proof}

By employing similar arguments as mentioned above, we can consider $\THL^{\rm F}$ as a closed subspace of $\THL$. Let $\tilde{S}_{\Lambda}$ denote the orthogonal projection from $\THL$ onto $\THL^{\rm F}$.
The following expression can be occasionally advantageous:
\be
\tilde{S}_{\Lambda}={\bs L} \Big(1 \otimes P_{\Omega_{\LL}^{\rm B}} \Big){\bs R}\Big(1 \otimes P_{\Omega_{\LL}^{\rm B}}\Big).\label{SLLR}
\ee 

\begin{Prop}
We have the following:
\begin{itemize}
\item[\rm (i)] $\tilde{S}_{\Lambda} \tCone=\tCone^{\rm F}$. More precisely, $
\tilde{S}_{\Lambda} \tCone=\tCone^{\rm F}\otimes P_{\Omega_{\LL}^{\rm B}}
$.
\item[\rm (ii)] $\tilde{S}_{\Lambda} \unrhd 0$ w.r.t. $\tCone$.
\end{itemize}
\end{Prop}
\begin{proof}
The validity of (i) can be derived from (i) of Proposition \ref{PropsL} and \eqref{DirectSumP}. By employing \eqref{SLLR}, we can easily verify (ii).
\end{proof}

\subsubsection*{\it Proof of Theorem \ref{RPConeBasic}}
We will demonstrate the proof of the assertions regarding the fermion-phonon system. The same can be established for the fermionic system in a similar manner.

Set $\tilde{\mathfrak{M}}_{\LL}=\{A \in \mathscr{L}(\Fock_{\LL}) : [A, \hat{Q}^{(L)}_{\Lambda}]=0\}$.
It is worth noting that $\mathfrak{M}_{\LL}=\mathcal{U}_{\Lambda}^* \tilde{\mathfrak{M}}_{\LL} \mathcal{U}_{\LL}$. Additionally, we have $\tilde{\mathfrak{M}}_{\LL}\otimes 1=\mathfrak{N}_{\Lambda, L}$. Let $\tilde{\vphi}$ be an arbitrary vector in $\tCone$. By applying Corollary \ref{ExpPP}, we obtain the inequality:
\be
\big\la\tilde{\vphi}|\tilde{A} \otimes \vartheta \tilde{A} \vartheta^{-1} \tilde{\vphi}\ra\ge 0,\quad  \tilde{A} \in \tilde{\mathfrak{M}}_{\LL}.
\ee
Especially,  it holds that 
\be
 \big\la \tilde{\vphi}| \delta \hn_{i_m} \delta \hn_{i_{m-1}} \cdots  \delta \hn_{i_1} \delta \hn_{-i_1-1}
 \delta \hn_{-i_2-1} \cdots \delta \hn_{-i_m-1} \tilde{\vphi} \big\ra \ge  0
\ee
for $i_1, \dots, i_m\in \LL$. By setting
\be
\vphi=\mathcal{U}_{\Lambda}^*\tilde{\vphi},\ \ A=\mathcal{U}_{\Lambda}^* \tilde{A} \mathcal{U}_{\LL},\ \ 
\tau=\mathcal{U}_{\LR}^* \vartheta \mathcal{U}_{\LL},
\ee
we get \eqref{VecPP1}. In addition, we  obtain \eqref{VecPP2} by using the property $\mathcal{U}_{\Lambda}^* \delta \hn_j\mathcal{U}_{\Lambda}=(-1)^j\delta \hn_j$ for all $j\in \Lambda$. \qed

\subsection{Infinite chain}\label{SectInfiniteCh}

In this subsection, our focus is on the infinite system.
Firstly, it is worth noting that both $\hat{Q}^{(L)}_{\BbbZ}$ and $\hat{Q}^{(R)}_{\BbbZ}$ are well-defined self-adjoint operators on $\AFock_{\BbbZ_-}$ and $\AFock_{\BbbZ_+}$, respectively.
Here, $\BbbZ_-=\{j\in \BbbZ : j< 0 \}$ and $\BbbZ_+=\{j\in \BbbZ : j \ge 0\}$.
Therefore, $\AFock_{ \BbbZ_-}(q)\ (q\in \mathbb{I}_{\BbbZ}=\BbbZ)$ can be defined as 
\be
\AFock_{\BbbZ_-}(q)=\{\vphi \in \Fock_{\BbbZ_-} : \hat{Q}^{(L)}_{\BbbZ} \vphi=q\vphi\}.
\ee 
In addition, $\mathfrak{F}_{\BbbZ_-}$,  $\mathfrak{F}_{\BbbZ_-}(q)$  and $\BFock_{\BbbZ_-}$can be  well-defined,   and it holds that 
\be
\mathfrak{F}_{\BbbZ_-}=\AFock_{\BbbZ_-}\otimes \BFock_{\BbbZ_-},\ \ 
\mathfrak{F}_{\BbbZ_-}(q)=\AFock_{\BbbZ_-}(q)\otimes \BFock_{\BbbZ_-}.
\ee
From these facts, it is natural to define $\tilde{\h}_{\BbbZ}$ as 
\be
\tilde{\h}_{\BbbZ}=\bigoplus_{q\in \mathbb{I}_{\BbbZ}}
\tilde{\h}_{\BbbZ}(q),\ \ \tilde{\h}_{\BbbZ}(q)=
\mathscr{L}^2(\mathfrak{F}_{\BbbZ_-}(q)).
\ee
In this setting, we define
\be
\tilde{\Cone}_{\rm r, \BbbZ}=\bigoplus_{q\in \mathbb{I}_{\BbbZ}}
\tilde{\Cone}_{\rm r, \BbbZ}(q),\ \ \tilde{\Cone}_{\rm r, \BbbZ}(q)=
\mathscr{L}_+^2(\mathfrak{F}_{\BbbZ_-}(q)).
\ee
The involution associated with $\tilde{\Cone}_{\rm r, \BbbZ}$ is given by 
\be
J=\bigoplus_{q\in \mathbb{I}_{\BbbZ}} J_q,
\ee
where $J_q$ is the involution associated with $\tilde{\Cone}_{\rm r, \BbbZ}(q)$: $J_q\xi=\xi^*\ (\xi\in \mathscr{L}^2(\Fock_{\BbbZ_-}(q)))$.
Let $\mathfrak{N}_{\BbbZ, L}(q)=\{\mathcal{L}(A) : A\in  \mathscr{L}(\Fock_{\BbbZ_-}(q))\}\ (q\in \mathbb{I}_{\BbbZ})$. Note that the involution $\vartheta$ in Section \ref{DefConeRP} can be defined even for $\Lambda=\BbbZ$. According to  Proposition \ref{PPI2}, the quadruple $\{
\mathfrak{N}_{\BbbZ, L}, \tilde{\h}_{\BbbZ}(q), J_q, \tilde{\Cone}_{\rm r, \BbbZ}(q)
\}$  is a standard form associated with $\{\Fock_{\BbbZ_-}(q), \Fock_{\BbbZ_+}(q), \vartheta\}$.
Furthermore, we have
\be
\{
\mathfrak{N}_{\BbbZ, L}, \tilde{\h}_{\BbbZ}, J, \tilde{\Cone}_{\rm r, \BbbZ}
\}=\bigoplus_{q\in \mathbb{I}_{\BbbZ}} \{
\mathfrak{N}_{\BbbZ, L}, \tilde{\h}_{\BbbZ}(q), J_q, \tilde{\Cone}_{\rm r, \BbbZ}(q)
\},
\ee
where $\mathfrak{N}_{\BbbZ, L}=\bigoplus_{q\in \mathbb{I}_{\BbbZ}} \mathfrak{N}_{\BbbZ, L}(q)$.
This can be regarded as the standard form for describing the infinite system.

Compared to the representation on $\h_{\Lambda}$, the primary benefit of utilizing the representation on $\tilde{\h}_{\Lambda}$ is that it allows for the explicit construction of various mathematical objects related to the infinite chain. For instance, we can concretely define and construct objects such as $\tilde{\h}_{\BbbZ}$, $\tilde{\Cone}_{\rm r, \BbbZ}$, and so on.

Taking the above definitions into account, we set
\be
\tilde{\h}_{\BbbZ}^{\rm F}=\bigoplus_{q\in \mathbb{I}_{\BbbZ}}
\tilde{\h}_{\BbbZ}^{\rm F}(q),\ \ \tilde{\h}^{\rm F}_{\BbbZ}(q)=
\mathscr{L}^2(\AFock_{\BbbZ_-}(q))
\ee
and 
\be
\tilde{\Cone}_{\rm r, \BbbZ}^{\rm F}=\bigoplus_{q\in \mathbb{I}_{\BbbZ}}
\tilde{\Cone}_{\rm r, \BbbZ}^{\rm F}(q),\ \ \tilde{\Cone}^{\rm F}_{\rm r, \BbbZ}(q)=
\mathscr{L}_+^2(\AFock_{\BbbZ_-}(q)).
\ee
It holds that 
\be
\tilde{\h}_{\BbbZ}=\tilde{\h}_{\BbbZ}^{\rm F} \otimes \tilde{\h}_{\BbbZ}^{\rm B},
 \ \ \ \tilde{\h}_{\BbbZ}(q)=\tilde{\h}_{\BbbZ}^{\rm F}(q) \otimes \tilde{\h}_{\BbbZ}^{\rm B},
\ee
where
\be
 \tilde{\h}_{\BbbZ}^{\rm B}=\mathscr{L}^2(\BFock_{\BbbZ_-}).
\ee
Set  $P_{\Omega_{\BbbZ_-}^{\rm B}}=|\Omega^{\rm B}_{\BbbZ_-}\ra \la \Omega_{\BbbZ_-}^{\rm B}| $, where $\Omega_{\BbbZ_-}^{\rm B}$ is the  Fock vacuum in $\BFock_{\BbbZ_-}$.
Through the identification of $\vphi\in \tilde{\h}_{\BbbZ}^{\rm F}$ with $\vphi\otimes P{\Omega_{\BbbZ_-}^{\rm B}} \in \tilde{\h}_{\BbbZ}$, we can consider $\tilde{\h}_{\BbbZ}^{\rm F}$ as a closed subspace of $\tilde{\h}_{\BbbZ}$.
 The orthogonal projection, $\tilde{S}_{\BbbZ}$,  from $\tilde{\h}_{\BbbZ}$ to $\tilde{\h}_{\BbbZ}^{\rm F}$ is given by 
\be
\tilde{S}_{\BbbZ}={\bs L} \Big(1 \otimes P_{\Omega_{\BbbZ_-}^{\rm B}} \Big){\bs R}\Big(1 \otimes P_{\Omega_{\BbbZ_-}^{\rm B}}\Big).
\ee

Recall here the definition of $\Omega_{\Lambda}$, i.e., \eqref{DefOme}.
For $\Lambda, \Lambda\rq{}\in \mathbb{F}$ with 
 $\Lambda\subseteq \Lambda\rq{}$,  we find that
\be
\Omega_{\Lambda_L\rq{}\setminus \Lambda_L}\otimes \IL\subseteq \ILd,
\ee
where $\Omega_{\Lambda_L\rq{}\setminus \Lambda_L} \otimes \IL
=\{\Omega_{\Lambda_L\rq{}\setminus \Lambda_L} \otimes \vphi : \vphi \in \IL\}
$. 
Here, we used the following identification:
\be
\Omega_{\LL\rq{}\setminus \LL} \otimes \chi^{\rm F} \otimes \chi^{\rm B}=\Big(\chi^{\rm F} \otimes \Omega^{\rm F}_{\LL\rq{}\setminus \LL} \Big) \otimes \Big(\chi^{\rm B} \otimes \Omega^{\rm B}_{\LL\rq{}\setminus \LL}\Big),\ \ \chi^{\rm F} \in \AFock_{\LL},\ \ \chi^{\rm B} \in \BFock_{\LL},
\ee
which implies that 
\be
\Omega_{\LL\rq{}\setminus \LL} \otimes \IL=\Big[
\AFock_{\LL} \otimes \Omega^{\rm F}_{\LL\rq{}\setminus \LL}
\Big]\otimes \Big[
\BFock_{\LL} \otimes \Omega^{\rm B}_{\LL\rq{}\setminus \LL}
\Big].
\ee
Due to the inclusion $\mathbb{I}_{\Lambda} \subseteq \mathbb{I}_{\Lambda\rq{}}$, for every  $q\in \mathbb{I}_{\Lambda}$, we have 
\be
\Omega_{\Lambda_L\rq{}\setminus \Lambda_L}\otimes \IL(q) \subseteq \ILd(q).
\ee
Taking this into consideration, let $\tau_{\Lambda\rq{} \Lambda}: \IL(q) \to \ILd(q)$
be an isometric   linear mapping  given by 
\be
\tau_{\Lambda\rq{}\Lambda} \vphi=\Omega_{\Lambda_L\rq{}\setminus \Lambda_L}\otimes \vphi,\ \ \vphi \in 
\IL(q).
\ee
Using this, we can construct an isometric  linear mapping  $\tilde{U}_{\Lambda\rq{}\Lambda} : \tilde{\mathfrak{H}}_{\Lambda} \to \tilde{\mathfrak{H}}_{\Lambda\rq{}}$ as follows: 
\begin{align}
\tilde{U}_{\Lambda\rq{}\Lambda} \bigoplus_{q\in \mathbb{I}_{\Lambda}} \vphi_q\otimes \vartheta \psi_q=
\bigoplus_{q\in \mathbb{I}_{\Lambda\rq{}}} \tilde{\vphi}_q \otimes \vartheta \tilde{\psi}_q,
\end{align}
where
\begin{align}
\tilde{\vphi}_q=
\begin{cases}
 \tau_{\Lambda\rq{}\Lambda} \vphi_q &  \mbox{ for $q\in \mathbb{I}_{\Lambda}$}\\
 0 & \mbox{ for $q\in \mathbb{I}_{\Lambda\rq{}} \setminus \mathbb{I}_{\Lambda}$}.
\end{cases}
\end{align}

For $\Lambda, \Lambda', \Lambda''\in \mathbb{F}$ with $\Lambda\subseteq \Lambda'\subseteq \Lambda''$, we can easily verify that
\be
\tau_{\Lambda''\Lambda'}\tau_{\Lambda'\Lambda}=\tau_{\Lambda''\Lambda}.
\ee
Using this relation, we also have
\be
\tilde{U}_{\Lambda\rq{}\rq{}\Lambda\rq{}}\tilde{U}_{\Lambda\rq{}\Lambda}=\tilde{U}_{\Lambda\rq{}\rq{}\Lambda}.
\ee
Therefore, by considering $\mathbb{F}$ as a directed set with respect to the inclusion relation, we can define the following inductive limit \cite{KaRi}:
$
\ilim \tilde{\h}_{\Lambda}.$ 
Trivially, we have the following identification:
\be
\tilde{\mathfrak{H}}_{\BbbZ}=\ilim \tilde{\h}_{\Lambda}.
\ee
Similarly, it holds that 
\be
\tilde{\mathfrak{H}}_{\BbbZ}^{\rm F}=\ilim \tilde{\h}_{\Lambda}^{\rm F}.
\ee

According to the reference \cite{Bratteli1997}, the following result holds:
\begin{Prop}\label{PropIndLim}
For each $\Lambda\in \mathbb{F}$, there exists an isometric linear mapping $\tilde{U}_{\Lambda}$ from $\THL$ into 
$\tilde{\mathfrak{H}}_{\BbbZ}$ satisfying the following:
\begin{itemize}
\item[\rm (i)] If $\Lambda\subseteq \Lambda\rq{}$, then $\tilde{U}_{\Lambda\rq{}}\tilde{U}_{\Lambda\rq{}\Lambda}=\tilde{U}_{\Lambda}$.
\item[\rm (ii)] $\bigcup_{\Lambda \in \mathbb{F}} \tilde{U}_{\Lambda}\tilde{\mathfrak{H}}_{\Lambda}$ is dense in $\tilde{\mathfrak{H}}_{\BbbZ}$.
\end{itemize}
Furthermore, the Hilbert space $\tilde{\mathfrak{H}}_{\BbbZ}$ and the net of isometric linear mappings $ \{\tilde{U}_{\Lambda} : \Lambda\in \mathbb{F}\}$    are uniquely determined,   up to unitary equivalence. 

Similar statements hold true for $\tilde{\mathfrak{H}}_{\BbbZ}^{\rm F}$ and the corresponding net of isometric linear mappings.
\end{Prop}

Now we are ready to prove Proposition \ref{HilInd2}.
\subsubsection*{\it Proof of Proposition \ref{HilInd2}}
By setting
$
\iota_{\Lambda\rq{}\Lambda}=\mathcal{U}_{\Lambda\rq{}}^{-1} \tilde{U}_{\Lambda\rq{}\Lambda} \mathcal{U}_{\Lambda},
$
we can observe that the following diagram is commutative:
\be
\begin{tikzcd}
  \h_{\Lambda}   \ar[r, "\iota_{\Lambda\rq{}\Lambda}"] \arrow[d, "\mathcal{U}_{\Lambda}"]& \h_{\Lambda\rq{}} \arrow[d, "\mathcal{U}_{\Lambda\rq{}}"]  \\
  \THL \arrow[r, "\tilde{U}_{\Lambda\rq{}\Lambda}"]& \THLd
\end{tikzcd}
\ee
Hence, Proposition \ref{HilInd2} immediately follows from Proposition \ref{PropIndLim}. \qed
\bigskip

In the following discussion, we identify $\tilde{\h}_{\Lambda}$ with $\tilde{U}_{\Lambda}\tilde{\h}_{\Lambda}$. Consequently, $\tilde{\h}_{\Lambda}$ can be viewed as a closed subspace of $\tilde{\h}_{\BbbZ}$. We represent the orthogonal projection from $\tilde{\h}_{\BbbZ}$ to $\tilde{\h}_{\Lambda}$ as $\tilde{P}_{\Lambda}$. By employing similar reasoning as mentioned above, given that $\Lambda\subseteq \Lambda'$, we can regard $\tilde{\h}_{\Lambda}$ as a subspace of $\tilde{\h}_{\Lambda'}$. The orthogonal projection from $\tilde{\h}_{\Lambda'}$ to $\tilde{\h}_{\Lambda}$ is denoted as $\tilde{P}_{\Lambda\Lambda'}$.

Regarding the fermion part, we can also consider $\THL^{\rm F}$ as a subspace of $\tilde{\mathfrak{H}}_{\BbbZ}^{\rm F}$. We use $\tilde{P}_{\Lambda}^{\rm F}$ to represent the orthogonal projection from $\tilde{\mathfrak{H}}_{\BbbZ}^{\rm F}$ to $\THL^{\rm F}$. Similarly, the orthogonal projection from $\THLd^{\rm F}$ to $\THL^{\rm F}$ is denoted as $\tilde{P}_{\Lambda \Lambda'}^{\rm F}$.

By construction, there is a unique unitary operator $\mathcal{U}_{\BbbZ}$ from $\h_{\BbbZ}$ onto $\tilde{\h}_{\BbbZ}$ such that $\mathcal{U}_{\BbbZ} \restriction \h_{\Lambda}=\mathcal{U}_{\Lambda}$ for each $\Lambda \in \mathbb{F}$.
\begin{define}\upshape
The {\it reflection positivity on $\h_{\BbbZ}$} is defined by $\Cone_{\rm r, \BbbZ}=\mathcal{U}_{\BbbZ}^{-1} \tilde{\Cone}_{\rm r, \BbbZ}$. Similarly, we can define the reflection positivity on $\h^{\rm F}_{\BbbZ}$.
\end{define}

The following lemma is an immediate consequence of the definitions.
\begin{lemm}
Let $\Lambda, \Lambda\rq{}\in \mathbb{F}$ with $\Lambda\subseteq \Lambda\rq{}$. Then we have the following:
\begin{itemize}
\item[\rm (i)]
$
\tilde{P}_{\Lambda \Lambda\rq{}}\tilde{P}_{\Lambda\rq{}}=\tilde{P}_{\Lambda}.
$
\item[\rm (ii)] $\tilde{P}_{\Lambda} \le \tilde{P}_{\Lambda\rq{}}$. In addition, the net $\{P_{\Lambda} : \Lambda\in \mathbb{F}\}$ converges $\sigma$-strongly to the identity $1$.
\end{itemize}
Similar statements hold true for $\tilde{P}_{\Lambda}^{\rm F}$ and $\tilde{P}_{\Lambda \Lambda\rq{}}^{\rm F}$.
\end{lemm}

\begin{coro}Let $\Lambda, \Lambda\rq{}\in \mathbb{F}$ with $\Lambda\subseteq \Lambda\rq{}$.
The following diagram is commutative:
\be
\begin{tikzcd}
  \THL \arrow[d, "\tilde{S}_{\Lambda}"]  &\arrow[l, "\tilde{P}_{\Lambda\Lambda\rq{}}"] \THLd \arrow[d, "\tilde{S}_{\Lambda\rq{}}"]  &\arrow[l, "\tilde{P}_{\Lambda\rq{}}"]  \tilde{\h}_{ \BbbZ} \arrow[d, "\tilde{S}_{\BbbZ}"]\\
  \THL^{\rm F} &\arrow[l, "\tilde{P}_{\Lambda\Lambda\rq{}}^{\rm F}"] \THLd^{\rm F} &\arrow[l, "\tilde{P}_{\Lambda\rq{}}^{\rm F}"] \tilde{\h}_{\BbbZ}^{\rm F} 
\end{tikzcd}
\ee
\end{coro}

For $\Lambda, \Lambda' \in \mathbb{F}$ with $\Lambda \subseteq \Lambda'$, we define $\pi_{\Lambda\Lambda'}$ as the orthogonal projection from $\ell^2(\Lambda'_L)$ to $\ell^2(\Lambda_L)$. Then, the orthogonal projection from $\AFock_{\Lambda'L}$ to $\AFock_{\Lambda_L}$ is given by
\be
p^{\rm F}_{\Lambda\Lambda\rq{}} = \bigoplus_{n=0}^{|\Lambda|} \otimes^n \pi_{\Lambda\Lambda\rq{}}\restriction \wedge^n \ell^2(\Lambda'_L).
\ee
Similarly, the orthogonal projection from $\BFock_{\Lambda_L'}$ to $\BFock_{\Lambda_L}$ is given by 
\be
p^{\rm B}_{\Lambda\Lambda' } = \bigoplus_{n=0}^{\infty} \otimes^n \pi_{\Lambda\Lambda\rq{}}\restriction \otimes_{\rm s}^n \ell^2(\Lambda'_L).
\ee
By using the orthogonal projections $p^{\rm F}_{\Lambda\Lambda\rq{}}$ and $p^{\rm B}_{\Lambda\Lambda\rq{}}$, we can express $\tilde{P}_{\Lambda\Lambda\rq{}}$ as 
\be
\tilde{P}_{\Lambda\Lambda'} =\bigoplus_{q\in \mathbb{I}_{\Lambda\rq{}}} \mathcal{L}\Big(p^{(q)}_{\Lambda\Lambda\rq{}}\Big) \mathcal{R}\Big(p_{\Lambda\Lambda\rq{}}^{(q)}\Big), \label{ExPLL}
\ee
where 
\begin{align}
p_{\Lambda\Lambda\rq{}}^{(q)} = \begin{cases}
p^{\rm F}_{\Lambda\Lambda\rq{}} \otimes p_{\Lambda\Lambda'}^{\rm B} & \mbox{for $q\in \mathbb{I}_{\Lambda}$}\\
0  & \mbox{for $q\in \mathbb{I}_{\Lambda\rq{}}\setminus \mathbb{I}_{\Lambda}$}.
\end{cases}
\end{align}
\begin{Prop}\label{ConeConect}
Let $\Lambda, \Lambda\rq{}\in \mathbb{F}$. If $\Lambda\subseteq\Lambda\rq{}$, then  we have the following:
\begin{itemize}
\item[\rm (i)] $\tCone \subseteq \tConed$.
\item[\rm (ii)]
$\tilde{P}_{\Lambda\Lambda\rq{}}\tConed = \tCone$.
\item[\rm (iii)] $\tilde{P}_{\Lambda\Lambda\rq{}} \unrhd 0$ w.r.t. $\tConed$.
\end{itemize}
Similar statements hold true for $\tCone^{\rm F}$ and $\tilde{P}^{\rm F}_{\Lambda \Lambda\rq{}}$.
\end{Prop}
\begin{proof}
(i) Let $\xi\in \tCone$. Thus, we can express $\xi$ as
\be
\xi=\bigoplus_{q\in \mathbb{I}_{\Lambda}} \xi_q,\ \ \ \xi_q\in \tCone(q).
\ee 
Because $\IL(q)$ is a subspace of $\ILd(q)$, we have 
$\xi_q\in \tConed(q)$, provided that $q\in \mathbb{I}_{\Lambda}$.
Set 
\be
\tilde{\xi}_q=\begin{cases}
\xi_q & \mbox{for $q\in \mathbb{I}_{\Lambda}$}\\
0 & \mbox{for $q\in \mathbb{I}_{\Lambda\rq{}}\setminus \mathbb{I}_{\Lambda}$}
\end{cases}
\ee
and $\tilde{\xi}=\bigoplus_{q\in \mathbb{I}_{\Lambda\rq{}}} \tilde{\xi}_q$.
The vector $\xi$ can be identified with $\tilde{\xi}$ which belongs to $\tConed$.

(ii) Let $\xi=\bigoplus_{q\in \mathbb{I}_{\Lambda\rq{}}} \xi_q\in \tConed$. By using the formula \eqref{ExPLL}, we see that 
\be
\tilde{P}_{\Lambda\Lambda\rq{}}\xi=\bigoplus_{q\in \mathbb{I}_{\Lambda\rq{}}} p_{\Lambda\Lambda\rq{}}^{(q)} \xi_q p_{\Lambda\Lambda\rq{}}^{(q)} \in \tCone,  
\ee
which implies that $\tilde{P}_{\Lambda\Lambda\rq{}} \tCone\subseteq \tConed$. The converse inclusion is straightforward.

The assertion (iii) follows immediately from (i) and (ii).

By using similar arguments as presented above, we can prove the corresponding assertions for  $\tCone^{\rm F}$ and $\tilde{P}^{\rm F}_{\Lambda \Lambda\rq{}}$.
\end{proof}

Next, we derive a useful expression for $\tilde{P}_{\Lambda}$. 
Given a $\Lambda\in \mathbb{F}$,  let $\pi_{\Lambda}$ be the orthogonal projection from $\ell^2(\BbbZ_-)$ to $\ell^2(\Lambda_L)$. By setting 
\be
p^{\rm F}_{\Lambda} =\bigoplus_{n=0}^{\infty} \otimes^n \pi_{\Lambda}\restriction \wedge^n \ell(\BbbZ_-),\ \ p^{\rm B}_{\Lambda} =\bigoplus_{n=0}^{\infty} \otimes^n \pi_{\Lambda}\restriction \otimes_{\rm s}^n \ell(\BbbZ_-), 
\ee
we obtain the following formula:
\be
\tilde{P}_{\Lambda} =\bigoplus_{q \in \mathbb{I}_{\BbbZ}} \mathcal{L}\Big(p^{(q)}_{\Lambda}\Big) \mathcal{R}\Big(p_{\Lambda}^{(q)}\Big), \label{ExPL}
\ee
where 
\begin{align}
p_{\Lambda}^{(q)} = \begin{cases}
p_{\Lambda}^{\rm F}\otimes p_{\Lambda}^{\rm B} & \mbox{for $q\in \mathbb{I}_{\Lambda}$}\\
0  & \mbox{for $q\in \mathbb{I}_{\BbbZ}\setminus \mathbb{I}_{\Lambda}$}.
\end{cases}
\end{align}
The formulas \eqref{ExPLL} and \eqref{ExPL} will play a crucial role in the subsequent arguments.

\begin{Prop} For each $\Lambda \in \mathbb{F}$, 
we have the following:
\begin{itemize}
\item[\rm (i)] $\tilde{P}_{\Lambda} \tilde{\Cone}_{\rm r, \BbbZ}=\tilde{\Cone}_{\rm r, \Lambda}$.
\item[\rm (ii)] $\tilde{P}_{\Lambda} \unrhd 0$ w.r.t. $\tilde{\Cone}_{\rm r, \BbbZ}$.
\end{itemize}
Similar statements hold true for $\tilde{P}_{\Lambda}^{\rm F},\ \tCone^{\rm F}$ and $\tilde{\Cone}^{\rm F}_{\rm r, \BbbZ}$.
\end{Prop}
\begin{proof}
(i) By using arguments similar to those in the proof of Proposition \ref{ConeConect}, we can show (i).

(ii) can be derived by applying Proposition \ref{PPI3} and utilizing the expression \eqref{ExPL}.

By employing similar reasoning as presented above, we can establish the corresponding results for $\tilde{P}_{\Lambda}^{\rm F},\ \tCone^{\rm F}$, and $\tilde{\Cone}^{\rm F}_{\rm r, \BbbZ}$.
\end{proof}

Basic  properties of the  reflection positivity  can be summarized as follows.
\begin{Thm}\label{StrongPropPP} Let $\Lambda, \Lambda\rq{}\in \mathbb{F}$ with $\Lambda\subseteq \Lambda\rq{}$.
The following diagram is commutative:
\be
\begin{tikzcd}
\rCone \arrow[d, "\mathcal{U}_{\Lambda}"]  &\arrow[l, "P_{\Lambda\Lambda\rq{}}"] \rConed \arrow[d, "\mathcal{U}_{\Lambda\rq{}}"]  &\arrow[l, "P_{\Lambda\rq{}}"]  \Cone_{\rm r, \BbbZ} \arrow[d, "\mathcal{U}_{\BbbZ}"]\\
  \tCone \arrow[d, "\tilde{S}_{\Lambda}"]  &\arrow[l, "\tilde{P}_{\Lambda\Lambda\rq{}}"] \tConed \arrow[d, "\tilde{S}_{\Lambda\rq{}}"]  &\arrow[l, "\tilde{P}_{\Lambda\rq{}}"]  \tilde{\Cone}_{\rm r, \BbbZ} \arrow[d, "\tilde{S}_{\BbbZ}"]\\
  \tCone^{\rm F} &\arrow[l, "\tilde{P}_{\Lambda\Lambda\rq{}}^{\rm F}"] \tConed^{\rm F} &\arrow[l, "\tilde{P}_{\Lambda\rq{}}^{\rm F}"] \tilde{\Cone}_{\rm r, \BbbZ}^{\rm F} 
\end{tikzcd}
\ee
Furthermore, every projection in the diagram is positivity preserving.
\end{Thm}

\subsection{Proof of Theorem \ref{GSGProp} assuming Theorem \ref{FirstThm}}\label{PfThmInfSys}
Below, we provide a detailed proof regarding the assertions of the theorems concerning the fermion-phonon system. The same can be proven for the fermionic system as well.

1. 
For each $\Lambda\in \mathbb{F}_{\rm o}$, let us consider the ground state $\psi_{\Lambda}$ of the finite system. As already demonstrated, $\psi_{\Lambda}$ exhibits positivity with respect to $\rCone$. By utilizing Theorem \ref{StrongPropPP}, we establish its positivity with respect to $\Cone_{\rm r, \BbbZ}$. Hence, we have $\big\la A\otimes \tau_{\BbbZ} A\tau_{\BbbZ}^{-1}\big\ra_{\Lambda}\ge 0$ for all $A\in \mathfrak{M}_{\BbbZ_-}$. Taking the limit $\Lambda\uparrow \BbbZ$, we conclude the validity of \eqref{InfiniteGS}. It is important to note that we utilize the embedding property of $\rCone$ into $\Cone_{\rm r, \BbbZ}$ in this part.

2. 
  We set 
\be
\la\!\la A\ra\!\ra_{\BbbZ}=\big\la \mathcal{U}_{\BbbZ}^*A\,  \mathcal{U}_{\BbbZ}\big\ra_{\BbbZ}.
\ee
Trivially, $\la\!\la\cdot \ra\!\ra_{\BbbZ}$ is a state on $\mathscr{L}(\tilde{\h}_{\BbbZ})$.
Assuming that $\la \cdot \ra_{\BbbZ}$ is a normal state on $\mathfrak{M}_{\BbbZ_-}\otimes 1$, the application of \cite[Theorem 2.5.31]{Bratteli1987} leads us to the conclusion that there exists a unique $\psi_{\BbbZ}\in \h_{\BbbZ}$ that fulfills properties (i) and (ii).
 Hence, $\la\!\la A\ra\!\ra_{\BbbZ}=\big\la \tilde{\psi}_{\BbbZ}|A\tilde{\psi}_{\BbbZ}\big\ra_{\BbbZ}$ holds, where $\tilde{\psi}_{\BbbZ}=\mathcal{U}_{\BbbZ} \psi_{\BbbZ}$.
It follows from (i) that $\tilde{\psi}_{\BbbZ} \ge 0$ w.r.t. $\tilde{\Cone}_{{\rm r}, \BbbZ}$.

Suppose further that $\la \cdot \ra_{\BbbZ}$ is faithful on $\mathfrak{M}_{\BbbZ_-}\otimes 1$. Then $\la\!\la \cdot \ra\!\ra_{\BbbZ}$ is faithful on $\tilde{\mathfrak{M}}_{\BbbZ_-}\otimes 1=\mathfrak{N}_{\BbbZ, L}$ as well. Take $\phi=\bigoplus_{q\in \mathbb{I}_{\BbbZ}} \phi_q\in \Fock_{\BbbZ_-}\setminus \{0\}$, arbitrarily and set $A=|\phi\ra\la\phi|$.
Because $\la\!\la \cdot \ra\!\ra_{\BbbZ}$ is faithful on $\tilde{\mathfrak{M}}_{\BbbZ_-} \otimes 1$, we have
\be
0< \la \tilde{\psi}_{\BbbZ}|{\bs L}(A)\tilde{\psi}_{\BbbZ}\ra=\la
\phi|\tilde{\psi}_{\BbbZ}^2\phi\ra.
\ee
Since $\phi$ is arbitrary, we conclude that $\tilde{\psi}_{\BbbZ}>0$, which implies that $\psi_{\BbbZ}>0$ w.r.t. $\Cone_{{\rm r}, \BbbZ}$. \qed

\section{Reflection  positivity preserving property of $e^{-\beta H_{\Lambda}}$}\label{Sect4}
\subsection{Statement of the result}
The purpose of this subsection is to prove the following:
\begin{Thm}\label{ThmPP}
Assume \hyperlink{A}{\bf (A)} and \hyperlink{B1}{\bf (B. 1)}.
For all   $\beta\ge 0$, we have $e^{-\beta H_{\Lambda}}\unrhd 0$ w.r.t. $\rCone$ and $e^{-\beta H_{\Lambda}^{\F}} \unrhd 0$ w.r.t. $\rCone^{\F}$.
\end{Thm} 

\subsubsection*{\it Proof of Theorem \ref{FirstThm} assuming Theorem \ref{ThmPP}}
The validity of statement (i) in Theorem \ref{FirstThm} directly follows from the application of Theorem \ref{ThmPP}.
Considering Proposition \ref{GSP} and the uniqueness of the ground state of $H_{\Lambda}$, it follows that $\psi_{\Lambda}$ must exhibit positivity with respect to $\rCone$. Consequently, by employing Theorem \ref{RPConeBasic}, we can promptly deduce  \eqref{GSRPP0} and \eqref{GSRPP}.
\qed

\subsection{The Lang--Firsov transformation}
In order to show Theorem  \ref{ThmPP}, we need some preliminaries.
We define
\begin{align}
T&=\sum_{j\in \Lambda}(-t)\big(
c_{j}^*
 c_{j+1}
+c_{j+1}^*c_{j}
\big),\label{EKinetic}\\
P&=\sum_{i, j\in \Lambda}U(i-j) \dhn_i\dhn_j,\\
I&=g \sum_{j\in \Lambda} \dhn_j(a_j+a_j^*),\\
K&= \omega \sum_{j\in \Lambda} a_j^* a_j\label{BKinetic}.
\end{align}
It can be easily shown that $H_{\Lambda}=T+P+I+K$.
For each $j\in \Lambda$, let us define:
\begin{align}
 \phi_j=\sqrt{\frac{1}{2\omega}}(a_j^*+a_j),\ \ \pi_j=\im 
 \sqrt{\frac{\omega}{2}}(a_j^*-a_j).
\end{align} 
Both $\phi_j$ and $\pi_j$ are essentially self-adjoint operators, and we denote their closures using the same symbols.
 In terms of $\phi_j$ and $\pi_j$, the operator $K$ can be represented as
 \be
 K=\frac{1}{2}\sum_{j\in \Lambda} (\pi_j^2+\omega^2 \phi_j^2).
 \ee
 Next, let
\begin{align}
L=-\im  \sqrt{2} \omega^{-3/2} g\sum_{j\in \Lambda} \dhn_j \pi_j.
\end{align} 
The operator $L$ is essentially antiself-adjoint. We also denote its closure by the
same symbol.
The transformation known as the {\it Lang--Firsov transformation} is a unitary operator defined by
\begin{align}
\mathcal{V}_{\Lambda}=e^{-\im \pi N{\mathrm{p}}/2}e^L, \label{DefUniV}
\end{align}
where $N_{\mathrm{p}}=\sum_{j\in\Lambda} a_j^*a_j$ \cite{Lang1963}.
 We readily confirm  the
following:
\begin{align}
\mathcal{V}_{\Lambda} c_j \mathcal{V}_{\Lambda}^{-1}&=e^{\im \alpha \phi_j} c_{j},
 \ \ \alpha=\sqrt{2} \omega^{-1/2} g, \label{LF1}\\
\mathcal{V}_{\Lambda} a_j\mathcal{V}_{\Lambda}^{-1}&=\im a_j-\frac{g}{\omega}\dhn_j. \label{LF2}
\end{align} 
Using  these formulas, we obtain:
\begin{lemm}

We have $\mathcal{V}_{\Lambda} H_{\Lambda} \mathcal{V}_{\Lambda}^{-1}=\mathsf{T}+\mathsf{P}+K-g^2 |\Lambda|/4\omega$, where 
\begin{align}
\mathsf{T}&= \sum_{j\in \Lambda} (-t)
\Big(e^{-\im  \alpha(\phi_j-\phi_{j+1})}
c_{j}^* c_{j+1}
+
e^{+\im  \alpha(\phi_j-\phi_{j+1})}
c_{j+1}^* c_{j}
\Big), \label{DefKineT}\\
\mathsf{P}&=\sum_{i, j\in \Lambda}U(i-j) \dhn_i\dhn_j.
\end{align} 
\end{lemm} 
\begin{proof}
 Due to \eqref{LF1} and \eqref{LF2}, the Coulomb interaction term is modified as 
\be
\sum_{i, j\in \Lambda} U_{\rm eff}(i-j) \delta \hn_i \delta \hn_j,
\ee
where
\be
U_{\rm eff}(i)=U(i)-\frac{g^2}{\omega} \delta_{i, 0}\ \ \forall i\in \Lambda.
\ee 
Here,  $\delta_{i, j}$ stands for  the Kronecker delta. However, because $n_i^2=n_i$ for all $i\in \Lambda$, 
we see that $(\delta n_i)^2=1/4$ for all $i\in \Lambda$. Hence, we conclude the desired assertion.
\end{proof}

\subsection{Expressions of the  Hamiltonian}
To demonstrate that the heat semigroup generated by the Hamiltonian preserves reflection positivity, this subsection is divided into three parts, sequentially transforming the Hamiltonian. By employing the representation of the resulting Hamiltonian, it becomes possible to prove Theorem \ref{ThmPP} .

\subsubsection{Expression  I}

We introduce a unitary operator defined as follows:
\be
\mathcal{W}_{\Lambda} =\mathcal{U}_{\Lambda} \mathcal{V}_{\Lambda},
\ee
where $\mathcal{U}_{\Lambda}$ and $\mathcal{V}_{\Lambda}$ are defined by \eqref{DefUniU} and \eqref{DefUniV}, respectively.
We initiate our discussion with the subsequent lemma.

\begin{lemm} Let $\tilde{H}_{\Lambda}=\mathcal{W}_{\Lambda}H_{\Lambda} \mathcal{W}_{\Lambda}^{-1}+g^2|\Lambda|/4\omega$. 
As a linear operator on $\tilde{\h}_{\Lambda}$, we have 
\begin{align}
\tilde{ H}_{\Lambda}=\mathbb{T}-\mathbb{W}+K,
\end{align} 
where
\begin{align}
\mathbb{T}&=
\sum_{j\in \Lambda_{\mathrm{e}}}\sum_{\vepsilon=\pm 1}
(-t) \Big(
e^{-\im \alpha(\phi_j-\phi_{j+\vepsilon })}
c_{j}^* c_{j+\vepsilon }^*
+
e^{+\im  \alpha(\phi_j-\phi_{j+\vepsilon })}
c_{j+\vepsilon} c_{j}
\Big), \label{PP}\\
\mathbb{W}&=\sum_{i, j\in
 \Lambda}W(i-j)\delta \hn_{i}\delta \hn_{j},\ \
W(j)=(-1)^{j+1} U(j).
\end{align} 
\end{lemm}
\begin{proof} Firstly, let us recall that $\mathsf{T}$ is defined by equation \eqref{DefKineT}. It is noteworthy that $\mathsf{T}$ can be represented as follows:
\begin{align}
\mathsf{T}
=\sum_{j\in \Lambda_{\mathrm{e}}}
 \sum_{\vepsilon=\pm 1}(-t)
\Big(
e^{-\im  \alpha (\phi_j-\phi_{j+\vepsilon})}c_{j}^*
 c_{j+\vepsilon }+e^{+\im \alpha (\phi_j-\phi_{j+\vepsilon})}
 c_{j+\vepsilon }^*c_j
\Big).
\end{align} 
By applying Lemma \ref{Uni}, we have
\begin{align}
\mathcal{U}_{\Lambda} \mathsf{T}\mathcal{U}_{\Lambda}^*&=\sum_{j\in \Lambda_e} \sum_{\vepsilon=\pm 1} 
(-t)\Big(e^{-\im  \alpha (\phi_j-\phi_{j+\vepsilon})} \mathcal{U}_{\Lambda}c_{j  }^* \mathcal{U}_{\Lambda}^*\mathcal{U}_{\Lambda}c_{j+\vepsilon}\mathcal{U}_{\Lambda}^*
 +e^{+\im  \alpha (\phi_j-\phi_{j+\vepsilon})}
 \mathcal{U}_{\Lambda}c_{j+\vepsilon}^* \mathcal{U}_{\Lambda}^*\mathcal{U}_{\Lambda}
 c_{j  } \mathcal{U}_{\Lambda}^* \Big)\no
 &=\mathbb{T}.
\end{align}
Likewise, employing Lemma \ref{Uni}, we find that 
$\mathcal{U}_{\Lambda}\delta \hn_j \mathcal{U}_{\Lambda}^*=(-1)^j \delta \hn_j$, hereby implying that
\be
\mathcal{U}_{\Lambda}\sum_{i, j\in
 \Lambda}U(i-j)\delta \hn_{i}\delta \hn_{j}\mathcal{U}_{\Lambda}^*=-\mathbb{W}.
 \ee
  \end{proof}

Taking  the identification (\ref{IdnLR}) into account, 
we have the following:
\begin{itemize}
\item {\it  Fermion-phonon interaction:} $\mathbb{T}$ can be written as 
\begin{align}
\mathbb{T}=&\mathbb{T}_L\otimes \one + \one \otimes \mathbb{T}_R+\mathbb{T}_{LR},
\end{align}
where
\begin{align}
\mathbb{T}_L=&\sum_{{j\in \Lambda_L\cap \Lambda_e}\atop{j\pm 1\in \Lambda_L}}
\sum_{\vepsilon=\pm 1}
(-t)\Big(e^{-\im  \alpha (\phi_j-\phi_{j+\vepsilon})}c_{j  }^* c_{j+\vepsilon  }^*
+e^{+\im  \alpha (\phi_j-\phi_{j+\vepsilon})}c_{j+\vepsilon  }c_{j  }\Big), \\
\mathbb{T}_R=&
 \sum_{{j\in \Lambda_R\cap \Lambda_e}\atop{j\pm 1\in \Lambda_R}}
\sum_{\vepsilon=\pm 1}
(-t) \Big(e^{-\im  \alpha (\phi_j-\phi_{j+\vepsilon})}c_{j  }^* c_{j+\vepsilon  }^*
+e^{+\im  \alpha (\phi_j-\phi_{j+\vepsilon})}c_{j+\vepsilon  }c_{j  }\Big).
 \end{align}
 If $\ell$ is even, the cross term $\mathbb{T}_{LR}$ can be expressed as follows:
 \begin{align}
\mathbb{T}_{LR}=& (-t)
\bigg\{\Big[ e^{+\im \alpha \phi_{-1}}
(-\one)^{\hat{N}^{(L)}}c_{-1  }^* \Big]\otimes\Big[ e^{-\im \alpha \phi_{0}}
 c_{0  }^*\Big]+\no
 & \hspace{7mm}+\Big[ e^{-\im \alpha \phi_{-1}}c_{-1}(-\one)^{\hat{N}^{(L)}} \Big]\otimes\Big[ e^{+\im \alpha \phi_{0}} c_{0  }\Big]
\bigg\}+\no
&+(-t)
\bigg\{\Big[ e^{-\im \alpha \phi_{-\ell}}
c_{-\ell  }^* (-\one)^{\hat{N}^{(L)}} \Big] \otimes 
 \Big[
 e^{+\im \alpha \phi_{\ell-1}}
 c_{\ell-1  }^*\Big]+\no
 &\hspace{7mm} + 
 \Big[ e^{+\im \alpha \phi_{-\ell}} (-\one)^{\hat{N}^{(L)}}c_{-\ell}
 \Big]
 \otimes \Big[
  e^{-\im \alpha \phi_{\ell-1}}
 c_{\ell-1}\Big]
\bigg\}, \label{ellEven}
\end{align} 
while if $\ell $ is odd, then 
\begin{align}
\mathbb{T}_{LR}=& (-t)
\bigg\{\Big[ e^{+\im \alpha \phi_{-1}}
(-\one)^{\hat{N}^{(L)}}c_{-1  }^* \Big]\otimes\Big[ e^{-\im \alpha \phi_{0}}
 c_{0  }^*\Big]+\no
 & \hspace{7mm}+\Big[ e^{-\im \alpha \phi_{-1}}c_{-1}(-\one)^{\hat{N}^{(L)}} \Big]\otimes\Big[ e^{+\im \alpha \phi_{0}} c_{0  }\Big]
\bigg\}+\no
&+ (-t)
\bigg\{
\Big[
e^{+\im \alpha \phi_{\ell}}
(-\one)^{\hat{N}^{(L)}}c_{-\ell  }^*  \Big]
\otimes 
\Big[e^{-\im \alpha \phi_{\ell-1}} c_{\ell -1}^*\Big]+\no
&\hspace{7mm}+
\Big[
e^{-\im \alpha \phi_{\ell}}
c_{-\ell   }  (-\one)^{\hat{N}^{(L)}} \Big]\otimes
\Big[
e^{+\im \alpha \phi_{\ell-1}}
 c_{\ell-1}\Big]
\bigg\}. \label{ellOdd}
\end{align} 
\begin{rem} \upshape
The difference between \eqref{ellEven} and \eqref{ellOdd} is significant, as illustrated in Lemma \ref{TransLR} and its corresponding proof.
\end{rem}
\item {\it  Coulomb interaction:} $\mathbb{W}$ can be expressed as 
\begin{align}
\mathbb{W}= & \mathbb{W}_L \otimes \one + \one \otimes \mathbb{W}_R+
\mathbb{W}_{LR},
\end{align}
where
\begin{align}
\mathbb{W}_L=&  \sum_{i, j\in \Lambda_L}W(i-j)
\delta \hn_i\delta \hn_{j},\\
\mathbb{W}_R=&  \sum_{i, j\in \Lambda_R}
W(i-j)\delta \hn_{i}\delta \hn_{j},\\
\mathbb{W}_{LR}=& 2 \sum_{i\in \Lambda_L, j\in \Lambda_R} W(i-j)\delta \hn_{i} \otimes \delta \hn_{j}.
\end{align} 
\item {\it Phonon energy: } We have \be
K=K_L\otimes \one +\one \otimes K_R,\ee where 
\begin{align}
K_L=\frac{1}{2}\sum_{j \in \Lambda_L}(\pi_j^2+\omega^2\phi_j^2),
\ \ \ K_R=\frac{1}{2}\sum_{j \in \Lambda_R}(\pi_j + \omega^2\phi_j^2).  \label{KLR}
\end{align} 
\end{itemize}
\subsubsection{Expression II}
Let us recall the definition of $b_j$ given by equation \eqref{Defb}. Assuming that $\ell$ is odd, we can rewritten $\mathbb{T}_L$ and $\mathbb{T}_{LR}$ as follows:
\begin{align}
\mathbb{T}_L=&
\sum_{{j\in \Lambda_L\cap \Lambda_e}\atop{j\pm 1\in \Lambda_L}}
\sum_{\vepsilon=\pm 1}
(-t)
\Big(e^{-\im \alpha (\phi_j-\phi_{j+\vepsilon})}
b_{j  }^* b_{j+\vepsilon  }^*+
e^{+\im  \alpha (\phi_j-\phi_{j+\vepsilon})}b_{j+\vepsilon  }b_{j  }\Big),
\\
\mathbb{T}_{LR}=&  (-t)
\Big\{(e^{+\im \alpha \phi_{-1}}
 b_{-1  }^*)\otimes 
 (e^{-\im \alpha \phi_0}c_{0}^*)+(e^{-\im \alpha \phi_{-1}}b_{-1  })\otimes (e^{+\im \alpha \phi_0}c_{0})
\no 
& +(
e^{+\im \alpha \phi_{-\ell}} b_{-\ell}^*)\otimes(e^{-\im \alpha \phi_{\ell-1}} c_{\ell-1}^*)+
 (e^{-\im \alpha \phi_{-\ell}} b_{-\ell})\otimes(e^{+\im \alpha \phi_{\ell-1}} c_{\ell-1}
 )\Big\}.
\end{align}

\begin{lemm}\label{TransLR}
Let $\vartheta$ be the involution defined by \eqref{DefTheta}.
Suppose that $\ell$ is an odd number.
We have the following:
\begin{itemize}
\item[\rm (i)]
 $
\mathbb{T}_R=\vartheta \mathbb{T}_L \vartheta^{-1}$.
\item[\rm (ii)] 
$
 \mathbb{W}_R=\vartheta \mathbb{W}_L \vartheta^{-1}$.
 \item[\rm (iii)] 
 \begin{align}
\mathbb{T}_{LR}=&(-t)
\Big\{(e^{+\im \alpha \phi_{-1}}b_{-1  }^*) \otimes (\vartheta e^{+\im \alpha \phi_{-1}} b_{-1  }^* \vartheta^{-1})
 +(e^{-\im \alpha \phi_{-1}}b_{-1  })\otimes (\vartheta e^{-\im \alpha \phi_{-1}} b_{-1  }\vartheta^{-1})+\no
 &+
 (e^{+\im \alpha \phi_{-\ell}}b_{-\ell  }^*) \otimes (\vartheta e^{+\im \alpha \phi_{-\ell}} b_{-\ell  }^* \vartheta^{-1})
 +(e^{-\im \alpha \phi_{-\ell}}b_{-\ell})\otimes (\vartheta e^{-\im \alpha \phi_{-\ell}}b_{-\ell  }\vartheta^{-1})
 \Big\}. \label{TLRTensor}
 \end{align}
 \item[\rm (iv)]
 $\displaystyle 
\mathbb{W}_{LR}=2 \sum_{i, j\in \Lambda_L}W(i+j+1) \delta \hn_{i} \otimes \vartheta\delta \hn_{j} \vartheta^{-1}.
$
\item[\rm (v)] $K_R=\vartheta K_L\vartheta^{-1}$.
\end{itemize}
\end{lemm} 
\begin{proof} (i)
We can rewrite $\mathbb{T}_L$ as 
\begin{align}
\mathbb{T}_L=\sum_{{j\in \Lambda_L\cap \Lambda_o}\atop{j\pm 1\in \Lambda_L}}
\sum_{\vepsilon=\pm 1}
(-t)\Big\{- \big(e^{+\im \alpha (\phi_j-\phi_{j+\vepsilon})}b_{j  }^* b_{j+\vepsilon  }^*+e^{-\im \alpha (\phi_j-\phi_{j+\vepsilon})}b_{j+\vepsilon  }b_{j  }\big)\Big\}.
\end{align}
It is important to note that the summation over $j$ in the given expression runs over all odd numbers in $\Lambda_L$. By utilizing \eqref{DefTheta} and the antilinearity property of $\vartheta$, we can derive the following:
\begin{align}
\vartheta \mathbb{T}_L\vartheta^{-1}=& \sum_{{j\in \Lambda_L\cap \Lambda_o}\atop{j\pm 1\in \Lambda_L}}
\sum_{\vepsilon=\pm 1}
(-t)\Big(e^{-\im \alpha (\phi_{r(j)}-\phi_{r(j)+\vepsilon})}c_{r(j)  }^* c_{r(j)+\vepsilon  }^*+
e^{+\im \alpha (\phi_{r(j)}-\phi_{r(j)+\vepsilon})}c_{r(j)+\vepsilon  }c_{r(j)  }\Big)\no
=& \sum_{{j\in \Lambda_R\cap \Lambda_e}\atop{j\pm 1\in \Lambda_R}}
\sum_{\vepsilon=\pm 1}
(-t)\Big(e^{-\im \alpha (\phi_j-\phi_{j+\vepsilon})}c_{j  }^* c_{j+\vepsilon  }^*+e^{+\im \alpha (\phi_j-\phi_{j+\vepsilon})}c_{j+\vepsilon  }c_{j  }\Big)\no
=&\mathbb{T}_R.
\end{align}
 It should be observed that, at this stage, the factor $(-1)^{j}$ in \eqref{DefTheta} is employed.

(ii) and (iv) are straightforward.

(iii) Special attention must be given to the cross term.
Since $\ell$ is an odd number, we have
$
\vartheta^{-1} c_{\ell-1 } \vartheta=b_{-\ell }
$. Consequently, we arrive at \eqref{TLRTensor}.  (It should be noted that if $\ell$ is even, an additional negative factor arises, thereby invalidating the subsequent arguments.)

(v) By employing \eqref{DefTheta}, we find
\be
\vartheta \phi_j\vartheta^{-1}=\phi_{r(j)},\ \ \vartheta \pi_j\vartheta^{-1}=-\pi_{r(j)},\ \ j\in \Lambda_L.
\ee
Consequently, by virtue of \eqref{KLR}, we obtain (v).
\end{proof}

\subsubsection{Expression III}\label{ExpressionIII}

In the subsequent analysis, it is assumed that $\ell$ is always odd.
For a given $\Lambda\in \mathbb{F}$, we introduce the notation:
\be
\mathfrak{A}_{\Lambda}=\left\{A\in \mathscr{L}(\Fock_{\Lambda_L}) : A\xi, \ \xi A\in \THL\ \mbox{for all $\xi\in \THL$}\right\}. \label{DefALamba}
\ee
It can be readily verified that
$\mathbb{T}_{L}, \mathbb{T}_R, \mathbb{W}_L, \mathbb{W}_R\in \mathfrak{A}_{\Lambda}
$.
Consequently, employing the definitions of ${\bs L}(\cdot)$ and ${\bs R}(\cdot)$ outlined in Section \ref{SectDirectSum}, we deduce the following:

\begin{itemize}
\item {\it Fermion-phonon interaction:}
\begin{align}
\mathbb{T}_L\otimes \one\restriction \tilde{\mathfrak{H}}_{\Lambda}=&{\bs L}(\mathbb{T}_L)\no
=&
\sum_{{j\in \Lambda_L\cap \Lambda_e}\atop{j\pm 1\in \Lambda_L}}
\sum_{\vepsilon=\pm 1}
(-t){\bs L}\big(
e^{-\im \alpha(\phi_j-\phi_{j+\vepsilon})}
b_{j  }^* b_{j+\vepsilon  }^*+
e^{+\im \alpha(\phi_j-\phi_{j+\vepsilon})}
b_{j+\vepsilon  }b_{j  }\big),
\end{align}
where
$
{\bs L}(A)
$ is defined by \eqref{leftrightOP}.
Similarly, we have 
\begin{align}
\one \otimes \mathbb{T}_R\restriction \tilde{\mathfrak{H}}_{\Lambda}=&{\bs R}(\mathbb{T}_L)\no
=&
\sum_{{j\in \Lambda_L\cap \Lambda_e}\atop{j\pm 1\in \Lambda_L}}
\sum_{\vepsilon=\pm 1}
(-t){\bs R}\big(e^{-\im \alpha(\phi_j-\phi_{j+\vepsilon})}b_{j  }^* b_{j+\vepsilon  }^*+
e^{+\im \alpha(\phi_j-\phi_{j+\vepsilon})}b_{j+\vepsilon  }b_{j  }\big),
\end{align}
where 
$
{\bs R}(A)
$ is defined by \eqref{leftrightOP}.
Regarding the cross term, we obtain the following expression:
\begin{align}
\mathbb{T}_{LR}\restriction \tilde{\mathfrak{H}}_{\Lambda}=&  (-t)
\Big\{{\bs L}\big(e^{+\im \alpha \phi_{-1}}b_{-1  }^*\big) {\bs R}\big(e^{-\im \alpha\phi_{-1}}b_{-1}\big)+{\bs L}\big(e^{-\im \alpha\phi_{-1}}b_{-1  }\big) {\bs R}\big( e^{+\im \alpha\phi_{-1}} b_{-1  }^*\big)\no
 &+{\bs L}\big(
e^{+\im \alpha\phi_{-\ell}} b_{-\ell  }^*
\big) {\bs R}\big(
e^{-\im \alpha\phi_{-\ell}}
 b_{-\ell  } \big)
 +{\bs L}\big( e^{-\im \alpha\phi_{-\ell}}b_{-\ell} \big){\bs R}\big(
 e^{+\im \alpha\phi_{-\ell}}
  b_{-\ell  }^*\big) \Big\}.
\end{align}

\item 
{\it Coulomb interaction:}
\begin{align}
\mathbb{W}_L\otimes \one\restriction \tilde{\mathfrak{H}}_{\Lambda}=& \sum_{i, j\in \Lambda_L}W(i-j){\bs L}\big(
\delta \hn_i\delta \hn_{j}
\big),\\
\one \otimes \mathbb{W}_R\restriction \tilde{\mathfrak{H}}_{\Lambda}=& \sum_{i,j\in \Lambda_L}W(i-j){\bs R}\big(
\delta \hn_i\delta \hn_{j}
\big),\\
\mathbb{W}_{LR}\restriction \tilde{\mathfrak{H}}_{\Lambda}=& 2\sum_{i, j\in \Lambda_L}W(i+j+1){\bs L}\big(\delta \hn_{i}\big){\bs R}\big(\delta \hn_{j}\big).
\end{align} 
\item {\it Phonon energy:} 
\be
K={\bs L}(K_L)+{\bs R}(K_L).
\ee
Here, recall the definition of ${\bs L}(\cdot)$ and ${\bs R}(\cdot)$ for unbounded operators as provided in Section \ref{SectDirectSum}.
\end{itemize}

\subsection{Proof of Theorem  \ref{ThmPP}} \label{PfSect4}
By recalling Definition \ref{DefRPCone2}, it suffices to prove that $e^{-\beta \tilde{H}_{\Lambda}} \unrhd 0$ w.r.t. $\tCone$ for all $\beta \ge 0$.

First, it should be noted that using the expressions provided in Section \ref{ExpressionIII}, we can express $\tilde{H}_{\Lambda}$ as
\begin{align}
\tilde{H}_{\Lambda}={\bs L}(\mathbb{K})+{\bs R}(\mathbb{K})-\mathbb{V},
\end{align}
where
\begin{align}
\mathbb{K}=\mathbb{T}_L-\mathbb{W}_L+K_L,\ \ \mathbb{V}=-\mathbb{T}_{LR}+\mathbb{W}_{LR}. \label{DefKW}
\end{align}
Let us define
\be
G_0={\bs L}(\mathbb{K})+{\bs R}(\mathbb{K}). \label{DefG0}
\ee
Then, by applying Proposition \ref{PPI3}, we have
\be
e^{-\beta G_0}={\bs L}(e^{-\beta \mathbb{K}}){\bs R}(e^{-\beta \mathbb{K}}) \unrhd 0\ \ \ \mbox{w.r.t. $\tCone$}.
\ee
Furthermore, employing Proposition \ref{PPI3} once again, it follows that
$\mathbb{V} \unrhd 0$ w.r.t.  $\tCone$. Consequently, based on Lemma \ref{ppiexp1}, we can conclude that $e^{-\beta \tilde{H}_{\Lambda}} \unrhd 0$ w.r.t.  $\tCone$ for all $\beta \ge 0$.
\qed

\section{Ergodicity of $e^{-\beta H_{\Lambda}^{\F}}$}\label{Sect5}

\subsection{Statement of the result}

Our goal in this section is to prove the following:

\begin{Thm}\label{PISemi}
Assume \hyperlink{A}{\bf (A)} and \hyperlink{B2}{\bf (B. 2)}. 
The semigroup $e^{-\beta H_{\Lambda}^{\F}}$ is ergodic   w.r.t. $\rCone^{\F}$.
\end{Thm}
We will provide a proof of Theorem \ref{PISemi} in Subsection \ref{PfThmPI}.
\subsubsection*{\it 
Proof of Theorem \ref{SparP} assuming Theorem \ref{PISemi}}

(i) is straightforward. By utilizing Proposition \ref{PFF}, we observe that the ground state of $H_{\Lambda}^{\F}$ is strictly positive w.r.t.  $\rCone^{\F}$. Consequently, based on Theorem \ref{RPConeBasic}, we can establish (ii) of Theorem \ref{SparP}.
\qed
\subsection{Preliminaries}
We decompose $\Lambda_L$ as $\Lambda_L=\Lambda_{L, e} \cup \Lambda_{L, o}$, where 
\begin{align}
\Lambda_{L, e}=\{j\in \Lambda_L\, :\, \mbox{$j$ is even}\},\ \ \ 
\Lambda_{L, o}=\{j\in \Lambda_L\, :\, \mbox{$j$ is odd}\}.
\end{align} 
For a given $n\in \BbbN$, we define
\begin{align}
\Theta_{\Lambda, e}^n:=\{I\subseteq \Lambda_{L, e}\, :\, 
  |I|=n
\},\ \
\Theta_{\Lambda, o}^n:=\{I \subseteq  \Lambda_{L, o}\, :\, 
  |I|=n
\},
\end{align} 
where $ |I|$ indicates the cardinality of $I$.
Now we define
\begin{align}
\Theta_{\Lambda}(q):=
\bigcup_{{m_e-m_o=q}}\Theta_{\Lambda, e}^{m_e}\times \Theta_{\Lambda, o}^{m_o}
,\quad q\in \mathbb{I}_{\Lambda}.
\end{align} 
For each $X=(I_e, I_o)\in \Theta_{\Lambda}(q)$, we introduce a vector $e_X\in \AFock_{\LL}(q)$ defined as 
\begin{align}
e_X:=\B_{I_e}^* \B_{I_o}^*  \Omega^{\rm F}_{\LL}, \label{DefEx}
\end{align} 
where $\Omega^{\rm F}_{\LL}$ represents  the Fock vacuum in $\AFock_{\LL}$ and, for $I=\{i_1, i_2, \dots, i_n\}$ with $i_1<i_2<\cdots<i_n$, 
\begin{align}
\B_{I  }^*:= b_{i_1  }^* b_{i_2  }^*\cdots b_{i_n  }^*.
\end{align}

The following lemma can be readily derived from the definitions of $\AFock_{\LL},\ \THL^{\rm F}(q)$ and $\THL^{\rm F}$.

\begin{lemm}\label{CONSsLmm}
We have the following:
\begin{itemize}
\item[{\rm (i)}] For each $q\in \mathbb{I}_{\Lambda}$,  $\{e_X\, :\, X\in \Theta_{\Lambda}(q)\}$ is a CONS of
$\AFock_{\LL}(q)$.

\item[{\rm (ii)}] For each $q\in \mathbb{I}_{\Lambda}$,  $\{|e_X\ra\la e_Y|\, :\, X, Y\in \Theta_{\Lambda}(q)\}$ is a
	     CONS of $\THL^{\rm F}(q)$.
\item[{\rm (iii)}] $\bigcup_{q\in \mathbb{I}_{\Lambda}}\{|e_X\ra\la e_Y|\, :\, X, Y\in \Theta_{\Lambda}(q)\}$ is a
	     CONS of $\tilde{\mathfrak{H}}_{\Lambda}^{\rm F}$. 
\end{itemize} 
\end{lemm} 
\subsection{Lower bound for $e^{-\beta \tilde{H}_{\Lambda}^{\F}}$} \label{LowerBddE}

In the remaining part of this section, we will continue to use the expressions provided in Section \ref{PfSect4}.
We start by presenting the following lemma.
\begin{lemm}\label{WLowerBd}
Let $w_0>0$ be the smallest eigenvalue of the matrix $\{W_{i, j}\}_{i, j}$ with $W_{i, j}=W(i+j+1)$.
Then we have 
\begin{align}
\mathbb{W}_{LR} \unrhd \mathbb{W}_0\ \ \ w.r.t.\  \tCone^{\F},
\end{align}
where $\mathbb{W}_0=w_0 \sum_{i\in \Lambda_L}
{\bs L}(\delta \hn_i){\bs R}(\delta \hn_i)
$.
\end{lemm}
\begin{proof}
Based on the assumption \hyperlink{B2}{\bf (B. 2)}, the matrix ${\bs W}=\{W_{i, j}\}_{i, j}$ is positive definite. Consequently, all eigenvalues, $\lambda_1, \dots, \lambda_{|\LL|}$, are strictly positive. Moreover, there exists a real unitary matrix ${\bs M}$ that diagonalizes ${\bs W}$ as follows:
\be
{\bs W} = {\bs M}^{-1} {\rm diag}(\lambda_1, \dots, \lambda_{|\LL|}) {\bs M}^{-1}.
\ee
 By setting $ \mu_i=\sum_{j\in \LL} M_{i, j} \delta \hn_j$, we have
 \begin{align}
 \mathbb{W}_{LR}= \sum_{k} \lambda_k {\bs L}(\mu_k){\bs R}(\mu_k) \unrhd w_0 \sum_k{\bs L}(\mu_k){\bs R}(\mu_k) \ \ \mbox{w.r.t. $\tCone^{\F}$}.
 \end{align}
 Here, we have used the fact that $
 {\bs L}(\mu_k){\bs R}(\mu_k) \unrhd 0 
 $ w.r.t. $\tCone^{\F}$.
This completes the proof of Lemma \ref{WLowerBd}.
\end{proof}

By applying Lemmas  \ref{ppiexp1} and \ref{WLowerBd}, we obtain  the following:
\begin{coro}\label{CorRed}
Let 
\begin{align}
\tilde{H}^{\F}_{\Lambda, 0}={\bs L}(\mathbb{K}^{\F})+{\bs R}(\mathbb{K}^{\F})-\mathbb{V}_0,
\end{align}
where $\mathbb{V}_0=-\mathbb{T}_{LR}+\mathbb{W}_0$ and 
$\mathbb{K}^{\F}=\mathbb{T}_L-\mathbb{W}_L$ with $g=0$. Then we have 
$e^{-\beta \tilde{H}_{\Lambda}^{\F}} \unrhd e^{-\beta \tilde{H}^{\F}_{\Lambda, 0}}\unrhd  0$ w.r.t. $\tCone^{\F}$ for all $\beta \ge 0$.
\end{coro}
Based on Corollary \ref{CorRed}, in order to establish Theorem \ref{PISemi}, it is sufficient to demonstrate the ergodicity of $e^{-\beta \tilde{H}_{\Lambda, 0}^{\F}}$ with respect to $\tCone^{\F}$. Therefore, we will examine the properties of the semigroup $e^{-\beta \tilde{H}^{\F}_{\Lambda, 0}}$ in the remaining part of this section.

It is worth noting that $\mathbb{W}_0$ can be represented as
\be
\mathbb{W}_0=\frac{w_0}{2} (P_0+P_1)-\frac{w_0|\Lambda|}{8},
\ee
where 
\begin{align}
P_0=\sum_{i\in \Lambda_L} {\bs L}(1-\hn_i){\bs R}(1-\hn_i),\ \
P_1= \sum_{i\in \Lambda_L} {\bs L}(\hn_i){\bs R}(\hn_i).
\end{align}
Hence, $\tilde{H}^{\F}_{\Lambda, 0}$ can be rewritten as 
\be
\tilde{H}_{\Lambda, 0}^{\F}=G-\frac{w_0}{2} P-\frac{w_0|\Lambda|}{8}, \label{tildeHDef}
\ee
where $P=P_0+P_1$ and 
\be
G={\bs L}(\mathbb{K}^{\F})+{\bs R}(\mathbb{K}^{\F})+\mathbb{T}_{LR}. \label{DefiG}
\ee
In the subsequent discussions, we will disregard the constant term $-w_0|\Lambda|/8$ in \eqref{tildeHDef}, as it does not affect the following arguments.
By using the Dyson formula, we have
\be
e^{-\beta \tilde{H}^{\F}_{\Lambda, 0}}=\sum_{n=0}^{\infty}D_n(\tilde{H}^{\F}_{\Lambda, 0}),\label{DysonH0}
\ee
where 
\be
D_n(\tilde{H}_{\Lambda, 0}^{\F})=\Big(\frac{w_0}{2}\Big)^n \int_{0\le s_1\le \cdots \le s_n \le \beta} d{\bs s}
P(s_1)\cdots P(s_n)e^{-\beta G}\label{DnH0}
\ee
with $P(s)=e^{-s G} Pe^{s G}$ and $D_0(\tilde{H}^{\F}_{\Lambda, 0})=e^{-\beta G}$. Note that the right-hand side of \eqref{DysonH0}
 converges in the operator norm topology.
Because $e^{-s G} \unrhd 0$ and $P\unrhd 0$ w.r.t. $\tCone^{\F}$, we find that $P(s_1)\dots P(s_n)e^{-\beta G}\unrhd 0$ w.r.t. $\tCone^{\F}$, provided that ${\bol s}=(s_1, \dots, s_n)\in R_{\beta, n}$,
where $R_{\beta, n}$ is defined by 
\be
R_{\beta, n}=\{{\bol s}=(s_1, \dots, s_n)\in \BbbR^n\, :\, 0\le s_1\le \cdots \le s_n \le \beta\}.
\ee
 This implies that
$D_n(\tilde{H}_{\Lambda, 0}^{\F})\unrhd 0$ w.r.t. $\tCone^{\F}$ for all $n\in \BbbZ_+$.
As a result,  we have the following:
\begin{lemm}
For each $n\in \BbbZ_+$ and $\beta \ge 0$, we have
\be
e^{-\beta \tilde{H}^{\F}_{\Lambda}} \unrhd D_n(\tilde{H}^{\F}_{\Lambda, 0})\unrhd 0\ \ \ \mbox{w.r.t. $\tCone^{\F}
$}.
\ee
\end{lemm}

Without loss of generality, we may assume that $w_0/2=1$. 
This substitution simplifies the subsequent notation without altering the essence of the discussion.

Given an  $X\subseteq \Lambda_L$, we define 
${\bol \vepsilon^X}=(\vepsilon^X_i)_{i\in X}\in \{0, 1\}^{|\Lambda_L|}$ by 
\be
\vepsilon^X_i=\begin{cases}
1, & \mbox{if $i\in X$}\\
0,  & \mbox{if $i\in \Lambda_L\setminus X$}.
\end{cases}
\ee
Using this notation, we define $\mathbb{P}_X(\bol s)$ as follows:
\be
\mathbb{P}_X(\bol s)=\prod_{j\in \Lambda_L} P_{\vepsilon^X_j, j}(s_j),\ \ {\bol s}\in R_{\beta, |\Lambda_L|},
\ee
where $P_{\vepsilon, i}(s)=e^{-s G}P_{\vepsilon, i}e^{sG}$ with $P_{\vepsilon, i}={\bs L}(\hn_i)
{\bs R}(\hn_i)$ for $\vepsilon=1$,  $P_{\vepsilon, i}={\bs L}(1-\hn_i)
{\bs R}(1-\hn_i)$ for $\vepsilon=0$. 
 Using the fact that 
 \be
  |e_X\ra\la e_X|=\Bigg[\prod_{i\in X} \hn_i\Bigg]\Bigg[\prod_{i\in \Lambda_L\setminus X}(1- \hn_i)\Bigg], \ee
  we obtain
\be
\mathbb{P}_X(s, \dots, s)=e^{-s G} \mathbb{E}_Xe^{sG} ,\quad  s\in [0, \beta], \label{SpecialE}
\ee
where 
\be
\mathbb{E}_X={\bs L}(E_X){\bs R}(E_X),\ \ \ E_X=|e_X\ra\la e_X|.
\ee
For a given sequence $\mathscr{X}=(X_i)_{i=1}^N$ of subsets of $\Lambda_L$,  where $X_i\subseteq \Lambda_L$ for $i=1,\dots,N$, we define the operator 
\be
\mathbb{P}_{\mathscr{X}}({\bol s})=\mathbb{P}_{X_1}({\bol s}_1)\cdots\mathbb{P}_{X_N}({\bol s}_N), \label{BBPX}
\ee
where ${\bol s}=({\bol s}_1, \dots, {\bol s}_N) \in R_{\beta, n(N)}$ and $n(N)=N|\Lambda_L|$.
Then  $D_{n(N)}(\tilde{H}^{\F}_{\Lambda, 0})$ can be expressed as 
\be
D_{n(N)}(\tilde{H}^{\F}_{\Lambda, 0})=\sum_{\mathscr{X}\in \Upsilon_N} \int_{R_{\beta, n(N)}} d{\bs s}\,  \mathbb{P}_{\mathscr{X}}({\bol s})e^{-\beta G}, \label{PathIntD}
\ee
where $\Upsilon_N=\{\mathscr{X}=(X_j)_{j=1}^N\, :\, X_j\subseteq \Lambda_L\}$, the set of all \lq\lq{}paths with length $N$\rq\rq{} in the fermion configuration space.
Equation \eqref{PathIntD} provides a path integral representation of the operator $D_{n(N)}(\tilde{H}^{\F}_{\Lambda, 0})$. It expresses the operator as a sum over all possible paths $\mathscr{X}$ in the fermion configuration space, where each path is associated with an operator $\mathbb{P}_{\mathscr{X}}({\bol s})$ and weighted by the exponential factor $e^{-\beta G}$. This path integral representation allows us to analyze the operator in terms of paths and provides a convenient framework for further calculations and analysis.

Because $\mathbb{P}_{\mathscr{X}}({\bol s})e^{-\beta G} \unrhd 0$ w.r.t. $\tCone^{\F}$, we have the following:

\begin{Prop}\label{ExpHLower}
For all  $\beta \ge 0,\ N\in \BbbZ_+$ and $\mathscr{X}\in \Upsilon_N$, one obtains 
\be
e^{-\beta \tilde{H}^{\F}_{\Lambda}} \unrhd \int_{R_{\beta, n(N)}}d{\bs s}\,  \mathbb{P}_{\mathscr{X}}({\bol s})e^{-\beta G}, \label{PathLower}
\ee
where $n(N)=N|\Lambda_L|$, $\mathbb{P}_{\mathscr{X}}({\bol s})$ is given by \eqref{BBPX} and $G$ is defined by \eqref{DefiG}. 
%Here, we understand that the right-hand side of \eqref{PathLower} equals to $e^{-\beta G}$  for $N=0$.%
\end{Prop}

As we shall observe, this lower bound serves as a valuable tool in establishing Theorem \ref{PISemi}.

\subsection{Reductions of the problem }

The proof of Theorem \ref{PISemi} turns out to be more intricate than anticipated. Here, we shall analyze the conditions for Theorem \ref{PISemi} to hold using the path integral representation.

\subsubsection{Reduction I}

Fix $\vphi, \psi\in \tCone^{\F}\backslash \{0\}$, arbitrarily.
Our purpose is to show the following:
\begin{Prop}\label{PIUnitary}
There exists a $\beta\ge 0$ such that 
$
\la\vphi|e^{-\beta\tilde{H}_{\Lambda}^{\F} }  \psi\ra>0.
$
\end{Prop}

We will provide a proof of Proposition \ref{PIUnitary} in Subsection \ref{PfThmPI}.
\begin{proof}[Proof of Theorem \ref{PISemi} assuming Proposition \ref{PIUnitary}] 
\  \smallskip\\
From Proposition \ref{PIUnitary}, it follows that the semigroup 
$
 e^{-\beta \tilde{H}_{\Lambda}^{\F}}$ is ergodic w.r.t. $\tCone^{\F}$. 
  Taking Definition \ref{DefRPCone3} into account, we conclude Theorem \ref{PISemi}.
\end{proof}
Based on the preceding discourse, to establish Theorem \ref{PISemi}, it is sufficient to demonstrate Proposition \ref{PIUnitary}. In the subsequent discussion, we will thoroughly examine the conditions for Proposition \ref{PIUnitary} to hold, utilizing the path integral representation constructed in the previous subsection.

The fundamental strategy is as follows. By virtue of Proposition \ref{ExpHLower}, in order to prove Proposition \ref{PIUnitary}, it is enough to demonstrate the existence of $\beta\ge 0$, $N\in \BbbZ_+$, and $\mathscr{X}\in \Upsilon_N$ such that
\be \la \vphi|\int_{R_{\beta, n(N)}} d{\bs s} \mathbb{P}_{\mathscr{X}}({\bs s}) \psi\ra>0. \label{PathIntBd}
\ee
In the ensuing analysis, we will effectively establish \eqref{PathIntBd}.

Because of \eqref{DefrPq} and \eqref{DirectSumP}, we can express $\vphi $ and $\psi$ as 
\begin{align}
\vphi=\bigoplus_{q\in \mathbb{I}_{\Lambda}} \sum_{X, Y\in \Theta_{\Lambda}(q)}\vphi^{(q)}_{X, Y}  |e_X \ra \la  e_Y|,
\ \ \psi=\bigoplus_{q\in \mathbb{I}_{\Lambda}} \sum_{X, Y\in \Theta_{\Lambda}(q)}  \psi^{(q)}_{X, Y} |e_X \ra \la  e_Y|
\end{align}
with $\vphi^{(q)}_{X, Y}, \psi^{(q)}_{X,Y} \in \BbbC$.
We begin with the following lemma.
\begin{lemm}\label{LemmNonZero}
There exist  $q, q\rq{}\in \mathbb{I}_{\Lambda}$,  $X_0\in \Theta_{\Lambda}(q)$ and $Y_0\in \Theta_{\Lambda}(q')$ such that 
$\vphi^{(q)}_{X_0X_0}>0 $ and $ \psi^{(q\rq{})}_{Y_0Y_0}>0$.
\end{lemm}

\begin{proof}
Because $\vphi$ is a non-zero vector, there exists a $q\in \mathbb{I}_{\Lambda}$ such that 
\be
\vphi^{(q)}=\sum_{X, Y\in \Theta_{\Lambda}(q)} \vphi^{(q)}_{XY} 
|e_X\ra\la e_Y|\label{Defphiq}
\ee
 is non-zero. Trivially, it holds that $\vphi^{(q)} \ge 0$ w.r.t. $\tCone^{\F}(q)$.
Assume that $\vphi^{(q)}_{X, X}=\la e_{X}|\vphi^{(q)} e_{X}\ra=0$ for all $X\in \Theta_{\Lambda}(q)$.
Then, because $(\vphi^{(q)})^2\le \|\vphi^{(q)}\| \vphi^{(q)}$ holds, we have
\be
\la e_{X}|(\vphi^{(q)} )^2e_{X}\ra
\le \| \vphi^{(q)}\|\la e_{X}|\vphi^{(q)} e_{X}\ra=0,
\ee
which implies that $\|\vphi^{(q)}\|_2^2=\Tr[(\vphi^{(q)})^2]=0$.  This contradicts with the condition that $\vphi^{(q)} \neq 0$.
Similarly, the proof can be established for $\psi^{(q')}_{Y_0Y_0}$.
\end{proof}

Let us define  $\BbbZ_{\ge 2}=\{N\in \BbbZ : N\ge 2\}$.  Consider $N \in \mathbb{Z}_{\geq 2}$, chosen arbitrarily. We examine  a path  $\mathscr{X}_0=(X_0, \mathscr{Z}, Y_0)\in \Upsilon_{N}$, where $\mathscr{Z}\in \Upsilon_{N-2}$ will be determined  later.
Consequently, we have
\be
\mathbb{P}_{\mathscr{X}_0}({\bol s}_0)=\mathbb{E}_{X_0} \mathbb{P}_{\mathscr{Z}}({\bol s}_*)  e^{-\beta G}\mathbb{E}_{Y_0}
e^{\beta G}. \label{ChoiceS0}
\ee
In this equation,  ${\bol s}_0\in R_{\beta, n(N)}$ is given by 
$
{\bol s}_0=({\bol s}_1, {\bol s}_{*}, {\bol s}_N),
$
where 
\be {\bol s}_1=(\underbrace{0, \dots, 0}_{|\Lambda_L|}),\ \ 
{\bol s}_N=(\underbrace{\beta, \dots, \beta}_{|\Lambda_L|})
\ee 
and 
\be
{\bol s}_*=(s_{*, 1}, \dots, s_{*, n(N-2)})\in R_{\beta, n(N-2)}.
\ee
Furthermore, it should be noted that the property \eqref{SpecialE} has been employed to derive \eqref{ChoiceS0}.
For $N=2$, we understand that $X_0=Y_0$, $\mathscr{Z}=\varnothing$ and $\mathbb{P}_{\mathscr{X}_0}({\bol s}_0)=\mathbb{E}_{X_0} 
e^{\beta G}.$
By using Proposition \ref{ExpHLower}, we have
\begin{align}
\la\vphi|e^{-\beta\tilde{H}_{\Lambda} } \psi\ra
 &\ge\int_{R_{\beta, n(N)}}d{\bs s}\,  \la\vphi|\mathbb{P}_{\mathscr{X}_0}({\bol s})e^{-\beta G}\psi\ra.\label{IntLower}
\end{align}

Let us recall the direct sum representation of $\THL^{\F}$, as given in equation \eqref{THLDirect}. Then, for a given $q \in \mathbb{I}_{\Lambda}$ and $X \in \Theta_{\Lambda}(q)$, we define:
\begin{align}
E^{(q)}_{X}&=(0, \dots, 0, \underbrace{E_{X}}_{q^{\rm th}}, 0, \dots, 0)\in \THL^{\F}.
\end{align}
\begin{lemm}\label{LemmRed1}
Let $q, q\rq{}\in \mathbb{I}_{\Lambda},$ $X_0\in \Theta_{\Lambda}(q)$ and $ Y_0\in \Theta_{\Lambda}(q\rq{})$
be given in Lemma \ref{LemmNonZero}. 
 Suppose that there exist $\beta\ge 0$,  $N\in \BbbZ_{\ge 2}$,  $\mathscr{Z}\in \Upsilon_{N-2} $ and $ {\bol s}\in R_{\beta, n(N-2)}$   such that 
\be
\big\la 
E^{(q)}_{X_0}
 \big|\mathbb{P}_{\mathscr{Z}}({\bol s}) e^{-\beta G}
 E^{(q\rq{})}_{Y_0}
\big \ra>0.\label{PZs}
\ee
Then  $\la \vphi|e^{-\beta \tilde{H}^{\F}_{\Lambda}}  \psi\ra>0$ holds. Hence, the semigroup $e^{-\beta \tilde{H}^{\F}_{\Lambda}} $ is ergodic 
w.r.t. $\tCone^{\F}$.
\end{lemm}
\begin{proof}
It is worth noting that $
\la \vphi|\mathbb{P}_{\mathscr{X}_0}({\bol s})e^{-\beta G}  \psi\ra
$ is continuous in $\bol s$ and $
\mathbb{P}_{\mathscr{X}_0}({\bol s})e^{-\beta G} \unrhd 0
$ w.r.t. $\tCone^{\F}$ for  all ${\bol s} \in R_{\beta, n(N)}$. Therefore, if there exists an  ${\bol s}\in R_{\beta, n(N)}$  such that $
\la \vphi|\mathbb{P}_{\mathscr{X}_0}({\bol s})e^{-\beta G} \psi\ra>0
$ holds, then the right-hand side of \eqref{IntLower} is strictly positive. Consequently, due to  \eqref{ChoiceS0}, it suffices to show that $\la \vphi| \mathbb{E}_{X_0} \mathbb{P}_{\mathscr{Z}}({\bol s}) e^{-\beta G}\mathbb{E}_{Y_0} \psi\ra>0$ for some ${\bs s}\in R_{\beta, n(N-2)}$.
To accomplish this, we note that
$
\vphi\ge \vphi^{(q)}\ \ \mbox{w.r.t. $\tCone^{\F}$}
$
holds for all $q\in \mathbb{I}_{\Lambda}$, where $\vphi^{(q)}$ is given by \eqref{Defphiq} and we identify $\vphi^{(q)}$ with
\be
(0, \dots, 0, \underbrace{\vphi^{(q)}}_{q^{\rm th}}, 0,\dots, 0) \in \IL^{\F}.
\ee
Similarly, we have $\psi \geq \psi^{(q')}$ as well.
Using these inequalities, we have 
\begin{align}
\la \vphi| \mathbb{E}_{X_0} \mathbb{P}_{\mathscr{Z}}({\bol s}) e^{-\beta G} \mathbb{E}_{Y_0}  \psi\ra
\ge& \big\la   \vphi^{(q)}| \mathbb{E}_{X_0} \mathbb{P}_{\mathscr{Z}}({\bol s}) e^{-\beta G} \mathbb{E}_{Y_0} \psi^{(q\rq{})}\big\ra\no
=& \vphi^{(q)}_{X_0X_0}\psi^{(q\rq{})}_{Y_0Y_0}\big\la  E_{X_0}^{(q)} \big| \mathbb{P}_{\mathscr{Z}}({\bol s})  e^{-\beta G} E^{(q\rq{})}_{Y_0} \big\ra
>0, \label{LowerB4}
\end{align}
where
\be
E^{(q)}_{X}=(0, \dots, 0, \underbrace{E_{X}}_{q^{\rm th}}, 0, \dots, 0)\in \THL^{\F}.
\ee
The proof of Lemma \ref{LemmRed1} is now concluded.
\end{proof}

\subsubsection{Reduction  II}
Our objective here is to prove Proposition \ref{RedPI}. Let us  define $\boldsymbol{u} \in R_{\beta, n(N-2)}$ as
\be
{\bol u}:=(\underbrace{u_1, \dots,u_1}_{|\Lambda_L|},  \underbrace{u_1+u_2, \dots, u_1+u_2}_{|\Lambda_L|},  \dots, \underbrace{u_{1}+\cdots+u_{N-2}, \dots, u_{1}+\cdots+u_{N-2}}_{|\Lambda_L|}), \label{Defu}
\ee
where $u_1, \dots, u_{N-2}\in [0, \beta)$ satisfy $u_1+\cdots+u_{N-2}=\beta$. Inserting  ${\bs s}={\bs u}$
into  the left-hand side of \eqref{PZs} and using \eqref{SpecialE}, 
 we obtain
\begin{align}
\big\la E^{(q)}_{X_0}|\mathbb{P}_{\mathscr{Z}}({\bol u}) e^{-\beta G}E^{(q')}_{Y_0}\big\ra
=
\Big\la E^{(q)}_{X_0}\Big|e^{-u_1 G}\mathbb{E}_{X_1}e^{-u_2 G} \mathbb{E}_{X_2}\cdots \mathbb{E}_{X_{N-3}}e^{-u_{N-2} G} E^{(q')}_{Y_0}\Big\ra. \label{A}
\end{align}
Here, we choose $\mathscr{Z}=(X_i)_{i=1}^{N-2}$ such that $X_{N-2}=Y_0$.
 In the subsequent analysis, our objective is to obtain a lower bound for the right-hand side of \eqref{A}.
We define $G_0^{\F}$ as $G_0^{\F}=\bs{L}(\mathbb{K}^{\F})+\bs{R}(\mathbb{K}^{\F})$.
Using the identity
$
G=G_0^{\F}+\mathbb{T}_{LR}
$ 
and the Dyson formula, 
we have
\be
e^{-s G}=\sum_{k=0}^{\infty}D_{s, k}(G),
\ee
where
\be
D_{s, k}(G)=\int_{R_{s, k}} d{\bs s}(-\mathbb{T}_{LR}(s_1)) \cdots (-\mathbb{T}_{LR}(s_k))e^{-sG_0^{\F}} \label{DefDGm}
\ee
with $\mathbb{T}_{LR}(s)=e^{-sG_0^{\F}} \mathbb{T}_{LR} e^{sG_0^{\F}}$.
Because $(-\mathbb{T}_{LR}(s_1)) \cdots (-\mathbb{T}_{LR}(s_k))e^{-sG_0^{\F}} \unrhd 0$ for ${\bol s } \in R_{s,k}$, we have
\be
e^{-sG} \unrhd D_{s, k}(G)\unrhd 0\ \mbox{w.r.t. $\tCone^{\F}$}, \quad   k\in \BbbZ_+.
\ee
Inserting this into the right-hand side of \eqref{A}, we find that 
\begin{align}
&\mbox{the right-hand side of \eqref{A}}\no
\ge&  \Big\la E^{(q)}_{X_0} \Big|D_{u_1, \vepsilon_1}(G) \mathbb{E}_{X_1} D_{u_2, \vepsilon_2}(G) \mathbb{E}_{X_2} \cdots \mathbb{E}_{X_{N-3}}   D_{u_{N-2}, \vepsilon_{N-2}}(G)E^{(q')}_{Y_0}\Big\ra \label{433Below}
\end{align}
for every ${\bol \vepsilon}=(\vepsilon_i)_{i=1}^{N-2}\in \{0, 1\}^{N-2}$, where $D_{u, \vepsilon=0}(G)=e^{-u G_0^{\F}}$
 and $D_{u, \vepsilon=1}(G)$  is  given by \eqref{DefDGm}
 with $k=1$.

We observe that  $\mathbb{T}_{LR}$ can be expressed as 
\be
\mathbb{T}_{LR}=\sum_{\xi\in \{-\ell, -1\}} \sum_{\sharp \in \{-, +\}}(-t) {\bs L}(b_{\xi}^{\sharp }) {\bs R}(\{b_{\xi}^{\sharp }\}^{*}),
\ee
where 
$b^{\sharp }_{\xi}=b_{\xi}$ if $\sharp =-$, and $b^{\sharp }_{\xi}=b_{\xi}^*$ if $\sharp =+$.
With this mind,  we  can rewrite $D_{u, 1}(G)$ as 
\be
D_{u, 1}(G)=\sum_{\xi\in \{-\ell, -1\}} \sum_{ \sharp \in \{-, +\}} t \int_0^u ds
J(\xi,  \sharp , s), 
\ee
where 
\be
J( \xi,  \sharp , s)=
 {\bs L}
\left(
b_{\xi  }^{\sharp }(s)e^{-u  \mathbb{K}^{\F}}
\right)
{\bs R}
\left(
\big\{
b_{\xi  }^{\sharp }(s) e^{-u  \mathbb{K}^{\F}}
\big\}^*
\right),\ \ b^{\sharp }_{\xi}(s)=e^{-s \mathbb{K}^{\F}} b_{\xi}^{\sharp }e^{s\mathbb{K}^{\F}}.
\ee
Because $J(\xi,   \sharp ,   s )\unrhd 0$ for $0\le s \le u$, we  obtain
\begin{align}
D_{u, 1}(G) \unrhd t \int_0^u  dsJ( \xi,  \sharp , s)\ \ \ \mbox{w.r.t. $\tCone^{\F}$} \label{DuGLower}
\end{align}
for each $\sharp \in \{-, +\}$ and $\xi\in \{-\ell, -1\}$.

Let ${\bol \mu}=(\mu_i)_{i=1}^k\subseteq \{1, 2, \dots, N-3\}$ be chosen such that $\mu_i < \mu_{i+1}$. 
 For  such a $\bol \mu$, we define
\be
{\bol \vepsilon}({\bol \mu})=(0, \dots, 0, \underbrace{1}_{\mu_1^{\rm th}}, 0, \dots, 0, \underbrace{1}_{\mu_2^{\rm th}}, 0, \dots
, \underbrace{1}_{\mu_k^{\rm th}}, 0, \dots, 0
) \in \{0, 1\}^{N-2}. \label{epsilonmu}
\ee
Corresponding to $\bol \mu$, we partition the path  $\mathscr{Z}\in \Upsilon_{N-2}$ as $\mathscr{Z}
=(\mathscr{X}_1^{\bol \mu}, \dots, \mathscr{X}_{k+1}^{\bol \mu})$, where 
\begin{align}
\mathscr{X}_1^{\bol \mu}&=(X_1, \dots, X_{\mu_1-1})\in \Upsilon_{\mu_1-1}, \no
\mathscr{X}_2^{\bol \mu}&=(X_{\mu_1}, \dots, X_{\mu_2-1})\in \Upsilon_{\mu_2-\mu_1}, \no
&\vdots\no
\mathscr{X}_{k+1}^{\bol \mu}&=(X_{\mu_k}, \dots, X_{N-2})\in \Upsilon_{\mu_{k+1}-\mu_k}
\end{align}
 with $\mu_{k+1}=N-2$. 
Here, it should be noted that we have chosen $X_{N-2}$ to be fixed as $X_{N-2}=Y_0$.
For each $i=1, 2,  \dots, k+1$,  let us define
\begin{align}
\mathscr{C}(\mathscr{X}^{\bol \mu}_i, {\bol u}_i)
=
e^{-u_{\mu_i}G_0^{\F}} \mathbb{E}_{X_{\mu_i}}e^{-u_{\mu_i+1}G_0^{\F}} \mathbb{E}_{X_{\mu_i+1}} \cdots 
e^{-u_{\mu_{i+1}-1} G_0^{\F}}\mathbb{E}_{X_{\mu_{i+1}-1}},
\end{align}
where ${\bol u}_i=(u_{\mu_i}, \dots, u_{\mu_{i+1}-1})$.
Inserting $\vepsilon=\vepsilon({\bs \mu})$ into the right-hand side of \eqref{433Below} and utilizing \eqref{DuGLower}, we find that 
\begin{align}
&\mbox{the right-hand side of \eqref{A}}\no
\ge &
\sum_{{\bol \xi}\in \{-\ell, -1\}^k} \sum_{{\bol \sharp }\in \{-, +\}^k}  t^k \int_0^{u_{\mu_1}}ds_1\dots \int_0^{u_{\mu_k}}ds_k
\mathcal{K}_{\mathscr{Z}, n(N-2)} ({\bol \xi},  {\bol \sharp }, {\bol u}, {\bol s}),
\end{align}
where
\begin{align}
\mathcal{K}_{\mathscr{Z}, n(N-2)}({\bol \xi},  {\bol \sharp }, {\bol u}, {\bol s})
=
& \Big\la E^{(q)}_{X_0}\Big|\mathscr{C}(\mathscr{X}^{\bol \mu}_1, {\bol u}_1)J(\xi_1,  \sharp _1, s_1)\mathscr{C}(\mathscr{X}^{\bol \mu}_2, {\bol u}_2)J(\xi_2,  \sharp _2, s_2)\cdots\no
& \cdots \mathscr{C}(\mathscr{X}^{\bol \mu}_k, {\bol u}_k)J(\xi_k,  \sharp _k, s_k)  \mathscr{C}(\mathscr{X}^{\bol \mu}_{k+1}, {\bol u}_{k+1})
 E^{(q')}_{Y_0}\Big\ra.
\end{align}
Furthermore,  we can rewrite $\mathcal{K}_{\mathscr{Z}, n(N-2)} ({\bol \xi},  {\bol \sharp }, {\bol u},  {\bs s})$ as 
\begin{align}
\mathcal{K}_{\mathscr{Z}, n(N-2)}({\bol \xi},  {\bol \sharp }, {\bol u}, {\bol s})=&
\Big|\Big\la e_{X_0} \Big|
C(\mathscr{X}^{\bol \mu}_1, {\bol u}_1)b_{\xi_{1}}^{\sharp _{1}}(s_{1})C(\mathscr{X}^{\bol \mu}_2, {\bol u}_2)b_{\xi_{2}}^{\sharp _{2}}(s_{2})\cdots \no
&\cdots C(\mathscr{X}^{\bol \mu}_k, {\bol u}_k)
b_{\xi_{k}}^{\sharp _{k}}(s_{k})
 C(\mathscr{X}^{\bol \mu}_{k+1}, {\bol u}_{k+1})
e_{Y_0}\Big\ra\Big|^2,
 \label{StrictI}
\end{align}
where, for given  path $\mathscr{X}=(X_i)_{i=1}^m\in \Upsilon_m$ and ${\bs s}=(s_i)_{i=1}^m \in R_{t, m}$, we set
\be
C(\mathscr{X}, {\bol s})
=
e^{-s_1\mathbb{K}^{\F}} E_{X_{1}}e^{-s_{2}\mathbb{K}^{\F}} E_{X_{2}} \cdots 
e^{-s_{m} \mathbb{K}^{\F}}E_{X_{m}}. \label{DefCXM}
\ee

Based on the preceding assertions and employing Lemma \ref{LemmRed1}, we deduce the ensuing outcome:
\begin{Prop}\label{RedPI}
Let $q, q\rq{}\in \mathbb{I}_{\Lambda},$ $X_0\in \Theta_{\Lambda}(q)$ and $ Y_0\in \Theta_{\Lambda}(q\rq{})$
 be given in Lemma \ref{LemmNonZero}. 
Suppose that there exist $0\le \beta<\infty$, $N\in \BbbZ_{\ge 2}, 0\le k\le N-3,$ $\xi_1, \dots, \xi_k\in \{-\ell, -1\}$, $\sharp _1, \dots, \sharp _k\in \{-, +\}, {\bs u} \in R_{\beta, n(N-2)}$ conforming to the form specified in \eqref{Defu}  and $\mathscr{Z}\in \Upsilon_{N-2}$ with $X_{N-2}=Y_0$ such that 
\begin{align}
\Big\la e_{X_0} \Big|&
C(\mathscr{X}^{\bol \mu}_1, {\bol u}_1)b_{\xi_{1}}^{\sharp _{1}}(s_{1})C(\mathscr{X}^{\bol \mu}_2, {\bol u}_2)b_{\xi_{2}}^{\sharp _{2}}(s_{2})\cdots  
b_{\xi_{k}}^{\sharp _{k}}(s_{k})
 C(\mathscr{X}^{\bol \mu}_{k+1}, {\bol u}_{k+1})
e_{Y_0}\Big\ra\neq 0.
\label{KeyEx}
\end{align}
Then $\big\la  \vphi \big|e^{-\beta \tilde{H}^{\F}_{\Lambda}}  \psi \big\ra>0$ holds. Thus, $e^{-\beta \tilde{H}^{\F}_{\Lambda}} $ is ergodic 
w.r.t. $\tCone^{\F}$. Note that, for $N=2$, we understand that the left-hand side of \eqref{KeyEx}  
equals to $
\delta_{X_0, Y_0}\big\la e_{X_0} \big|
e^{-\beta \mathbb{K}^{\F}}
e_{Y_0}\big\ra
$.
\end{Prop}

\subsection{Proof of Theorem \ref{PISemi}}\label{PfThmPI}

In Appendix \ref{PfPuzzle}, we prove the following:

\begin{lemm}\label{Puzzle}
There exist   $0\le \beta<\infty$, $N\in \BbbZ_{\ge 2}, 0\le k\le N-3,$ $\xi_1, \dots, \xi_k\in \{-\ell, -1\}$, $\sharp _1, \dots, \sharp _k\in \{-, +\}, {\bs u} \in R_{\beta, n(N-2)}$ 
in accordance with the form specified in \eqref{Defu}
and $\mathscr{Z}\in \Upsilon_{N-2}$ with $X_{N-2}=Y_0$, such that the validity of equation \eqref{KeyEx} is assured.
\end{lemm}
\begin{proof}[Proof of Theorem \ref{PISemi} assuming Lemma  \ref{Puzzle}] 
\  \smallskip\\
By combining this lemma with Proposition \ref{RedPI}, we derive the outcome presented in Proposition \ref{PIUnitary}. Consequently, we establish that the semigroup  $e^{-\beta H_{\Lambda}^{\F}}$ exhibits ergodicity with respect to $\rCone^{\F}$.
\end{proof}

\section{Long-range orders} \label{PfLRO}

\subsection{Energy inequality}

In this section, we establish the proof of Theorem \ref{LRO} by amalgamating the technique utilized in \cite{DLS} with our approach of order-preserving operator inequalities that has been established in the preceding sections.
The proof of the theorem for fermion-phonon systems is explained in detail below. The proof for fermionic systems can be given in the same way.

It is noteworthy that $\mathbb{W}$ can be represented as follows:
\begin{align}
\mathbb{W}=
-\frac{1}{2} \sum_{i\neq j} W(i-j) (\delta \hn_i-\delta \hn_j)^2+\frac{V|\Lambda|}{4},
\end{align} 
where $V=\sum_{j\in \Lambda} W(j)$.
Here, we have employed the fact that $(\delta \hn_j)^2=1/4\ \forall j\in \Lambda$.

For each $\bol{h}=\{h_j\}_{j\in \Lambda}\in \BbbR^{\Lambda}$, we introduce a modified Coulomb interaction denoted as $\mathbb{W}(\bol{h})$, given by:
\begin{align}
\mathbb{W}(\bol{h})= 
-\frac{1}{4} \sum_{i\neq j} W(i-j) \big\{ (\delta \hn_i-h_i-\delta \hn_j+h_j)^2
+(\delta \hn_i-h_{r(i)}-\delta \hn_j+h_{r(j)})^2\big\}+\frac{V|\Lambda|}{4}.
\end{align}
It is important to recall that we have extended the reflection mapping $r$ to the entire set $\Lambda$. Specifically, we set $r(i)=-i-1$ for {\it all} $i\in \Lambda$.
For every $\boldsymbol{h}\in \mathbb{R}^{\Lambda}$, we introduce a generalized Hamiltonian denoted by $\tilde{H}_{\Lambda}(\boldsymbol{h})$, defined as follows:
\begin{align}
\tilde{H}_{\Lambda}(\bol{h})=\mathbb{T}-\mathbb{W}(\bol{h}).
\end{align} 
Note that $\tilde{H}_{\Lambda}({\bol 0})=\tilde{ H}_{\Lambda} $.

We initiate our discussion with the following abstract lemma:

\begin{lemm}\label{EngInq}
Let $A, B_1, \dots, B_n, C_1, \dots, C_n$ denote self-adjoint operators that act on $\Fock_{\LL}$. Suppose that $B_1, \dots, B_n, C_1, \dots, C_n\in \mathfrak{A}_{\Lambda}$, where $\mathfrak{A}_{\Lambda}$ is given by \eqref{DefALamba}.
 Furthermore, assume that $A$ is positive and $e^{-\beta A} \in \mathfrak{A}_{\Lambda}\, \forall \beta \ge 0$. As a result, ${\bs L}(A)$ and ${\bs R}(A)$ can be defined.
Now, let us consider the Hamiltonian $H_A({\bol B}, {\bol C})$ acting on $\THL$: 
\be
H_A({\bol B}, {\bol C})={\bs L}(A)+{\bs R}(A)-\sum_{i, j=1}^n V_{ij}\{{\bs L}(B_j) {\bs R}(C_j)
+{\bs L}(C_j) {\bs R}(B_j)\},
\ee
where ${\bol B}=\{B_i\}_{i=1}^n$ and ${\bol C}=\{C_i\}_{i=1}^n$.
Let $E_A( {\bol B}, {\bol C}) =\inf \mathrm{spec}(
H_A({\bol B}, {\bol C})
)$. We assume that $(V_{ij})$ is real and  positive semi-definite. 
 For the sake of simplicity, suppose that  $H_A({\bol B}, {\bol C}), H_A({\bol B}, {\bol B})$ and $H_A({\bol C}, {\bol C})$
have ground states.
Under these assumptions, we obtain the inequality:
\be
E_A({\bol B}, {\bol C}) \ge \frac{1}{2}\{E_A( {\bol B}, {\bol B})
+
E_A( {\bol C}, {\bol C})
\}.
\ee
\end{lemm}

\begin{proof} 

To demonstrate the desired assertion, it is sufficient to consider the case where $V_{ij}=\delta_{ij}$. The rationale behind this is as follows:  Let $U$ be the unitary operator such that $V=U^{-1}{\rm diag}(\lambda_1, \dots, \lambda_n) U$, where $\lambda_1, \dots, \lambda_n\ge 0$ are the eigenvalues of the matrix $(V_{ij})$.
 By defining $\tilde{B}_i=\lambda_i^{1/2}\sum_{j=1}^n U_{ij}B_j$ and $\tilde{C}_i=\lambda_i^{1/2}\sum_{j=1}^n U_{ij}C_j$, we find  that 
 \be
 \sum_{i, j=1}^nV_{ij}\{{\bs L}(B_j) {\bs R}(C_j)
+{\bs L}(C_j) {\bs R}(B_j)\}
=\sum_{i=1}^n \{{\bs L}(\tilde{B}_i) {\bs R}(\tilde{C}_i)
+{\bs L}(\tilde{C}_i) {\bs R}(\tilde{B}_i)\}.
 \ee
Consequently, we establish the desired claim.
 
 Let $J$ be the involution associated with $\tCone$, as defined in equation \eqref{tConeJ}. 
We say that a vector $\xi$ in $\THL$ is $J$-{\it real}, if $J\xi=\xi$, that is, $\xi$ is self-adjoint. First, we claim that among the ground states of $H_A( {\bol B}, {\bol C})$, there is one that is $J$-real.
Let  $\xi\in \THL$. Then we can decompose $\xi$ as $\xi=\xi_r+\im \xi_i$ with $\xi_r
=(\xi+J\xi)/2, \xi_i=(\xi-J\xi)/2\im$.
Note that $\xi_r$ and $ \xi_i$ are $J$-real.
 Using this decomposition, we have
\begin{align}
&\la \xi|H_A({\bol B}, {\bol C})\xi\ra\no
=&\la \xi_r|H_A({\bol B}, {\bol C})\xi_r\ra+\la \xi_i|H_A({\bol B}, {\bol C})\xi_i\ra+\im \la \xi_r|H_A({\bol B}, {\bol C})\xi_i\ra-\im \la \xi_i|H_A({\bol B}, {\bol C})\xi_r\ra.
\end{align}
Since $A, B_i, C_i$ are self-adjoint, the terms $
\la \xi_r|H_A({\bol B}, {\bol C})\xi_i\ra
$ and $
\la \xi_i|H_A({\bol B}, {\bol C})\xi_r\ra
$
are real numbers. Therefore, we have  $
\la \xi|H_A({\bol B}, {\bol C})\xi\ra=\la \xi_r|H_A({\bol B}, {\bol C})\xi_r\ra+\la \xi_i|H_A({\bol B}, {\bol C})\xi_i\ra
$ holds, which implies that $
E_A({\bol B}, {\bol C})=\inf_{\xi:\ J\xi=\xi} \la \xi|H_A({\bol B}, {\bol C})\xi\ra
$. Thus, we can conclude that there exists a $J$-real ground state.

Let $\psi$ be a $J$-real ground state of $H_A({\bol B}, {\bol C})$. Assume that $\psi$ is normalized. 
By using the cyclic property  of the trace, we find that 
\be
\la \psi|{\bs L}(A)\psi\ra=\Tr[\psi A\psi]=\la |\psi||{\bs L}(A)|\psi|\ra.
\ee
Similarly, we have
\begin{align}
\la \psi|{\bs R}(A)\psi\ra=\la |\psi||{\bs R}(A)|\psi|\ra.
\end{align}
Let $\psi=U|\psi|$ be the polar decomposition of $\psi$.  By using the Cauchy--Schwarz inequality, we have
\begin{align}
|\la \psi| {\bs L}(B_i) {\bs R}(C_i) \psi\ra|&=|\Tr[\psi B_i \psi C_i ]|\no
& = |
\Tr[|\psi|^{1/2} BU |\psi|^{1/2} |\psi|^{1/2} CU |\psi|^{1/2}]
|\no
& \le (\Tr [B|\psi| B|\psi|])^{1/2} (\Tr[|\psi|C|\psi|C])^{1/2}\no
&\le \frac{1}{2}\Tr[B|\psi| B|\psi|]+\frac{1}{2} \Tr[|\psi|C|\psi|C]\no
&= \frac{1}{2} \la |\psi||{\bs L}(B) {\bs R}(B)|\psi|\ra+\frac{1}{2} \la |\psi||{\bs L}(C) {\bs R}(C)|\psi|\ra.
\end{align}
Applying above facts, we have
\begin{align}
\la \psi|H_A({\bol B}, {\bol C})\psi\ra
\ge& \frac{1}{2} \la |\psi||H_A({\bol B}, {\bol B})|\psi|\ra+\frac{1}{2}\la |\psi||H_A({\bol C}, {\bol C})|\psi|\ra\no
\ge &\frac{1}{2}\{E_A({\bol B}, {\bol B})
+
E_A( {\bol C}, {\bol C})
\}.
\end{align}
In the last inequality, we have utilized the fact that $\||\psi|\|=1$.
\end{proof}
Applying Lemma \ref{EngInq}, we obtain the following:
\begin{Prop}\label{EnergyI}
Let $E({\bol h})=\inf \mathrm{spec}(\tilde{H}_{\Lambda}({\bol h}))$.
For each $\bol{h}\in \BbbR^{\Lambda}$,  we have
$E(\bol{0}) \le E(\bol{h})$.
\end{Prop} 
\begin{proof}
First, observe that $\tilde{H}_{\Lambda}(\bol h)$ can be represented as follows:
\be
\tilde{H}_{\Lambda}({\bol h})={\bs L}(\mathbb{T}_L-\mathbb{W}_L({\bol h})+K_L)+{\bs R}(\mathbb{T}_L-\mathbb{W}_L({\bol h})+K_L)-\mathbb{V}({\bol h}),
\ee
where
\begin{align}
\mathbb{W}_L({\bol h})=&
-\frac{1}{4} \sum_{i,  j\in \Lambda_L} W(i-j) \big\{ (\delta \hn_i-h_i-\delta \hn_j+h_j)^2+
\no&+(\delta \hn_i-h_{r(i)}-\delta \hn_j+h_{r(j)})^2\big\}+\frac{V|\Lambda|}{8}
\end{align}
and 
\begin{align}
\mathbb{V}({\bol h})&=-\mathbb{T}_{LR}+\mathbb{W}_{LR}({\bol h})
\end{align}
with 
\begin{align}
\mathbb{W}_{LR}({\bol h})=&\frac{1}{2} \sum_{i, j\in \Lambda_L} W(i+j+1) \big\{
{\bs L}(\delta\hn_i-h_i) {\bs R}(\delta \hn_j-h_{r(j)})+\no
&+{\bs L}(\delta\hn_i-h_{r(i)} ){\bs R}(\delta \hn_j-h_{j})
\big\}.
\end{align}
For ${\bol h}_L=(h_{L, i})_{i\in \Lambda_L} \in \BbbR^{\Lambda_L}$ and ${\bol h}_R=(h_{R, i})_{i\in \Lambda_R} \in \BbbR^{\LR}$, we define ${\bol h}=({\bol h}_L, {\bol h}_R) \in \BbbR^{\Lambda}$ by 
\be
h_i=\begin {cases}
h_{L, i} & \mbox{if $i\in \Lambda_L$}\\
h_{R, i} & \mbox{if $i \in \Lambda_R$}.
\end{cases}
\ee
By applying Lemma \ref{EngInq}, we obtain 
\be
E({\bol h}_L, {\bol h}_R) \ge \frac{1}{2}E({\bol h}_L, r({\bol h}_L)) +\frac{1}{2}E(r({\bol h}_R), {\bol h}_R)\ge \min\{
E({\bol h}_L, r({\bol h}_L)), E(r({\bol h}_R), {\bol h}_R)
\},
\ee
where $r({\bol h}_L)=(h_{L, r(i)}) \in \BbbR^{\Lambda_R}$ and $r({\bol h}_R)=(h_{R, r(i)}) \in \BbbR^{\Lambda_L}$.
By considering the translation invariance and repeating the aforementioned argument multiple times, we derive the statement presented in the proposition.
\end{proof}

For any $\bol{h}, \bol{h}'\in \ell^2(\Lambda)$, we define the following inner products:
\begin{align}
\la \bol{h}|\bol{h}'\ra_W&=\sum_{i\neq j} W(i-j)
 (h_i^*-h_j^*)(h_i'-h_j'),\ \\ 
\la \bol{\delta \hn}|\bol{h}\ra_W&=\sum_{i\neq j}
 W(i-j)(\delta \hn_i-\delta \hn_j)(h_i-h_j). 
\end{align} 
It is important to note that $\langle \bol{h}|\bol{h}'\rangle_W$ is a complex number, while $\langle \bol{\delta \hn}|\bol{h}\rangle_W$ is a linear operator.

Let $\tilde{\psi}_{\Lambda}$ denote the ground state of $\tilde{H}_{\Lambda}$, and $\langle\!\langle \cdot \rangle\!\rangle_{\Lambda}$ represent the ground state expectation value, defined as $\langle\!\langle A \rangle\!\rangle_{\Lambda}=\langle \tilde{\psi}_{\Lambda}|A\tilde{\psi}_{\Lambda}\rangle$.
\begin{coro}\label{CoroEnergyI}
Let $E= \inf \mathrm{spec}(\tilde{H}_{\Lambda})$.
For each $\bol{h}\in \BbbC^{\Lambda}$,  we have
\be
\Big\la\!\!\Big\la \la \bol{\delta \hn}|\bol{h}\ra_{W} (\tilde{H}_{\Lambda}-E)^{-1} \la
 \bol{\delta \hn}|\bol{h}\ra_{W}\Big\ra\!\!\Big\ra_{\Lambda}
\le \la \bol{h}|\bol{h}\ra_W.\label{CoroEngInq}
\ee
\end{coro} 
\begin{proof} 
Applying Proposition \ref{EnergyI}, we have
$
d^2 E(\lambda \bol{h})\big/d\lambda^2\big|_{\lambda=0} \ge 0.
$
By using second order perturbation theory,
we obtain \eqref{CoroEngInq} for $\bs h$ real-valued. 
To extend this to complex-valued
${\bs h}$\rq{}s,  we observe that, if we express $A = A_{\rm R} + \im A_{\rm I}$ with $A_{\rm R}^*=A_{\rm R}$ and $A_{\rm I}^*=A_{\rm I}$, then we have
\be
\big\la\!\big\la  A^* (\tilde{H}_{\Lambda}-E)^{-1} 
 A\big\ra\!\big\ra_{\Lambda}=\big\la\!\big\la  A_{\rm R} (\tilde{H}_{\Lambda}-E)^{-1} 
 A_{\rm R}\big\ra\!\big\ra_{\Lambda}+\big\la\!\big\la  A_{\rm I}(\tilde{H}_{\Lambda}-E)^{-1} 
 A_{\rm I}\big\ra\!\big\ra_{\Lambda}.
\ee
\end{proof}

\subsection{Infrared bound}
For each $\bol{h}\in \ell^2(\Lambda)$, we define a linear operator $R$ on
$\ell^2(\Lambda)$ as
\begin{align}
(R\bol{h})_j=2 V h_j-2\sum_{i : i\neq j} W(i-j)h_i,
\end{align} 
where $V=\sum_{j\in \Lambda} W(j)$.
Because $\la {\bs h}|R{\bs h}\ra=\la {\bs h}|{\bs h}\ra_W\ge 0$ for all ${\bs h} \in \ell^2(\Lambda)$, we have $R\ge 0$.

\begin{lemm}\label{Infra}
For each $\bol{h}\in \BbbC^{\Lambda}$,  one obtains 
\be
\Big(\Big\la\!\!\Big\la
 \la \bol{\delta \hn}|\bol{h}\ra_{W}^*  \la
 \bol{\delta \hn}|\bol{h}\ra_{W}
\Big\ra\!\!\Big\ra_{\Lambda}\Big)^2
\le \frac{t}{2} \la \bol{h}|\bol{h}\ra_W
\big\la R \bol{h}| \tau \Delta\tau^{-1}(R\bol{h})\big\ra,
\ee
where $\Delta$ is the discrete Laplacian  on $\Lambda$\footnote{Namely,
$\la {\bs h}|\Delta {\bs h}\ra=\sum_{j=-\ell}^{\ell}|h_j-h_{j+1}|^2$  with $h_{\ell}=h_{-\ell}$.}
 and $(\tau {\bs h})_j=(-1)^j h_j$.
\end{lemm} 
\begin{proof}
Let $\eta=\la \bol{\delta \hn}|\bol{h}\ra_W\tilde{\psi}_{\Lambda}$, where $\tilde{\psi}_{\Lambda}$ is the ground state of $\tilde{H}_{\Lambda}$. Let $dE(\lambda)$ be the
spectral measure associated with $\tilde{H}_{\Lambda}-E$.
Then we define $d\mu(\lambda)
=\la \eta|dE(\lambda)\eta\ra
$. By using the Cauchy--Schwarz inequality, we have
\begin{align}
\Big\la\!\!\Big\la \la \bol{\delta \hn}|\bol{h}\ra_W^*\la
 \bol{\delta \hn}|\bol{h}\ra_W\Big\ra\!\!\Big\ra_{\Lambda}
=\int_0^{\infty}d\mu(\lambda)\le
 \Bigg(
\int_0^{\infty}\lambda d\mu(\lambda)
\Bigg)^{1/2}
 \Bigg(
\int_0^{\infty}\lambda^{-1} d\mu(\lambda)
\Bigg)^{1/2}.\label{CS}
\end{align}  
By applying Corollary  \ref{CoroEnergyI}, we have
\begin{align}
\int_0^{\infty}\lambda^{-1} d\mu(\lambda)=\Big\la\!\!\Big\la \la \bol{\delta \hn}|\bol{h}\ra_{W} (\tilde{H}_{\Lambda}-E)^{-1} \la
 \bol{\delta \hn}|\bol{h}\ra_{W}\Big\ra\!\!\Big\ra_{\Lambda}
\le \la \bol{h}|\bol{h}\ra_W.\label{CS2}
\end{align} 
In addition, we have
\begin{align}
\int_0^{\infty}\lambda d\mu(\lambda)
=\frac{1}{2}\Big\la\!\!\Big\la
\Big[
\big[\la\bol{\delta \hn}|\bol{h}\ra_W,  \tilde{H}_{\Lambda}\big], \la  \bol{\delta \hn}| \bol{h}\ra^*_W
\Big]
\Big\ra\!\!\Big\ra_{\Lambda}.\label{DoubleC}
\end{align} 
Since $\la \bol{\delta \hn}|\bol{h}\ra_W=\la \bol{\delta \hn}|R\bol{h}\ra$, we have
\begin{align}
&\mbox{the right-hand side  of (\ref{DoubleC})} \no
=& \frac{t}{2} \sum_{j\in \Lambda_{\rm e}} \sum_{\vepsilon=\pm 1} \big|
(R {\bs h})_{j+\vepsilon}+(R{\bs h})_j
\big|^2 \la\!\la e^{-\im \alpha (\phi_j-\phi_{j+\vepsilon})} c_j^*c_{j+\vepsilon}^*
+e^{+\im \alpha (\phi_j-\phi_{j+\vepsilon})} c_{j+\vepsilon}c_j\ra\!\ra_{\Lambda}\no
\le & t\sum_{j\in \Lambda_{\rm e}} \sum_{\vepsilon=\pm 1} \big|
(R {\bs h})_{j+\vepsilon}+(R{\bs h})_j
\big|^2 \no
=&t \big\la \tau
 R \bol{h}|\Delta(\tau R \bol{h})\big\ra.
\end{align} 
Summarizing (\ref{CS}), (\ref{CS2}) and (\ref{CS3}), we obtain the desired result.
\end{proof}

\begin{lemm}\label{UpperLemm}
Define a quantity denoted as $c_0$ by the expression: $\displaystyle
\lim_{|j|\to \infty} \la\!\la \delta \hn_j \delta \hn_0 \ra\!\ra_{\BbbZ}=(2\pi )^{-1/2} c_0
$, where $\la\!\la A\ra\!\ra_{\BbbZ}$ represents the ground state of the infinite system, which can be precisely defined as follows:
$
\la\!\la A\ra\!\ra_{\BbbZ}=\lim_{\Lambda \in \mathbb{F}_{\rm o}, \Lambda\uparrow \BbbZ} \la\!\la A\ra\!\ra_{\Lambda}
$.\footnote{
It should be noted that for accuracy, we need to consider a convergent subnet from the set $\{\la\!\la \cdot \ra\!\ra_{\Lambda} : \Lambda\in \mathbb{F}_{\rm o}\}$, see Section \ref{SecInfChain}.
} Subsequently, we obtain the following inequality:

\begin{align}
\la\!\la \delta\hn_0^2\ra\!\ra_{\BbbZ}
\le (2\pi)^{-1/2}c_0+(2\pi)^{-1}\int_{\mathbb{T}}dp \sqrt{\frac{F(p)}{\hat{R}(p)}},
\end{align} 
where  $F(p)$ and 
$
\hat{R}(p)
$ are defined in Theorem \ref{LRO}.
\end{lemm} 
\begin{proof} 
We provide a sketch of the  proof. Let $G(x)=\la\!\la \delta \hn_j \delta \hn_0\ra\!\ra_{\BbbZ}$.
By  applying the Fourier transformation and Lemma \ref{Infra}, we
have
\begin{align}
|\hat{G}(p)| \le (2\pi )^{-1/2}\sqrt{\frac{F(p)}{\hat{R}(p)}} \label{CS3}
\end{align} 
for all $p\in \mathbb{T}\backslash \{0\}$.
Therefore, $\hat{G}$ can be expressed as 
\begin{align}
\hat{G}(p) = c\delta(p)+I(p),\ \ p\in \mathbb{T},
\end{align} 
where $I(p)$ satisfies $0\le I(p)
\le (2\pi )^{-1/2} \sqrt{\frac{F(p)}{\hat{R}(p)}} 
$.  
Hence, we have
\begin{align}
\la\!\la\delta\hn_0^2\ra\!\ra_{\BbbZ}=(2\pi)^{-1/2} \int_{\mathbb{T}} dp \hat{G}(p)
\le (2\pi )^{-1/2}c+(2\pi )^{-1} \int_{\mathbb{T}} dp \sqrt{\frac{F(p)}{\hat{R}(p)}}.
\end{align} 
Moreover, because 
\begin{align}
\la\!\la \delta \hn_j \delta \hn_0\ra\!\ra_{\BbbZ}&=(2\pi )^{-1/2} \int_{\mathbb{T}} dp \hat{G}(p)e^{-\im jp}\no
&=(2\pi)^{-1/2} c+(2\pi)^{-1/2} \int_{\mathbb{T}} dp I(p) e^{-\im jp},
\end{align} 
we can deduce that $\displaystyle
\lim_{|j|\to \infty} \la\!\la \delta \hn_j\delta \hn_0\ra\!\ra_{\BbbZ}=(2\pi)^{-1/2} c
$.  This is a direct consequence of the fact that $\displaystyle 
\lim_{|j|\to \infty} \int_{\mathbb{T}} dp I(p) e^{-\im jp}=0
$, which follows   from the Riemann--Lebesgue lemma. 
\end{proof}

\subsection{Proof of Theorem \ref{LRO}}
By recalling the fact that $n_i^2 = n_i$ holds for all $i \in \Lambda$, we can easily verify that  $(\delta \hn_0)^2=1/4$.
Combining this with  Lemma \ref{UpperLemm}, we arrive at
\be
(2\pi)^{-1/2} c_0 \ge  \frac{1}{4}  -(2\pi)^{-1} \int_{\mathbb{T}} dp \sqrt{\frac{F(p)}{\hat{R}(p)}}.\label{c0Lower}
\ee
Consequently, if we choose a value for $t$ such that the right-hand side of the inequality is strictly positive, then $c_0$ must be strictly positive as well. By utilizing the fact that $(-1)^{i-j} \la \delta \hn_i \delta \hn_j\ra_{\BbbZ}=\la\!\la \delta \hn_i \delta \hn_j\ra\!\ra_{\BbbZ}$, we ultimately obtain the desired assertion stated in Theorem \ref{LRO}.
\qed

\appendix

\section{Proof of Lemma \ref{Puzzle}} \label{PfPuzzle}

\subsection{Preliminary analysis}
To establish the validity of Lemma \ref{Puzzle}, we require a technical lemma.

For each  $X \in \Theta_{\Lambda}(q)$, we define 
$
|X\ra=e_X\in \Fock_{\LL}^{\F}(q),
$
where $e_X$ is determined by \eqref{DefEx}.
Furthermore, we assign  $|\varnothing\ra=\Omega_{\LL}^{\rm F} $.

For a given path  $\mathscr{X}=(X_i)_{i=1}^m \in \Upsilon_m$,
we define
\be
C_{\tau}(\mathscr{X})
=e^{-\tau \mathbb{K}^{\F}}E_{X_1} e^{-\tau \mathbb{K}^{\F}} E_{X_2} \cdots e^{-\tau \mathbb{K}^{\F}} E_{X_m}.
 \label{DefCPath}
\ee

\begin{lemm}\label{Conn2}
For each  $X\in \Theta_{\Lambda}(q)$ with $X\neq \varnothing $, there exist 
$k\in \BbbZ_+, m_1, \dots, m_{k+1}\in \BbbZ_+$, $N\in \BbbZ_+$,  $\mathscr{X}_1\in \Upsilon_{m_1}, \dots, \mathscr{X}_{k+1} \in \Upsilon_{m_{k+1}}$,
 $\xi_1, \dots, \xi_k\in \{-\ell, -1\}$ and  $\sharp _1,  \dots, \sharp _k\in \{-, +\}$ such  that 
 \begin{align}
 C_{\tau}(\mathscr{X}_1)b_{\xi_1}^{\sharp _1}C_{\tau}(\mathscr{X}_2)b_{\xi_2}^{\sharp _2}\cdots b_{\xi_k}^{\sharp _k} C_{\tau}(\mathscr{X}_{k+1}) \ket{X} &=\{\pm \tau^{N}+O(\tau^{N+1})\} \ket{\varnothing }
 \end{align}
 holds true, 
 provided that  $\tau$ is sufficiently small.  
\end{lemm}
\begin{proof}
Without loss of generality, we can assume that the hopping amplitude is set to $t=1$.
Let $X=(I_e, I_o)\in \Theta_{\Lambda}(q)$. We define a subset ${\sf C}$ of $I_e\cup I_o$ as a {\it cluster} in $I_e\cup I_o$ if there exist $j, r\in \BbbZ_+$ such that ${\sf C}=\{j, j+1, \dots, j+r\}$ and $\mathrm{dist}(I_e\cup I_o\setminus {\sf C}; {\sf C})\ge 2$. Here, $\mathrm{dist}(A; B)$ denotes the minimum distance between the subsets $A$ and $B$ in $\LL$, given by $\mathrm{dist}(A; B)=\min\{|i-j| : i\in A, j\in B\}$.
Consequently, we can decompose $I_e\cup I_o$ as
\begin{align}
I_e\cup I_o=\bigcup_{k=1}^K {\sf C}_k,\ \ {\sf C}_k\cap {\sf C}_{k'} =\varnothing\ (k\neq k'), \label{DecK}
\end{align}
where each ${\sf C}_k$ represents a cluster within $I_e\cup I_o$, as illustrated in Figure \ref{ClusterPic}. In the subsequent analysis, we shall identify $X=(I_e, I_o)\in \Theta_{\Lambda}(q)$ with $I_e\cup I_o$, unless any ambiguity arises.

\begin{figure}
\begin{center}
\includegraphics[scale=0.6]{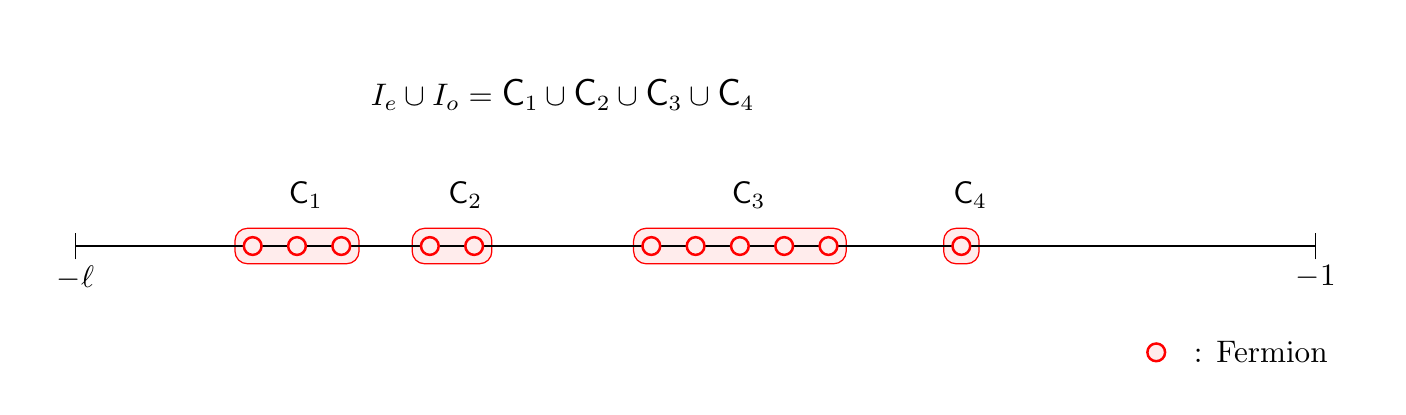}
\caption{Every configuration $I_e\cup I_o$ can be partitioned into clusters.} 
    \label{ClusterPic}
\end{center}
\end{figure}

\subsubsection*{\it Step 1: Proof for  the case where $K=1$}
 We consider the case where $I_e\cup I_o$ consists of a single cluster ${\sf C}=\{j, j+1, \dots, j+r\}$.

\subsubsection*{{\bf A.} \rm Suppose that $|{\sf C}|$ is even. i.e., $r$ is odd.}

Let $p_e=(r-1)/2$.  We define $\mathscr{Y}^{e, -}=(Y_0^{e, -}, Y_1^{e, -},\dots, Y_{p_e}^{e, -})\in \Upsilon_{p_e+1}$
as follows: 
\be
Y^{e, -}_i=Y_{i-1}^{e, -}\setminus \{j+2i, j+2i+1\},
\ee
where  $Y_0^{e, -}={\sf C}\setminus \{j, j+1\}$. It is worth noting that $Y_{p_e}^{e, -}=\varnothing$.
  The construction of $\mathscr{Y}^{e, -}$ is illustrated in Figure \ref{Proc1}.
\begin{figure}[h]
\begin{center}
\includegraphics[scale=0.56]{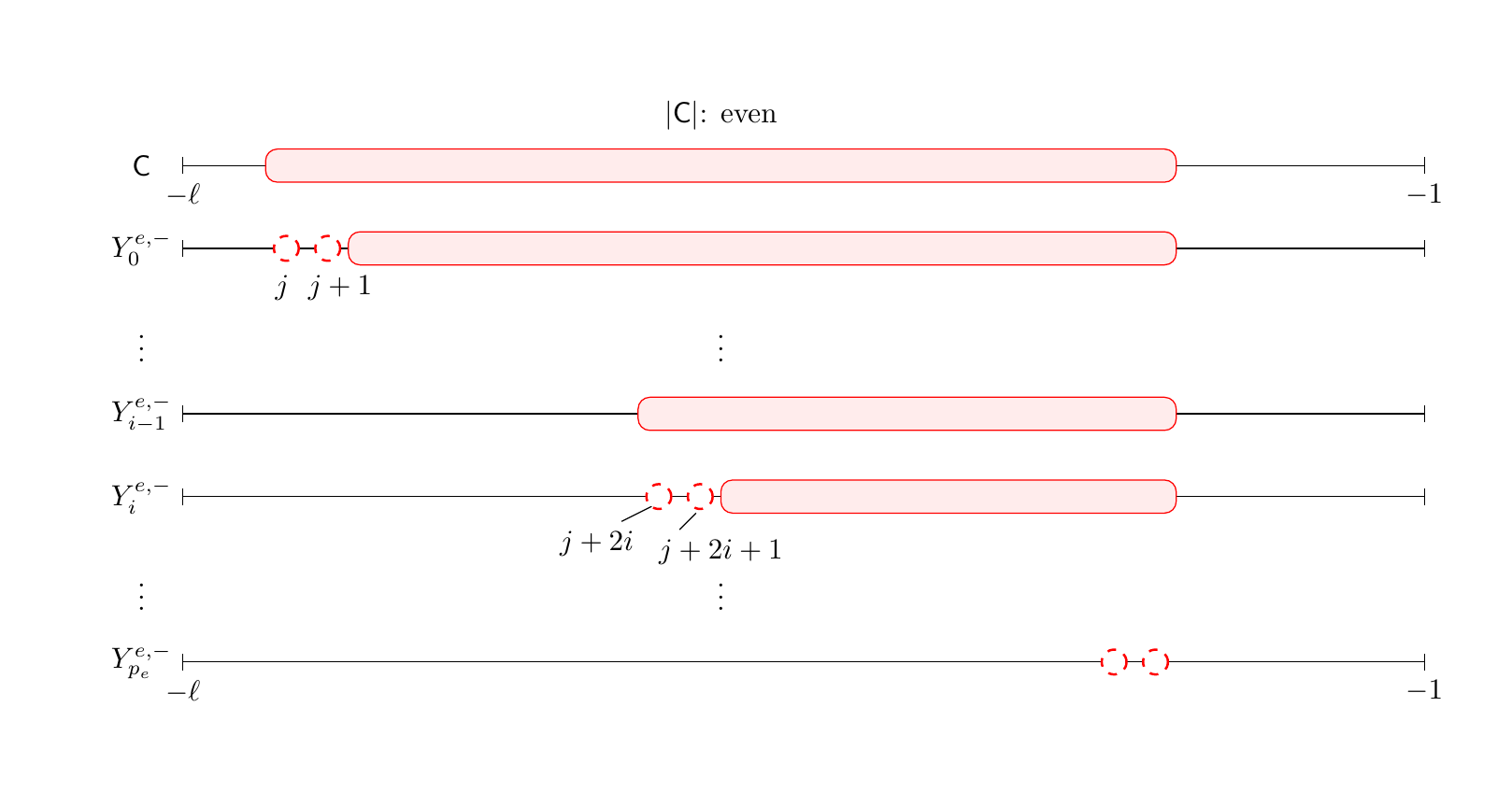}
\caption{
The fermions occupying ${\sf C}$ are successively annihilated in pairs.  This process is repeated iteratively until no fermions remain.} 
    \label{Proc1}
\end{center}
\end{figure}

The concept of proving the subsequent sublemma is sporadically utilized in the remaining process of proving Lemma \ref{Conn2}.
\begin{Subl}\label{FirstSub}
We have
\be
E_{Y_0^{e, -}} e^{-\tau \mathbb{K}^{\F}} |{\sf C}\ra
=\{\pm \tau+O(\tau^2)\} |Y_0^{e, -}\ra\ \ (\tau\to +0). \label{Overbasic}
\ee
\end{Subl}
\begin{proof}
Using Taylor's theorem,  we have 
\be
e^{-\tau \mathbb{K}^{\F}}=1+ (-\tau \mathbb{K}^{\F}) +I(\tau),\label{Dyson00}
\ee
where $I(\tau)$ represents a bounded operator with the property $\|I(\tau)\| = O(\tau^2)$ as $\tau \to +0$.
We observe that 
\be
E_{Y_0^{e, -}}|{\sf C}\ra=0, \quad E_{Y_0^{e, -}} (-\tau \mathbb{K}^{\F})  |{\sf C}\ra
=\pm \tau  |Y_0^{e, -}\ra.\label{E02}
\ee
Combining these with \eqref{Dyson00},  we obtain the desired result.
\end{proof}

  By employing reasoning analogous to that utilized in the proof of Sublemma \ref{FirstSub}, we can establish the following:
\begin{align}
E_{Y_1^{e, -}} e^{-\tau \mathbb{K}^{\F}}E_{Y_0^{e, -}} e^{-\tau \mathbb{K}^{\F}} |{\sf C}\ra
&=E_{Y_1^{e, -}} e^{-\tau \mathbb{K}^{\F}}\{\pm \tau+O(\tau^2)\} |Y_0^{e, -}\ra\no
&=\{\pm \tau^2+O(\tau^3)\}|Y_1^{e, -}\ra.
\end{align}
By iteratively applying this procedure, we arrive at 
\be
\overline{C}_{\tau}(\mathscr{Y}^{e, -})|{\sf C} \ra=\{\pm \tau^{p_e+1}+O(\tau^{p_e+2})\}|\varnothing\ra, \label{Prot}
\ee
where 
\be
\overline{C}_{\tau}(\mathscr{Y}^{e, -})=E_{Y_{p_e}^{e, -}} e^{-\tau \mathbb{K}^{\F}}E_{Y_{p_e-1}^{e, -}}e^{-\tau \mathbb{K}^{\F}}\cdots E_{Y_0^{e, -}} e^{-\tau \mathbb{K}^{\F}}. \label{BarC}
\ee
It is worth noting that $\overline{C}_{\tau}(
\mathscr{Y}^{e, -}
)=(C_{\tau}(\mathscr{Y}^{e, -}))^*$, where $C_{\tau}(\mathscr{Y}^{e, -})$ is defined  by \eqref{DefCPath}
with $\mathscr{X}=\mathscr{Y}^{e, -}$.
Furthermore, we have $
(C_{\tau}(\mathscr{Y}^{e, -}))^*=C_{\tau}(\overline{\mathscr{Y}}^{e, -}) 
$, where $
\overline{\mathscr{Y}}^{e, -}\in \Upsilon_{p_e+1}
$ represents the reversed  path of $\mathscr{Y}^{e, -}$ and is defined as  
$
\overline{\mathscr{Y}}^{e, -}=(  Y_{p_e}^{e, -}, Y_{p_e-1}^{e, -} ,\dots,Y_1^{e, -}, Y_0^{e, -})
$. 
Consequently, we obtain
\be
\overline{C}_{\tau}(
\mathscr{Y}^{e, -}
)=C_{\tau}(\overline{\mathscr{Y}}^{e, -}). \label{BarCY}
\ee
By combining this equation with \eqref{Prot}, we derive the statement presented in the lemma.
\subsubsection*{{\bf B.}  {\rm Suppose that  $|\sf C|$ is odd, i.e., $r$ is even. }}
In this scenario, additional efforts are required.
Let us set $p_o=r/2-1$. We define $\mathscr{Y}^{o, -}=(Y_1^{o, -}, Y_2^{o, -},\dots, Y_{p_o+1}^{o, -})\in \Upsilon_{p_o+1}$ as follows:
\be
Y_i^{o, -}=Y_{i-1}^{o, -}\setminus \{j+2i, j+2i+1\}
\ee
with $Y_1^{o, -}={\sf C} \setminus \{j, j+1\}$. It should be noted that $Y_{p_o+1}^{o, -}=\{j +r\}$.
Figure \ref{Proc2} illustrates the procedures involved.
\begin{figure}[h]
\begin{center}
\includegraphics[scale=0.56]{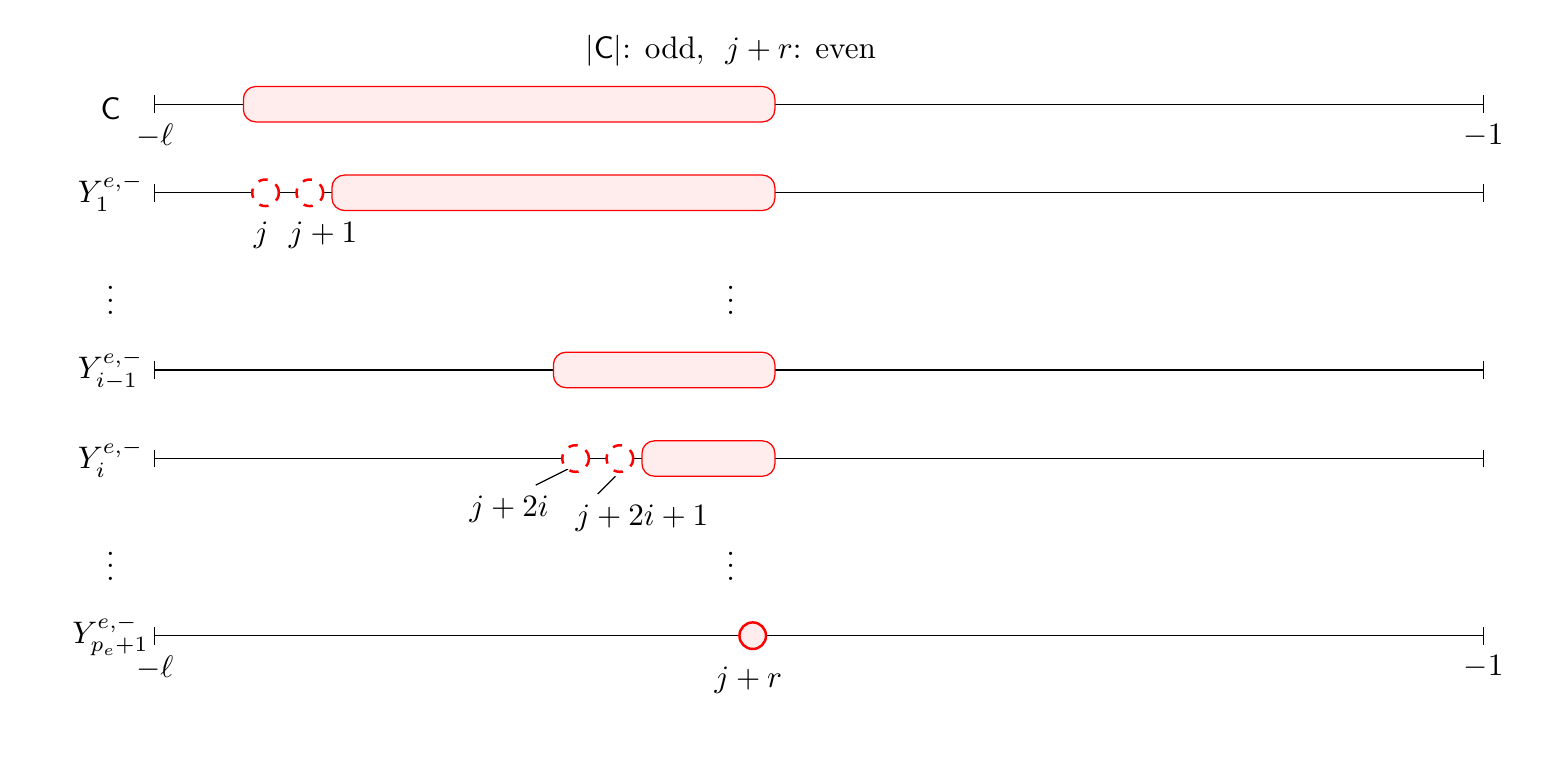}
\caption{The fermions present in $\sf C$ are eliminated in pairs. This process is iterated until only a single fermion remains at the site $j+r$.} 
    \label{Proc2}
\end{center}
\end{figure}
By employing reasoning similar to that utilized in proving \eqref{Prot}, we obtain
\begin{align}
\overline{C}_{\tau}(\mathscr{Y}^{o, -})|{\sf C}\ra=\{\pm \tau^{p_o+1}+O(\tau^{p_o+2})\}\big|\{j+r\}\big\ra.
 \end{align}
 Now, we will consider the following two scenarios:

 \subsubsection*{  \rm \underline{Case 1.}  Suppose that $j+r$ is even. }
 Initially, it should be noted that since $j+r\in \Lambda_L$, the value of $j+r$ is negative.
  We shall proceed as follows.
 \begin{description}
 \item[\rm 1 - 1.]
 Let us set  $q_e=|j+r|/2-2$.
 We define  $\mathscr{Z}^{e, +}=(Z_1^{e, +}, \dots, Z_{q_e}^{e, +})\in \Upsilon_{q_e}$ as
 \be
 Z_i^{e, +}=Z_{i-1}^{e, +} \cup \{j+r+2i, j+r+2i+1\}
 \ee
 with $Z_1^{e, +}=\{j+r, j+r+1, j+r+2, -1\}$.  It should be noted that $Z_{q_e}^{e, +}=\{j+r, j+r+1, \dots, -2, -1\}$.  Refer to Figure \ref{Proc3} for a visual representation.
 \begin{figure}
 \begin{center}
\includegraphics[scale=0.56]{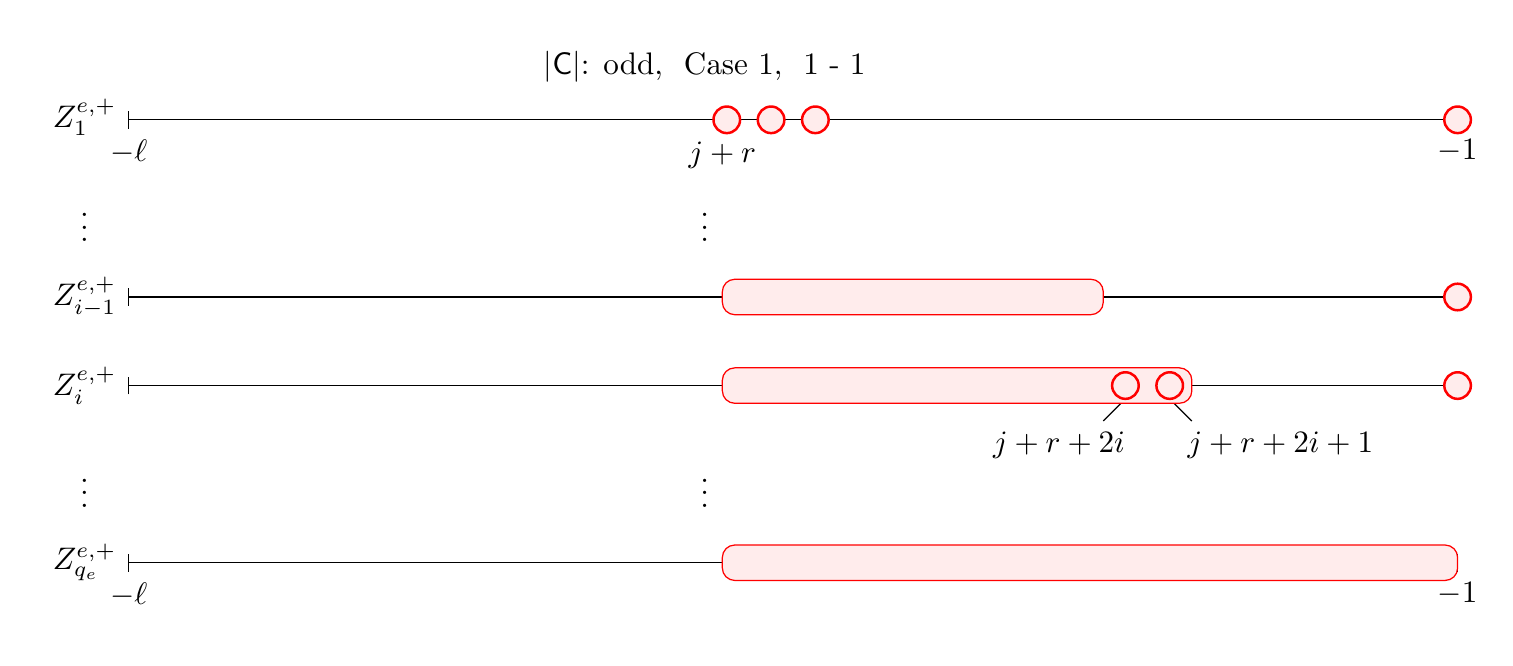}
\caption{Fermions are generated in pairs. This process is iterated until the sites ${j+r, j+r+1, \dots, -1}$ are occupied. } 
    \label{Proc3}
\end{center}
 \end{figure}
 
Subsequently, employing reasoning analogous to the proof of \eqref{Prot}, we obtain 
 \begin{align}
 \overline{C}_{\tau}(\mathscr{Z}^{e, +})b^*_{-1}|\{j+r\} \ra&= \overline{C}_{\tau}(\mathscr{Z}^{e, +})\big|\{j+r, -1\}\big\ra\no
 &=\{\pm \tau^{q_e}+O(\tau^{q_e+1})\}|Z_{q_e}^{e, +}\ra.
 \end{align}

 \item[\rm 1 - 2.]
 Let us define  $\mathscr{Z}^{e, -}=(Z_1^{e, -}, Z_2^{e, -}, \dots, Z^{e, -}_{q_e+1})\in \Upsilon_{q_e+1}$ as
 \be
 Z^{e, -}_i=Z_{i-1}^{e, -}\setminus \{j+r+2i-2, j+r+2i-1\}
 \ee
 with $Z_1^{e, -}=Z_{q_e}^{e, +}\setminus \{j+r, j+r+1\}$.
 Note that $Z_{q_e+1}^{e, -}=\varnothing$.  Refer to Figure \ref{Proc4} for a visual representation.
 \begin{figure}
 \begin{center}
\includegraphics[scale=0.56]{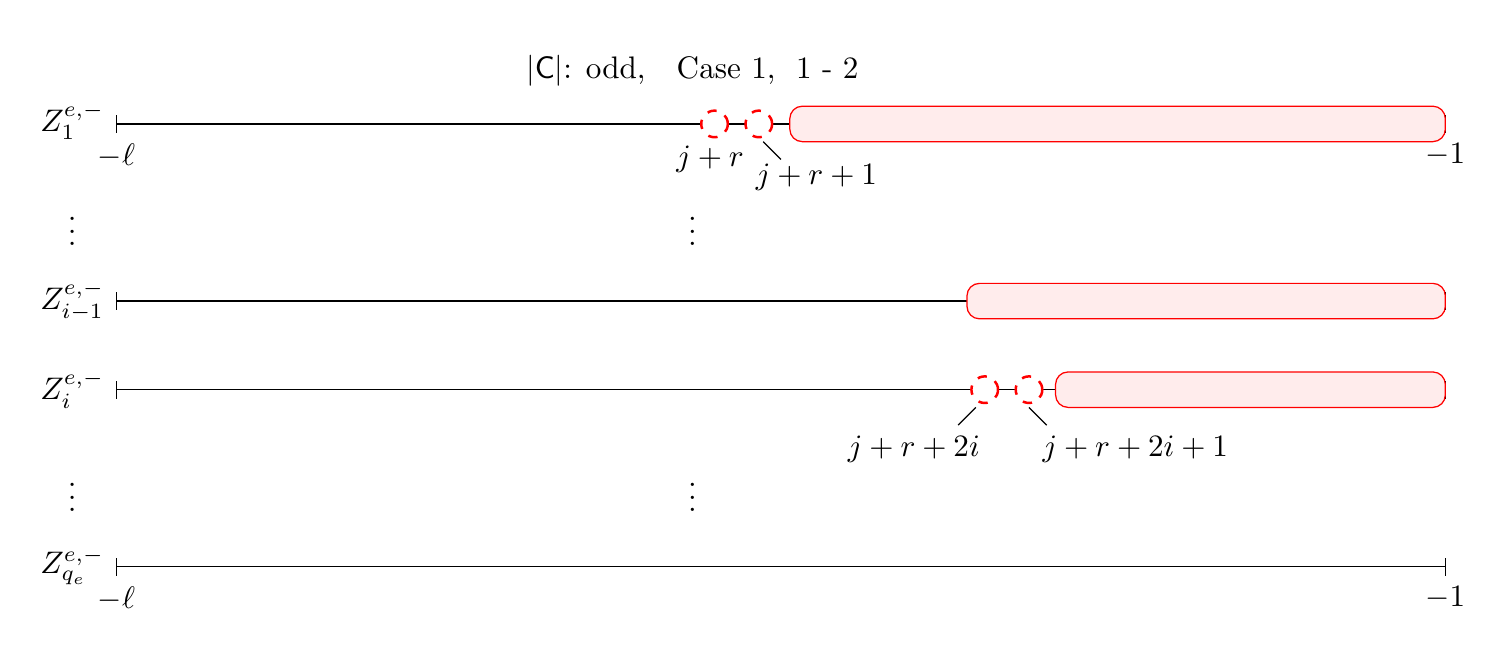}
\caption{
The fermions are successively eliminated in pairs until no fermion remains.
} 
    \label{Proc4}
\end{center}
 \end{figure}
 Then we have
 \be
  \overline{C}_{\tau}(\mathscr{Z}^{e, -})|Z_{q_e}^{e, +}\ra=\{
 \pm \tau^{q_e+1}+O(\tau^{q_e+2})
 \}|\varnothing\ra.
 \ee
 \end{description}
 Let us denote
 $ \overline{C}_{\tau}(\mathscr{Z}^{e, +}\cup \mathscr{Z}^{e, -})=
  \overline{C}_{\tau}(\mathscr{Z}^{e, -}) \overline{C}_{\tau}(\mathscr{Z}^{e, +}) 
 $.
 Summarizing the above observations, we arrive at the following:
 \begin{align}
 \overline{C}_{\tau}(\mathscr{Z}^{e, +}\cup \mathscr{Z}^{e, -})b_{-1}^*  \overline{C}_{\tau}(\mathscr{Y}^{o, -})|{\sf C}\ra
 =\{
\pm  \tau^{p_o+2q_e+2}+O(\tau^{p_o+2q_e+3})
 \}|\varnothing\ra.
 \end{align}
Hence, we obtain the desired statement in the lemma for this case.
 \subsubsection*{\rm \underline{Case 2.} 
Suppose that $j+r$ is odd.}
\begin{description}
 \item[\rm 2 - 1.] 
We define $q_o=(|j+r|-3)/2$. Let us define  $\mathscr{Z}^{o, +}=(Z_1^{o, +}, \dots, Z_{q_o}^{o, +})\in \Upsilon_{q_o}$by the following expression:
 \be
 Z_i^{o, +}=Z_{i-1}^{o, +} \cup \{j+r+2i-1, j+r+2i\}
 \ee
 with the initial condition $Z_1^{o, +}=\{j+r, j+r+1, j+r+2\}$.  It is worth noting that  $Z_{q_e}^{o, +}=\{j+r, j+r+1, \dots, -2, -1\}$.   Figure \ref{Proc5} illustrates the definition of $Z_i^{o, +}$.
\begin{figure}
 \begin{center}
\includegraphics[scale=0.56]{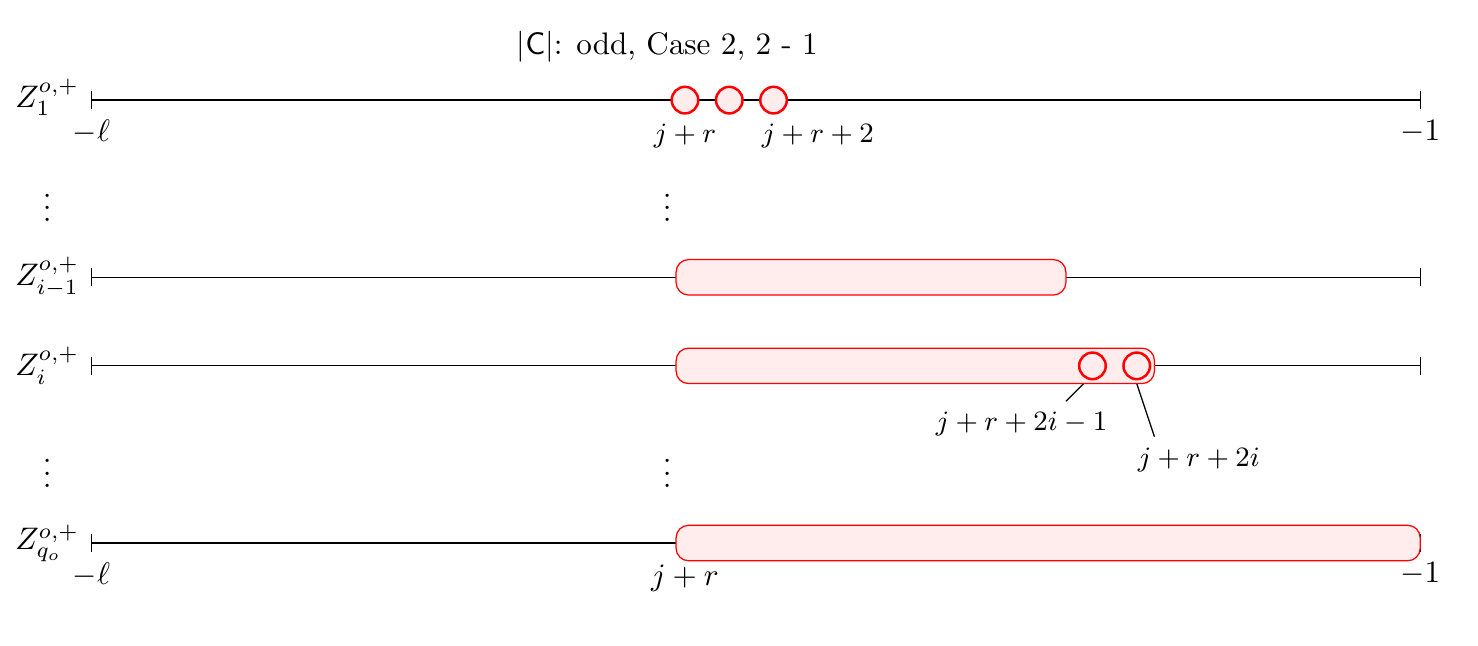}
\caption{
Fermions are successively generated in pairs until all the sites $\{j+r, j+r+1, \dots, -1\}$ are occupied.} 
    \label{Proc5}
\end{center}
 \end{figure}
 Then, by arguments similar to those of the proof of \eqref{Prot}, we have
 \begin{align}
  \overline{C}_{\tau}(\mathscr{Z}^{o, +})\big|\{j+r\}\big\ra=
 \{\pm \tau^{q_o}+O(\tau^{q_o+1})\}\big|Z_{q_o}^{o, +} \big\ra.
 \end{align}
 \item[\rm 2 - 2.]
 Let us define $\mathscr{Z}^{o, -}=(Z_1^{o, -}, \dots, Z_{q_o}^{o, -}) \in \Upsilon_{q_o}$ as  follow:
 \be
 Z_i^{o, -}=Z_{i-1}^{o, -}\setminus \{j+r+2i-2, j+r+2i-1\}
 \ee
 with $Z_1^{o, -}=Z_{q_o}^{o, +}\setminus \{j+r, j+r+1\}$, as depicted in Figure \ref{Proc6}.
 \begin{figure}
 \begin{center}
\includegraphics[scale=0.56]{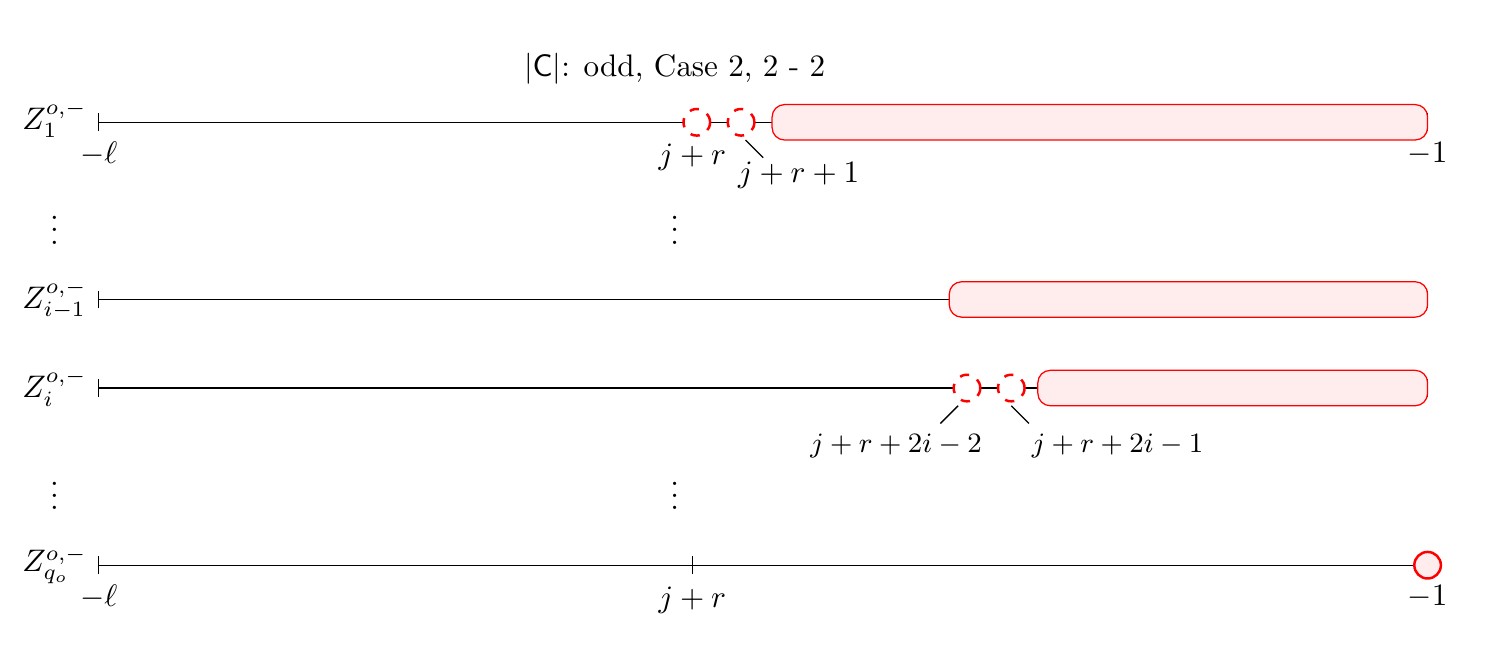}
\caption{
The fermions undergo successive annihilation in pairs until only a single fermion remains at the site $-1$.} 
    \label{Proc6}
\end{center}
 \end{figure}
 Because $Z_{q_o}^{o, -}=\{-1\}$,
    we have
 \begin{align}
  b_{-1}\overline{C}_{\tau}(\mathscr{Z}^{o, -})\big|Z_{q_o}^{o, +}\big\ra&
 =b_{-1}\{\pm \tau^{q_o}+O(\tau^{q_o+1})\} \big|Z_{q_o}^{o, -} \big\ra=\{\pm \tau^{q_o}+O(\tau^{q_o+1})\} |\varnothing\ra.
 \end{align}

 \end{description}
 Let us denote 
 $ \overline{C}_{\tau}(\mathscr{Y}^{o, -}\cup \mathscr{Z}^{o, +}\cup \mathscr{Z}^{o, -})=
  \overline{C}_{\tau}(\mathscr{Z}^{o, -})\overline{C}_{\tau}(\mathscr{Z}^{o, +}) \overline{C}_{\tau}(\mathscr{Y}^{o, -}) 
 $.
 Summarizing the aforementioned observations, we can deduce the following relation:
 \begin{align}
 b_{-1} \overline{C}_{\tau}(\mathscr{Y}^{o, -}\cup\mathscr{Z}^{o, +} \cup \mathscr{Z}^{o, -})|{\sf C}\ra
 =\{
\pm  \tau^{p_o+2q_o+1}+O(\tau^{p_o+2q_o+2})
 \}|\varnothing\ra.
 \end{align}
 By employing \eqref{BarCY}, we can thus obtain the desired result in this particular case.

\subsubsection*{ \it Step 2:  Proof for the general case}
Specifically, we assume that $I_e\cup I_o$ is decomposed as (\ref{DecK}). For each cluster ${\sf C}=\{j, j+1, \dots, j+r\}$, we introduce a linear operator $\mathscr{D}_{\sf C}$ defined as follows:
\be
\mathscr{D}_{\sf C}
=\begin{cases}
 \overline{C}_{\tau}(\mathscr{Y}^{e, -}),  & \mbox{if $|\sf C|$ is even}\\
 \overline{C}_{\tau}(\mathscr{Z}^{e, +}\cup \mathscr{Z}^{e, -})b_{-1}^*  \overline{C}_{\tau}(\mathscr{Y}^{o, -}), 
& \mbox{if $|\sf C|$ is odd and $j+r$ is even}\\
 b_{-1} \overline{C}_{\tau}(
 \mathscr{Y}^{o, -}\cup\mathscr{Z}^{o, +} \cup \mathscr{Z}^{o, -}), & \mbox{if $|\sf C|$ is odd and $j+r$ is odd},
\end{cases}
\ee
where the notations used in {\it Step 1} are employed.

 According to {\it Step 1}, there exists an  $N\in \BbbN$  such that 
\begin{align}
\mathscr{D}_{{\sf C}_1}\cdots \mathscr{D}_{{\sf C}_K} |X\ra=\{\pm \tau^N+O(\tau^{N+1})\} |\varnothing\ra.
\end{align}
Combining this with \eqref{BarCY}, we obtain the assertion in the lemma.
 \end{proof}

\subsection{Proof of Lemma \ref{Puzzle} }
If $X_0=Y_0=\varnothing$, then we have 
$
\la\varnothing| e^{-\beta \mathbb{K}^{\F}} |\varnothing \ra>0
$ for every $\beta \ge 0$. Consequently, we obtain  the desired result in Lemma \ref{Puzzle} with $N=2$.

Subsequently, we will consider the case where $X_0\neq \varnothing$ and $Y_0\neq \varnothing$.
By applying Lemma \ref{Conn2}, we can choose  $k\in \BbbZ_+$, 
$\mathscr{X}_1, \dots, \mathscr{X}_{k+1}$,
 $\xi_1, \dots, \xi_k\in \{-\ell, -1\}$, $\sharp _1, \dots, \sharp _k\in \{-, +\}$ and $N\in \BbbZ_+$
such that 
\be
C_{\tau}(\mathscr{X}_1)b_{\xi_1}^{\sharp _1} C_{\tau}(\mathscr{X}_2)b_{\xi_2}^{\sharp _2} \cdots b_{\xi_k}^{\sharp _k} C_{\tau}(\mathscr{X}_{k+1})
|{X_0}\ra
=\{\pm \tau^{N}+O(\tau^{N+1})\} |\varnothing\ra.
\ee
 Similarly, there exist $m\in \BbbZ_+$, 
$\mathscr{X}_{-1}, \dots, \mathscr{X}_{-m-1}$,
 $\xi_{-1}, \dots, \xi_{-m}\in \{-\ell, -1\}$, $\sharp _{-1},  \dots, \sharp _{-m}\in \{-, +\}$ and $M\in \BbbZ_+$
 such that 
\be
C_{\tau}(\mathscr{X}_{-1})b_{\xi_{-1}}^{\sharp _{-1}} C_{\tau}(\mathscr{X}_{-2})b_{\xi_{-2}}^{\sharp _{-2}} \cdots b_{\xi_{-m}}^{\sharp _{-m}} C_{\tau}(\mathscr{X}_{-m-1})
|{Y_0}\ra
=\{\pm \tau^{M}+O(\tau^{M+1})\} |\varnothing\ra.
\ee
Hence, we obtain
\begin{align}
&\big\la {Y_0}\big|C_{\tau}(\mathscr{X}_{-m-1}) b_{\xi_{-m}}^{\sharp _{-m}} \cdots b_{\xi_{-1}}^{\sharp _{-1}} C_{\tau}(\mathscr{X}_{-1})    e^{-\vepsilon \mathbb{K}^{\F}} C_{\tau}(\mathscr{X}_1)b_{\xi_1}^{\sharp _1} \cdots b_{\xi_k}^{\sharp _k} C_{\tau}(\mathscr{X}_{k+1}) \big|{X_0}\big\ra\no
=& \{\tau^{M+N}+O(\tau^{M+N})\}\la\varnothing| e^{-\vepsilon \mathbb{K}^\F} |\varnothing\ra, \label{InnerPrd}
\end{align}
provided that $\tau$ is sufficiently small.  Because 
 $\la\varnothing| e^{-\vepsilon \mathbb{K}^{\F}} |\varnothing\ra$ is strictly positive, 
we can conclude that the right-hand side of \eqref{InnerPrd} is non-zero, given that $\tau$ is sufficiently small.  Recall the definition of $C(\mathscr{X}; {\bs u})$ given in \eqref{DefCXM}.
By noting that  $C(\mathscr{X}_1; \tau+\vepsilon, \tau, \dots, \tau)=e^{-\vepsilon \mathbb{K}^{\F}} C_{\tau}(\mathscr{X}_1)$, we can establish the desired assertion in Lemma \ref{Puzzle} by choosing 
$\tau=(\beta-\vepsilon)/(N-2)$. Note  that 
this corresponds to the following choice of ${\bs u}\in R_{\beta, n(N-2)}$:
$$
u_1=\beta_*, \dots, u_{k-1}=\beta_*, u_k=\beta_*+\vepsilon, u_{k+1}=\beta_*, \dots, u_{N-2}=\beta_*,
$$
 where   $\beta_*=(\beta-\vepsilon)/(N-2)$ and   ${\bs u}$ is given by  \eqref{Defu}. Therefore,   by   first  fixing  $\vepsilon>0$ to be sufficiently small,  we can  then  choose   $\beta(>\vepsilon)$ such that $\beta_*^{M+N}+O(\beta_*^{M+N})>0$.  With this particular choice of $\vepsilon$ and $\beta$,   the right hand side of \eqref{InnerPrd} 
  becomes strictly positive.
 
 By employing similar arguments as in the previous cases, we can establish Lemma \ref{Puzzle} for the situation where either $X_0=\varnothing$ or $Y_0=\varnothing$.
\qed

\bibliographystyle{abbrvurl}

\end{document}